\documentclass[11pt]{article}
\usepackage[margin=1in]{geometry}
\usepackage{amsfonts}
\usepackage{amsmath,amsthm,amssymb} 
\usepackage{latexsym}
\usepackage{epic}
\usepackage{epsfig}
\usepackage{hyperref}
\makeatletter
    \newcommand{\customlabel}[2]{%
    \protected@write \@auxout {}{\string \newlabel {#1}{{#2}{\thepage}{#2}{#1}{}} }%
    \hypertarget{#1}{}
    }
\makeatother
\usepackage{braket}
\usepackage{verbatim}
\usepackage[justification=centering]{caption}
\usepackage{enumitem}
\usepackage{tikz}
\usetikzlibrary{positioning}
\usetikzlibrary{cd}
\usepackage{color}
\allowdisplaybreaks[1]
\usepackage[numbers]{natbib}
\usepackage{comment}

\usepackage{algorithm}
\usepackage{algorithmic}
\usepackage{setspace}
\usepackage{booktabs}
\usepackage{tcolorbox}

\usepackage{pgfplots}
\pgfplotsset{compat=1.18}

\newcommand{\M}{\textbf{M}}
\newcommand{\CM}{\mathcal{M}}
\newcommand{\N}{\textbf{N}}
\newcommand{\U}{\textbf{U}}
\newcommand{\T}{\textbf{T}}
\newcommand{\Free}{\textbf{Fr}}
\newcommand{\E}{\mathbb {E}}

\newcommand{\D}{\mathcal {D}}
\newcommand{\FF}{\mathbb{F}}
\newcommand{\opt}{\textbf{OPT}}
\newcommand{\eps}{\varepsilon}

\newcommand{\Aut}{\text{Aut}}
\newcommand{\abs}[1]{\vert #1 \vert}
\newcommand{\BigM}{\textbf{BigM}}

\newcommand{\C}{\mathcal{C}}

\newcommand{\rank}{\textsf{rank}}

\renewcommand{\span}{\textsf{span}}
\newcommand{\F}{\mathbb{F}}
\newcommand{\Q}{\mathbb{Q}}
\newcommand{\R}{\mathbb{R}}
\newcommand{\Z}{\mathbb{Z}}
\renewcommand{\P}{\mathbb{P}}
\newcommand{\A}{\mathcal {A}}

\newcommand{\Alg}{\textsf{Alg}}
\newcommand{\Opt}{\textsf{Opt}}
\newcommand{\Flag}{\textsf{Flag}}
\newcommand{\NFlag}{\overline{\textsf{Flag}}}
\renewcommand{\vec}{\mathbf}

\newcommand{\one}{\mathbf{1}}

\newtheorem{theorem}{Theorem}[section]

\newtheorem{corollary}[theorem]{Corollary}

\newtheorem{lemma}[theorem]{Lemma}

\newtheorem{claim}[theorem]{Claim}
\newtheorem{definition}{Definition}

\newcommand\pd[1]{\textcolor{magenta}{#1}}

\title{Online Matroid Embeddings}

\author{Andr\'{e}s Cristi$^1$ \and Paul D\"{u}tting$^2$  \and Robert Kleinberg$^3$ \and Renato Paes Leme$^2$ \and Neel Patel$^1$}
\date{$^1$EPFL, andres.cristi@epfl.ch, neel.patel@epfl.ch\\
$^2$Google Research, duetting@google.com, renatoppl@google.com\\
$^3$Cornell University, rdk@cs.cornell.edu\\
\bigskip
October 2025}

\begin{document}

\maketitle

\setcounter{page}{0}
\begin{abstract}
    We introduce the notion of an online matroid embedding, which is an algorithm for mapping an unknown matroid that is revealed in an online fashion to a larger-but-known matroid. 
    We establish the existence of such an embedding for binary matroids, and use it to relate variants of the binary matroid secretary problem to each other, showing that seemingly simpler problems are in fact equivalent to seemingly harder ones (up to constant-factors). Specifically, we show this to be the case for the version of the matroid secretary problem in which the matroid is not known in advance, and where it is known in advance. 
    We also show that the version with known matroid structure, is equivalent to the problem where weights are not fully adversarial but drawn from a known pairwise-independent distribution. 
\end{abstract}


\section{Introduction}

A common setup in online algorithms is to have a matroid whose structure is revealed to the algorithm one element at a time. The algorithm processes the elements of the ground set in sequence, and at each point in time, it has access to the dependencies between the elements that have already arrived. Typical examples 
include the famous matroid secretary problem (MSP) \cite{BabaioffIK07,BabaioffIKK07,Lachish14,FeldmanSZ18} and matroid prophet inequalities \cite{ChawlaHMS10,KleinbergW12}. 

In such problems, is there any advantage in knowing the matroid structure in advance? Imagine the following situation: we are processing an unknown matroid $\M$; however, we know a fixed (potentially very large) matroid $\BigM$ that has an isomorphic copy of every possible matroid $\M$, and we can construct this embedding online. We will show that if such an object exists, then we can reduce the version of the problem where the matroid is revealed online to the version of the problem where the matroid structure is known, by assuming our matroid is $\BigM$. 
Moreover, if the on-the-fly embedding maintains uniform random order, then the existence of such an online embedding implies that in that class, the matroid secretary problem with \emph{unknown} structure is no harder than the matroid secretary problem with \emph{known} structure. 

Important recent progress on the MSP has established that it is equivalent to the matroid prophet secretary problem with correlated distributions \cite{dughmi2021matroid,dughmi-outer-limits}. Though seemingly unrelated, another consequence of the existence of such online embeddings will be that---for certain matroids--- this equivalence holds even if we impose pairwise independence.

\subsection{Our Contribution}


\paragraph{Online Matroid Embedding}

Our main conceptual contribution is to define the notion of an \emph{online matroid embedding} (OME) in Section~\ref{sec:OMEs}. For a given class of matroids $\mathcal{C}$ and a host matroid $\BigM$, we define an OME as a set of matroid monomorphisms, i.e., mappings that preserve the matroid structure, from any matroid in the class $\mathcal{C}$ into $\BigM$ that can be constructed sequentially, only using calls to an independence oracle over the set of elements observed so far.

The use of embeddings in algorithm design is an idea that has been successfully explored in other contexts, most notably, metric embeddings both in classic settings \cite{Bourgain85,LinialEtAl95,bartal1998approximating} as well as more recently in online settings \cite{indyk2010online, barta2020online, newman2023online}. While our motivation and main application is the matroid secretary problem, we believe  that understanding maps between matroids preserving structure is an important mathematical question in its own right that can enable other algorithmic applications beyond the matroid secretary problem. 

\paragraph{Consequences for the MSP}

We use the concept of online matroid embeddings to gain insights into the complexity of the matroid secretary problem (MSP) on \emph{binary matroids} (see Section~\ref{sec:prelims}). Specifically, we relate different variants of the problem to each other and show that seemingly simpler ones are actually equivalent to harder ones (up to constant factors). See Figure~\ref{fig:reductions} for an overview of the reductions that we establish in this paper.


The three variants we are interested in are: (1) the online-revealed-matroid MSP, where the matroid is a priori unknown to the online algorithm and the algorithm has access to an independence oracle on the already arrived elements, (2) the known-matroid MSP, where the structure of the matroid is known to the algorithm in advance, and (3) the prophet MSP, where the matroid structure is known in advance and additionally the weights of the elements are drawn from a known, but possibly correlated distribution.

\begin{figure}
\begin{center}
\begin{tikzpicture}[
squarednode/.style={rectangle, minimum size=5mm},
]
\node[squarednode] (knownMSP) {\begin{tabular}{c}known-matroid\\MSP\end{tabular}};
\node[squarednode]   (prophetMSP) [right=of knownMSP] {\begin{tabular}{c}prophet MSP\\ w/ pairwise independence\end{tabular}};
\node[squarednode]   (onlinerevealedMSP) [left=of knownMSP] {\begin{tabular}{c}online-revealed-matroid\\MSP\end{tabular}};

\draw[->] (onlinerevealedMSP.north) to [bend left=25] node[above] {$\succeq$} ([xshift=-.1cm]knownMSP.north);
\draw[->] ([xshift=+.1cm]knownMSP.north) to [bend left=25] node[above] {$\succeq$} (prophetMSP.north);
\draw[->] (prophetMSP.south) to [bend left=25] node[below] {Thm.~\ref{thm:gen_pw_ind_PS_to_MSP}, Cor.~\ref{cor:gen_pw_ind_PS_to_MSP}} node[above]{$\preceq_C$} ([xshift=+.1cm]knownMSP.south);
\draw[->] ([xshift=-.1cm]knownMSP.south) to [bend left=25] node[below] {Thm.~\ref{thm:reduction}, Cor.~\ref{cor:reduction}} node[above]{$\preceq_\epsilon$}  (onlinerevealedMSP.south);
\end{tikzpicture}
\end{center}
\caption{Reductions for binary matroids. We use $P \succeq Q$ to indicate that $P$ is harder than $Q$, and we use $\succeq_\epsilon$ and $\succeq_C$ to designate an additive $\eps$ or multiplicative factor $C$  loss in approximation.}
\label{fig:reductions}
\end{figure}

Clearly, the online-revealed-matroid MSP is harder than the known-matroid MSP and the known-matroid MSP is harder than the prophet MSP, in the sense that if we have an $\alpha$-approximation for one problem, then we also have an $\alpha$-approximation for the other. 
Two main implications of our work are ``inverses'' of these statements for binary matroids, that hold up to a constant-factor loss, and apply even if we impose pairwise-independence in the prophet MSP. 
Such reductions between different average-case problems are notoriously difficult to achieve, as they need to ensure or maintain rather stringent assumptions on the input distribution that are essential for the target algorithm to be applicable in a meaningful way, and the required properties are easily disrupted.

\bigskip \emph{Step 1: A Reduction From Online-Revealed Matroid MSP to Known-Matroid MSP.}
In Theorem \ref{thm:reduction} we show that the existence of an OME for a class of matroids 
enables a reduction from the online-revealed-matroid MSP to the known-matroid MSP. The challenge in proving this is to show that the online embedding can be used in a way that (almost) maintains uniform random arrival order. More precisely, let $\M$ be the unknown matroid that is revealed to the algorithm in an online fashion and let $f$ be an OME into $\BigM$. Now consider the reduction: upon arrival of an element-weight pair $(e,w_e)$ in matroid $\M$ at iteration $t$, we construct the corresponding element-weight pair $(f(e),w_e)$ as an input to MSP on $\BigM$ at iteration $t$. However, this is not a valid input to known-matroid MSP on $\BigM$ as it does not construct a random arrival order over the elements that are not in the image of $\M$. 

To overcome the shortcoming of the above simple reduction, we interleave the elements in $\BigM$ that are in the image of $\M$ with the remaining elements in $\BigM$. 
In the proof of Theorem~\ref{thm:reduction}, our main technical argument shows that,
while this interleaving does not ensure uniformity of the arrival order of the elements in $\BigM$, it leads to an arrival order over elements in $\BigM$ that
is close to uniformly random arrival order in total variation distance (Section~\ref{sec:proof_of_MSP_known_to_unknow}). We then complete the reduction with a coupling argument that shows that the existence of an $\alpha$-competitive algorithm for the known-matroid MSP implies the existence of an $(\alpha-\epsilon)$-competitive algorithm for the online-revealed-matroid MSP.

Together with the existence of an OME that maps binary matroids into the complete binary matroid (see Section~\ref{sec:binary_matroid} and discussion below), our
reduction implies that an algorithm for the MSP over binary matroids cannot meaningfully use any advance information about the matroid (Corollary~\ref{cor:reduction}). This is in contrast to all known $O(1)$-competitive algorithms for special cases of binary matroids \cite{korula2009algorithms,dinitz2014matroid}.



\bigskip
\emph{Step 2: A Reduction from Prophet MSP w/ Pairwise Independence to Known-Matroid MSP.}
In Theorem~\ref{thm:gen_pw_ind_PS_to_MSP}, we show that the existence of an OME from a class of matroids $\C$ 
to $\BigM$ satisfying a $2$-transitivity property (that is satisfied by complete binary matroids, see definition in Section~\ref{sec:prelims}),
allows to translate an $\alpha$-competitive algorithm for prophet MSP with pairwise-independent weight distribution on $\BigM$
into a 
$C \cdot (\alpha - o(1))$-competitive 
algorithm for known-matroid 
MSP on matroid $\M \in \C$, for some constant $C > 0$.  



To establish Theorem~\ref{thm:gen_pw_ind_PS_to_MSP}, we build on \cite{dughmi2021matroid,dughmi-outer-limits} and show how to reduce prophet MSP with arbitrary correlation on $\M \in \C$ to prophet MSP with pairwise-independent weight distribution on $\BigM$.
\footnote{A technical detail that we are ignoring here is that our reduction is from a restricted version of the prophet MSP with arbitrary correlation, which results in an additional constant-factor loss.} 
To prove this, we first show that any uniformly random automorphism $f: \BigM \rightarrow \BigM$ satisfies the following property: for any pair of elements $e, e'$, $\Pr[f(e) = e'] = \frac{1}{n}$ and for any two pairs of independent elements $e_1, e_2$ and $e_1', e_2'$, $\Pr[f(e_1) = e_1' \land f(e_2) = e_2'] = \frac{1}{n\cdot (n-1)}$, where $n = |\BigM|$ (Lemma~\ref{lem:pw_ind_emmbedding}). This property allows us to construct an ``almost pairwise independent'' randomized OME 
$f': \M \rightarrow \BigM$ by simply composing the given OME 
with a uniformly random automorphism on 
$\BigM$.  
%
%
%
The resulting weight distribution is approximately pairwise independent in the sense that for any pair of elements $e,e' \in \BigM$ and weights $w,w'$ it holds that
$|\Pr[w(e) = w ]\cdot \Pr[w(e') = w'] - \Pr[w(e) = w' \land w(e') = w']| = O \left( \frac 1 {n^2 }\right)$.

To conclude the proof, we show that there exists a pairwise-independent distribution that is close to the induced weight distribution over $\BigM$ 
(Theorem~\ref{thm:tv_distance}). 
This is one of the most technical results of the paper (see discussion below, and Section~\ref{sec:techncial_proof}).
Combining Theorem~\ref{thm:tv_distance} with a coupling argument similar to the one in our other reduction
completes the proof.



We note that constructing 
an exactly $k$-wise independent distribution from an approximately 
$k$-wise independent distribution has been studied in previous work \cite{alon2003almost, alon2012almost,alon2007testing}. 
However, their techniques focus on a set of \emph{Bernoulli} random variables with identical marginals \cite{alon2003almost,alon2007testing} or ``uniformity'' of the underlying random variables \cite{alon2012almost}---both conditions do not hold in our setting as the weight distribution in the prophet MSP instance can be arbitrarily correlated. In fact, in both works \cite{alon2003almost,alon2012almost}, 
they show that if the random variables $X_1, \dots, X_n$ satisfy ``uniformity" and $|\E[X_i\cdot X_j] - \E[X_i] \cdot \E[X_j]| \leq \eps$ then there exists pairwise independent random variables $\tilde X_1, \dots, \tilde X_n$ within a 
distance of $O(n^2 \cdot \eps)$ --- which is not enough for our purpose as $\eps = \Theta(1/n^2)$ in our case.  

To obtain Theorem~\ref{thm:tv_distance}, we construct an explicit pairwise-independent weight distribution over $\BigM$ by a sequence of ``small'' perturbations to a naturally induced ``almost'' pairwise-independent distribution. At each step, we perturb the original distribution such that $\omega(1)$ many pairs of random variables end up being independent (Procedure~\ref{proc:one} and Procedure~\ref{proc:two} in Section~\ref{sec:techncial_proof}) while always decreasing the pairwise correlation of the rest of the pairs (Lemma~\ref{lem:pairwise_ind}). Then the main technical work is devoted to showing that the total deviation through our procedures is in the order of $\eps\cdot o(n^2)$ 
(Section~\ref{sec:proc_1_analysis} and Section~\ref{sec:proc_2_analysis}), which, combined with the fact that $\eps = O(1/n^2)$ leads to the desired result. We believe that our idea of sequentially constructing small perturbations would find further applications to obtain exact pairwise (or $k$-wise) independent distributions from their approximate counterparts in other settings.


\paragraph{Constructing OMEs}
In Section~\ref{sec:binary_matroid} (Theorem \ref{thm:omm_binary} and  Theorem \ref{thm:order_independent_omm}), we provide a complete analysis for binary matroids. Namely, for the class of binary matroids $\M$ with $n$ elements, there is an online matroid embedding into $\BigM$ the complete binary matroid $\mathbb{F}_2^n$. We also develop a technique for making the OME order-independent: we use properties of the automorphism group of $\F_2^n$ together with randomization to ensure that the images of the elements in $\M$ are not correlated with the arrival order. This technique is in fact more general (Theorem \ref{thm:general_order_independence}) and can be applied whenever the group of automorphisms of the host matroid is ``sufficiently rich'' in a sense that the theorem statement makes precise. 

The key property of binary matroids that we exploit to establish these results is that in $\mathbb{F}_2^n$, there is a unique element that completes a circuit, in the sense that there cannot be two circuits of the same size that intersect in all but one element of each.

We refer to online matroid embeddings where both $\M$ and $\BigM$ are of the same class as ``within-class'' OMEs. In Section \ref{sec:graphic_regular} we explore whether such ``within-class'' OMEs can exist for graphic matroids. We show that such embeddings cannot exist, in fact we show that graphic matroids cannot be embedded in an online-fashion to regular matroids.  
To rule out the existence of such an online embedding, we show that if it would exist, then $\BigM$ must contain an isomorphic copy of $\mathbb{F}_2^n$. However, $\mathbb{F}_2^n$ contains an isomorphic copy of the Fano plane
which is not representable over $\mathbb{F}_3$ \cite{tutte1958homotopy}. Hence $\BigM$ can’t be regular.

We believe that the lack of online matroid embeddings for graphic matroids/regular matroids that ``don't leave the class'' may shed light on why progress on the known-matroid MSP for graphic and regular matroids has not extended to the online-revealed version of these problems, and more generally the MSP for general binary matroids. 

In Section~\ref{sec:laminar} we give another example of an OME, namely for laminar matroids. We show (in Theorem~\ref{thm:laminar}) how to embed the class of laminar matroids $\M$ with at most $n$ elements into $\BigM$ which is a complete linear matroid of rank $n$ over any field with sufficiently many elements. 

Finally in Section \ref{sec:bounded_rank} we show an impossibility result of constructing an OME for the class of all matroids. This is shown by studying finite projective planes and showing that for those matroids,  elements that haven’t arrived yet impose non-trivial constraints on the already arrived elements. As a corollary we obtain an impossibility of constructing an OME for the class of all matroids representable over fields of characteristic at least $7$.

\paragraph{Approximate OMEs}
In Section \ref{sec:conclusion} we extend the notion of an OME to allow distortion, i.e., the map approximately preserves the rank. We observe that the $\alpha$-partition property in 
\cite{AbdolazimiKKG23} and 
\cite{DughmiKP24} can be viewed as an approximate matroid embedding into the free matroid. 

First, combining our formalism with their lower bounds on embedding into the free matroid, we also provide a lower bound on the distortion needed to embed the complete binary matroid into a graphic matroid (Corollary \ref{cor:graphic_distortion}). This result is an example of the power of the formalism: with the right definitions, extending the lower bound to larger classes becomes a simple corollary. 

In Theorem \ref{thm:upper_bound_binary_partition} we show the tightness of the $\Omega(n / \log n)$ lower bound in \cite{DughmiKP24} of the distortion of embedding the complete binary matroid into the free matroid by constructing an embedding achieving this distortion. 
%
We also show that there is no constant approximate online embedding of the class of graphical matroids into a free matroid when the underlying matroid is not known upfront (Theorem~\ref{thm:no_embedding_graphical}). Therefore, any constant competitive algorithm for unknown graphical matroid secretary that relies on constructing an online embedding of the graph into a free-matroid has to exploit the random arrival order of the underlying elements or develop new techniques that do not rely on online embedding into a free matroid.



\subsection{Discussion and Significance of Results}

We believe that the existence or non-existence of (approximate) online matroid embeddings can shed new light on different classes of matroids and how they relate to each other. In this work, we demonstrate two implications for the matroid secretary problem.

Our first implication (Theorem \ref{thm:reduction}) offers the first formalization of the intuition that, in general, advance knowledge of the matroid structure should not help in the design of a constant-competitive algorithm for the matroid secretary problem. In light of this, it would be interesting to develop algorithms for classes of matroids for which constant-competitive algorithms exist when the algorithm has advance knowledge of the matroid structure \cite[e.g.,][]{korula2009algorithms,dinitz2014matroid}.

Our second implication (Theorem~\ref{thm:gen_pw_ind_PS_to_MSP}), in turn, presents a novel ``line of attack'' for obtaining such an algorithm for the class of binary matroids (for which no constant-competitive algorithm is known). While it was already known that it suffices to find such an algorithm for the secretary prophet version with correlated weights~\cite{dughmi2021matroid,dughmi-outer-limits}, general correlated weight distributions offer little additional structure.
Our result shifts the challenge away from intractable arbitrary correlations, towards the better-understood realm of pairwise independent distributions. 
Pairwise independent distributions admit powerful tools like concentration inequalities 
and have found application in areas such as
hashing and constructions of pseudo-random generators (for more details, see surveys \cite{luby2006pairwise,salil2012pseudorandomness}), as well as prophet inequalities \cite{pi-uniform-prophet}.


\subsection{Related Work}
\paragraph{Matroid Secretary Problem} 

The matroid secretary problem was first studied in \cite{BabaioffIK07,BabaioffIKK07,BabaioffIKK18}, who gave a $O(\log(\rank))$-competitive algorithm for general matroids. This bound was improved to $O(\sqrt{\log(\rank)})$ in \cite{ChakrabortyL12}, and the state-of-the-art is a $O(\log \log(\rank))$-competitive algorithm \cite{Lachish14,FeldmanSZ18}. The algorithms of \cite{Lachish14,FeldmanSZ18} only uses independence oracle calls on subsets of the elements revealed so far.

For graphic matroids there is a $O(1)$-competitive algorithm, provided that the graphic matroid is known in advance \cite{korula2009algorithms}. The same is true for the more general class of regular matroids \cite{dinitz2014matroid}. Laminar matroids also admit an $O(1)$-competitive algorithm \cite{im2011secretary,jaillet2013advances}. Some evidence for the difficulty of the matroid secretary problem for general binary matroids can be found in \cite{LeichterMP22} and \cite{AbdolazimiKKG23}, showing that binary matroids are not $(b,c)$-decomposable, ruling out a promising approach to obtaining an $O(1)$-competitive algorithm for this class.

Oveis Gharan and Vondr\'{a}k \cite{oveis2013variants} systematized the study of matroid secretary problem variants, establishing a notation for classifying problem variants according to whether the elements arrive in adversarial or random order, whether the assignment of weights to elements is adversarial or random, and whether or not the matroid structure is known in advance. In their nomenclature, the main question addressed in our work is whether the RO-AA-MK variant is equivalent to the RO-AA-MN variant for matroids in general, or for specific classes of matroids. Interestingly, for variants with adversarial arrival order but random weight assignment, \cite{oveis2013variants} demonstrates a stark qualitative difference in approximability: the AO-RA-MK model (when the matroid structure is known in advance) admits a 64-competitive algorithm for all matroids, whereas the AO-RA-MN model (when the number of elements is known in advance but the matroid structure is revealed online) has no constant-competitive algorithm even for the class of {\em rank one matroids}! 

Very recently, \cite{SantiagoSZ25} gave a $O(1)$-competitive algorithm for the matroid secretary problem in the random assignment model when the matroid structure is not known in advance, and instead is only revealed over time. In a similar spirit, \cite{SantiagoSZ23} presents an online contention resolution scheme for graphic matroids, that uses almost no advance information about the graph. However, they assume that the endpoints of the edges are revealed upon their arrival which leads to an obvious OME into a graphical matroid.


\paragraph{Matroid Prophet Inequalities} 
The matroid prophet inequality problem was first studied in \cite{HajiaghayiKS07}. An asymptotically optimal $(1-o(1))$-competitive algorithm for $k$-uniform matroids was given in \cite{Alaei14}.
A tight $O(1)$-competitive algorithm for the matroid prophet inequality problem was given in \cite{KleinbergW12}, also see \cite{DuttingFKL20}  for the problem of maximizing submodular functions subject to matroid constraints.   
Constant-factor competitive algorithms can also be obtained via online contention resolution schemes (OCRS) \cite{FeldmanSZ21}. Random-order versions of the matroid prophet inequality problem are studied in \cite{EhsaniHKS18}. 

To the best of our knowledge, all these algorithms exploit that the matroid structure is known in advance. An additional difficulty for reductions of the type we present in this paper, is that typically these algorithms need to know the identity of the distribution that a certain element's weight is drawn from. For the i.i.d.~case this is obviously not an obstacle, and so our reductions apply. We believe that extensions of our techniques might shed further light on the variant of the matroid prophet inequality problem, in which the matroid is revealed online.


\paragraph{Metric Embeddings and Distortion}

An important inspiration for this work comes from the literature on metric embeddings. A classic result in this context is Bourgain's theorem \cite{Bourgain85}. The algorithmic importance of such embeddings, and Bourgain's theorem in particular, was first highlighted in the seminal papers of \cite{LinialEtAl95,bartal1998approximating}.

Since then metric embeddings have found applications in a host of algorithmic problems, see, e.g., the survey of \cite{Indyk01} and Chapter 15 of \cite{Matousek02}. Closer to our notion of online matroid embeddings is a recent line of work on online metric embeddings 
\cite{indyk2010online, barta2020online, newman2023online}
in which points of a metric space are presented one at a time to an algorithm who must then decide on a mapping to the host metric space. The main difference is that instead of preserving a matroid structure, those papers try to minimize metric distortion.



\section{Matroids, Morphisms, and $K$-representations}
\label{sec:prelims}
Throughout the paper, we will use $[n]$ to denote the set of integers $\{ 1, 2, \hdots, n\}$.

\paragraph{Matroid Definition} A matroid $\M$ is composed by a ground set $M$ and a rank function $\rank_\M:2^M \rightarrow \Z_+$ satisfying the following properties:
\begin{itemize}
\item 
$\rank_\M(\emptyset) = 0$;
\item $\rank_\M(S \cup \{i\}) - \rank_\M(S) \in \{0,1\}, \forall S, \{i\} \subseteq M$
\item $\rank_\M(S \cup T) + \rank_\M(S \cap T) \leq \rank_\M(S) + \rank_\M(T), \forall S,T \subseteq M$ (submodularity)
\end{itemize}
It follows from the second condition that $\rank_\M(S) \leq \abs{S}$. Whenever $\abs{S} = \rank_\M(S)$ we say that $S$ is an independent set of the matroid. Otherwise, we say that $S$ is dependent. A minimal dependent set is called a \emph{circuit}, i.e., $C \subseteq M$ is a circuit if $C$ is dependent but every strict subset $S \subsetneq C$ is independent. We say that a matroid has rank $r$ if $r = \max_{S} \rank_\M(S)$.

We say that an element $x \in \M$ is a loop if $\rank_\M(\{x\}) = 0$. We say that a matroid is loop-free if every set of one element is independent. Given a set $S \subseteq \M$ we define the span as $\span_\M(S) = \{x \in \M; \rank_\M(S \cup \{x\}) = \rank_\M(S)\}$.

\paragraph{Matroid Morphisms} 
We will  use the same notation to refer to a matroid and its ground set. Given two matroids $\M$ and $\N$ we will define a morphism $f:\M \rightarrow \N$ to be a map between their ground sets that preserves rank, i.e.: $$\rank_\N(f(S)) = \rank_\M(S), \forall S \subseteq \M.$$ Whenever the matroid morphism is an injective map, we will say it is a \emph{matroid monomorphism} or a \emph{matroid embedding}. Whenever it is bijective, we will say it is a \emph{matroid isomorphism}. An isomorphism from a matroid to itself is called an automorphism. (Aside: this paragraph defines the category of matroids in the sense of category theory. However, we won't use any other fact from category theory other than borrowing its very convenient language.)

We refer to the set of automorphisms $\M \rightarrow \M$ as $\Aut(\M)$, which forms a group under composition, i.e., given $f,g \in \Aut(\M)$, then $f \circ g \in \Aut(\M)$ (and $\circ$ satisfies the group axioms).


\paragraph{Element Copies} Given a matroid $\M$ and an integer $k$ we will define the matroid $\M_{[k]}$ by creating $k$ copies of each element of $\M$. Formally, the ground set of $\M_{[k]}$ is $\{(u,j); u \in \M, j \in [k]\}$. The rank function of $\M_{[k]}$ is induced by the projection $\phi: \M_{[k]} \rightarrow \M$ that maps $(u,j) \mapsto u$, i.e., $\rank_{\M_{[k]}}(S) = \rank_\M(\phi(S))$. By definition, the projection $\phi$ is a matroid morphism from $\M_{[k]} \rightarrow \M$.

If $\N$ is a matroid of at most $n$ elements, every morphism $f: \N \rightarrow \M$ can be written as: $f = \phi \circ f'$ where $f': \N \rightarrow \M_{[n]}$ is a monomorphism.

\paragraph{Direct Sum} Given two matroids $\M$ and $\N$, we define their direct sum $\M \oplus \N$ as the matroid whose ground set is the disjoint union of the ground sets of $\M$ and $\N$ and $\rank_{\M \oplus \N}(S) = \rank_\M(S \cap \M) +  \rank_\N(S \cap \N)$ for all $S$ in the disjoint union of ground sets.

\paragraph{Graphic Matroids} We will define a few special classes of interest. We start with \emph{graphic matroids}. Given a graph with edge set $E$, we can define a matroid with ground set $E$ by defining the $\rank(S)$ of a subset $S \subseteq E$ as the maximum number of edges in $S$ that don't form a cycle. We say that a matroid $\M$ is graphic if it is isomorphic to the matroid obtained from an undirected graph as we just described.

As an example, consider the matroid $\M$ with ground set $\{a,b,c\}$ and rank function such that $\rank(S) = \abs{S}$. The matroid is graphic since it is isomorphic to the matroid that can be obtained from any of the graphs in Figure \ref{fig:graphic_matroid}. An important thing to note, however, is that the matroid description contains no information about vertices. It only tells us which sets of edges are independent and which are not. As we can see in the figure, this is typically not enough to fully determine the graph structure.

\begin{figure}[h]
\centering
\begin{tikzpicture}[scale=.9, inner sep=1.5pt] 
 \node[circle,fill] at (0,0) {};
 \node[circle,fill] at (2,0) {};
 \node[circle,fill] at (4,0) {};
 \node[circle,fill] at (6,0) {};

 \draw (0,0)--(2,0);
 \draw (2,0)--(4,0);
 \draw (2,0)--(6,0);
 \node at (1,.3) {$a$};
 \node at (3,.3) {$b$};
 \node at (5,.3) {$c$};

 \begin{scope}[xshift=12cm]

 \node[circle,fill] at (0,0) {};
 \node[circle,fill] at (2,0) {};
 \node[circle,fill] at (-2,1) {};
 \node[circle,fill] at (-2,-1) {};
 \draw (0,0)--(2,0);
 \draw (0,0)--(-2,1);
 \draw (0,0)--(-2,-1);
 \node at (-1,.8) {$a$};
 \node at (-1,.-.8) {$b$};
 \node at (1,.3) {$c$};
 \end{scope}
 
\end{tikzpicture}
\caption{Two graphs that generate the same matroid on their edge set}
\label{fig:graphic_matroid}
\end{figure}
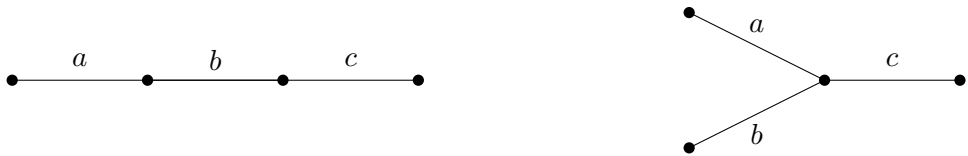

\paragraph{$K$-representable Matroids} Let $K$ be a field (e.g.~$\Q, \R, \F_p$) and let $K^d$ be the vector space formed by $d$-dimensional vectors with coordinates in $K$. We say that a subset of vectors $u_1, \hdots, u_k \in K^d$ is independent if the unique solution to $\alpha_1 u_1 + \alpha_2 u_2 + \hdots + \alpha_k u_k = 0$ for $\alpha_i \in K$ is $\alpha_1 = \alpha_2 = \hdots = \alpha_k = 0$. Any subset of $K^d$ together with the independency relation above defines a matroid. From now on, we will use the notation $K^d$ to represent both the vector space and the corresponding matroid. We say that a matroid $\M$ is $K$-representable if there is a matroid morphism $\M \rightarrow K^d$ for some integer $d$.

If a matroid $\M$ is $K$-representable for every field $K$ we say that $\M$ is a \emph{regular matroid}. Every graphic matroid is representable over any field by mapping an edge $(u,v)$ to the vector $e_u - e_v$ where $e_u$ the the $u$-th unit vector. (This is true even over $\F_2$ where $e_u - e_v = e_u + e_v$.)

For example, the matroids in Figure \ref{fig:graphic_matroid} can be represented by the vectors $(1,-1,0,0)$, $(0,1,-1,0)$, $(0,0,1,-1)$. As it is the case for graphic matroids, the matroid description has no information about vectors and the representation is again not unique. An equally good representation is $(1,0,0), (0,1,0), (0,0,1)$.

We will be specially interested in \emph{binary matroids} which are matroids that are representable over $\F_2$ (the finite field of $2$ elements where addition and multiplication are performed mod $2$).
\paragraph{Laminar matroids} A family of sets, $\mathcal{A}$, is called {\em laminar} if it satisfies the 
property that for any $A,A' \in \mathcal{A}$, at least one of the sets 
$A \cap A', \, A \setminus A', \,  A' \setminus A$ is empty. 
A {\em laminar matroid} $\M$ is one for which there exists a laminar 
family of sets $\mathcal{A}$ consisting of subsets of the ground set of $\M$ and a function $c : \mathcal{A} \to \Z_+$, such that 
the independent sets of $\M$ are precisely those sets $I \subseteq M$ such 
that $|I \cap A| \leq c(A)$ for all $A \in \mathcal{A}$.

\paragraph{Uniform Matroid} We will denote by $\U_{n,r}$ the uniform matroid of $n$ elements and rank $r$. This is the matroid with ground set $[n]$ and whose rank function is $\rank_{\U_{n,r}}(S) = \min(r, \abs{S})$. We call the $\Free_n := \U_{n,n}$ the free matroid of rank $n$, i.e., a matroid of $n$ elements in which every set is independent.

\paragraph{Trivial Matroid} Let $\T$ be the trivial matroid which has ground set $\{0\}$ and rank function $\rank_\T(S) = 0$ for all sets $S$. 

\paragraph{$2$-transitive Matroid}  We say that a simple matroid $\M$ (loop-free and no parallel elements) is $2$-transitive if for any pair of independent sets of size two, $\{e_1,e_1'\}, \{e_2, e_2'\} \in \M$ there exists an automorphism $f \in \operatorname{Aut}(\M)$ satisfying $f(e_1) = e_1'$ and $f(e_2) = e_2'$.
There are several matroids that satisfies the $2$-transitive property including complete affine matroids, complete projective matroids, free matroids, and their truncations \cite{kantor1985homogeneous}.

\section{Online Matroid Embeddings}
\label{sec:OMEs}

We are interested in studying matroids whose structure is revealed to an algorithm in an online fashion. For that, it will be useful to take into account the order in which elements are processed, which we will represent by an indexing of the ground set: $\pi: [n] \rightarrow \M$.

A matroid $\N$ is a restriction of matroid $\M$ if the ground set of $\N$ is a subset of the ground set of $\M$ and $\rank_\M$ coincides with $\rank_\N$ on the ground set of $\N$.

Given a matroid with ordered ground set specified by a pair $(\M, \pi)$, we say that $(\M', \pi')$ is a prefix-restriction of $(\M, \pi)$ if $n' = \abs{\M'} < \abs{\M}$, $\M'$ is the restriction of $\M$ to $\pi([n'])$ and  $\pi'$ is the restriction of $\pi$ to $[n']$.

Let $\C$ be a class of matroids that is closed under restriction (e.g., the class of all matroids, graphic matroids, binary matroids, $K$-representable matroids, matroids of rank at most $r$). An \emph{online matroid morphism} (OMM) for class $\C$ consists of a host matroid $\BigM$, together with  matroid  morphisms $$f_{\M,\pi} : \M \rightarrow \BigM$$ for every $\M \in \C$ and every indexing $\pi: [n] \rightarrow \M$ of the ground set of $\M$, such that for every prefix-restriction $(\M', \pi')$ of $(\M, \pi)$, the map $f_{\M', \pi'}$ is the restriction of $f_{\M, \pi}$ to the ground set of $\M'$.

If all morphisms $f_{\M, \pi}$ are monomorphisms, we say that they form an \emph{online matroid embedding} (OME). Given an online matroid morphism it is easy to construct an online matroid embedding by copying the elements of $\BigM$.

\begin{lemma}\label{lemma:omm_to_ome}
Let $\C$ be a class of matroids, where each matroid $\M \in \C$ has at most $n$ elements and $f_{\M, \pi} : \M \rightarrow \BigM$ form an online matroid morphism. Then there is an online matroid embedding $f'_{\M, \pi} : \M \rightarrow \BigM_{[n]}$.
\end{lemma}

\begin{proof}
We define $f'_{\M, \pi}$ as follows: for each $u \in \M$ if $u = \pi(k)$ let $f'_{\M, \pi}(u)= (f_{\M,\pi}(u),j)$ where $j = \abs{i \in [k]; f_{\M,\pi}(\pi(i)) = u}$. The functions $f'_{\M,\pi}$ are  injective by construction and they are matroid morphisms by the definition of $\BigM_{[n]}$. In fact: $f_{\M,\pi} = \phi \circ f'_{\M,\pi}$ where $\phi$ is the natural projection $\BigM_{[n]} \rightarrow \BigM$. Finally note that they can be constructed online since the identity of the copy used is only a function of the set of elements that arrived up to this point.
\end{proof}

With this definition we can ensure that an online algorithm is able to construct a monomorphism from an unknown matroid $\M$ to $\BigM$ in an online fashion. Consider a matroid $\M$ for which the elements arrive according to $\pi$. At each time $t$, we can observe the structure of the matroid $\M_t$ which is the restriction of $\M$ to $\pi([t])$. Let $\pi_t$ be the restriction of $\pi$ to $[t]$. If we have an online matroid embedding, we can first construct $f_{\M_1, \pi_1}$, then extend to $f_{\M_2, \pi_2}$ and so forth.\\

It will also be convenient to define a randomized online matroid morphism (embedding) which for every matroid $\M \in \C$ and ordering $\pi$ specifies a distribution over (mono)morphisms $f_{\M,\pi} : \M \rightarrow \BigM$ such that for every prefix-restriction $(\M', \pi')$ the distribution of the restriction of $f_{\M,\pi}$ to the ground set of $\M'$ coincides with the distribution of $f_{\M', \pi'}$.

Finally, we say that a randomized online matroid embedding is order-independent if the distribution of $f_{\M, \pi}$ doesn't depend on $\pi$. In other words, for any two orderings $\pi$ and $\pi'$, the morphisms $f_{\M, \pi}$ and $f_{\M, \pi'}$ are equally distributed.

\paragraph{Uniform Order-Independent Embedding} Let $f$ be an order-independent online embedding from $\M \rightarrow \BigM$. We consider an order independent-randomized embedding $g: \M \rightarrow \BigM$ by composing $f$ with uniformly random automorphism $f'\in \operatorname{Aut}(\BigM)$, i.e. $g = f' \circ f$.

Interestingly, whenever $\BigM$ satisfies $2$-transitive property then the embedding $g$ maps each element of $e\in \M$ uniformly at random over the matroid $\BigM$. In addition, for any independent set of pair of elements $\{e,e'\} \subseteq \M$ and pair of elements $\{\tilde e, \tilde e'\} \subseteq \BigM$, the events $\{g(e) = \tilde e\}$ and $\{g(e') = \tilde e'\}$ are `almost' independent. More formally,
\begin{lemma}\label{lem:pw_ind_emmbedding}
   Given a simple host matroid (loop-free and no parallel elements) $\BigM$ such that for any two pairs of distinct elements $\{e_1, e_2\}$ and $\{e_1', e_2'\}$ there exists an automorphism $f' \in \operatorname{Aut}(\BigM)$ satisfying $f'(e_1) = e_1'$ and $f'(e_2) = e_2'$, then a uniformly random automorphism $f$ sampled from $\operatorname{Aut}(\BigM)$ satisfies:
    \begin{enumerate}
        \item For any $e, e' \in \BigM$, $\Pr[f(e) = e'] = \frac{1}{n}$.
        \item For any two pairs of elements $e_1, e_2$ and $e_1' , e_2'$ (s.t. $e_1 \neq e_2$ and $e_1' \neq e_2'$), we have $$\Pr[f(e_1) = e_1' \land f(e_2) = e_2'] = \frac{1}{n\cdot (n-1)}$$
    \end{enumerate}
\end{lemma}
\begin{proof}
   Consider the action of the group of automorphisms  $\operatorname{Aut}(\BigM)$ on the set of pairs of distinct elements $P=\{(e_1,e_2): e_1,e_2\in \BigM, e_1\neq e_2\}$. For two pairs $(e_1,e_2), (e_1', e_2')\in P$, the set of automorphisms $f\in \operatorname{Aut}(\BigM)$ that satisfy $f(e_1,e_2)=(e_1',e_2')$ is nonempty, and therefore, it is a coset of the subgroup of stabilizers of $(e_1,e_2)$ (i.e., the set of automorphisms that satisfy $f(e_1,e_2)=(e_1,e_2)$). Since all cosets of a subgroup must have the same size, and because all pairs $(e_1',e_2')\in P$ define a different coset, for a uniformly drawn automorphism $f$,
   \[\Pr[f(e_1,e_2)=(e_1',e_2') ]= \frac{1}{|P|} = \frac{1}{n\cdot (n-1)}.\]
   An analogous argument gives that $\Pr[f(e)=e']=1/n$.
\end{proof}

\paragraph{Online vs.~Offline Embeddings} The difficulty of constructing an online matroid embedding is that the elements of $\BigM$ corresponding to certain elements of $\M$ must be chosen before the full matroid structure of $\M$ is known. If we merely wanted to construct a matroid $\BigM$ that contains an isomorphic copy of every matroid in $\C$, that would be very easy: $\BigM$ could be taken to be the direct sum of all the matroids in $\C$.


\section{OMEs for Binary Matroids}\label{sec:binary_matroid}

Before we discuss how to use online matroid embeddings in online algorithms, it is important to show first that they exist in non-trivial cases. For that, we will provide a complete analysis for binary matroids. Recall that $\F_2^n$ is the complete binary matroid of rank $n$ and that graphic matroids and regular matroids are special cases of binary matroids. Our first result is the existence of an OME for this class. This will be done by showing the existence of an OMM and using Lemma \ref{lemma:omm_to_ome} to convert an OMM to an OME. 

Our first step is to show a lemma that the matroid $\F_2^n$ is special in the sense that its group of matroid automorphisms  coincides with its group of vector space automorphisms:

\begin{lemma}\label{lem:vector_to_matroid_aut}
A mapping $A: \F_2^n \rightarrow \F_2^n$ is a matroid automorphism iff it is an automorphism of vector spaces.
\end{lemma}

\begin{proof}
An automorphism of vector spaces $A: \F_2^n \rightarrow \F_2^n$ is a bijection such that for any vectors $v_1, \hdots, v_k \in \F_2^n$ it holds that $A(\sum_{i=1}^k v_i) = \sum_{i=1}^k Av_i$. This in particular implies that a set of vectors $v_1, \hdots, v_k$ is independent iff the vectors $Av_1, \hdots, Av_k$ are independent. This is because a coefficient vector $(\alpha_1,\ldots,\alpha_k)$ satisfies the equation $\sum_{i=1}^k \alpha_i v_i = 0$ if and only if it satisfies $\sum_{i=1}^k \alpha_i Av_i= A(\sum_{i=1}^k \alpha_i v_i) = 0$, so the first equation has only trivial solutions if and only if the second equation has only trivial solutions.

For the opposite direction, if $A$ is a matroid automorphism and $e_1, \hdots, e_n$ is the standard basis of $\F_2^n$ then $Ae_1, \hdots, A e_n$ must be linearly independent elements of $\F_2^n$. Now, take any vector $v = \sum_{i \in S} e_i$. Since $\{v\} \cup \{e_i; i \in S\}$ forms a circuit, then $\{Av\} \cup \{Ae_i; i \in S\}$ must form a circuit. Since the only non-zero constant in $\F_2$ is $1$, it must hold that: $Av + \sum_{i\in S} Ae_i = 0$ and hence $Av = \sum_{i \in S} Ae_i$. Hence $A$ is also an automorphism of vector spaces.
\end{proof}

\begin{theorem}\label{thm:omm_binary}
Let $\C$ be the class of binary matroids of at most $n$ elements and let $\BigM$ be the complete binary matroid $\F_2^n$. Then there exists an OMM for $\C$ into $\BigM$.
\end{theorem}

\begin{proof}
Given a binary matroid $\M$ we construct a mapping $f:\M \rightarrow \F_2^n$ as follows. We keep a counter $k$ initially set to $1$. For each element $a$ we process, 
if it is independent of the previously arrived elements (i.e.~there are no circuits containing $a$ and the elements seen so far), we set $f(a) = e_k$ and increment $k$. Otherwise, $a$ forms a circuit with a set of previously arrived elements $u_1,\hdots, u_m$ for some integer $m \geq 0$. This means that their image $f(a), f(u_1), \hdots, f(u_m)$ must be a minimal $\F_2$-linearly dependent set. Since the only non-zero constant in $\F_2^n$ is $1$, then it must hold that: $$f(a) + f(u_1) + \hdots + f(u_m) = 0$$ and hence we can map: $f(a)$ to $f(u_1) + \hdots + f(u_m)$ (recall that $1=-1$ in $\F_2$).

Finally, we need to argue that $f$ is a matroid morphism. Observe that if $\M$ is a binary matroid, then there exists a morphism $g:\M \rightarrow \F_2^n$. Let $\{b_1, \hdots, b_r\}$ be the elements of $\M$ such that $f(b_i) = e_i$. By the fact that $g$ is matroid morphism, $g(b_1), \hdots, g(b_r)$ are linearly independent elements in $\F_2^n$. By Lemma \ref{lem:vector_to_matroid_aut} there is an automorphism $A \in \Aut(\F_2^n)$ that takes $g(b_i)$ to $e_i$. Since matroid morphisms compose, $A g:\M \rightarrow \F_2^n$ is matroid morphism. 

\begin{center}
\begin{tikzcd}
\M \arrow[rd, "g"'] \arrow[r, "f"] & \F_2^n \\
& \F_2^n \arrow[u, "A"']
\end{tikzcd}
\end{center}

Finally, we argue that $f(a) = Ag(a)$ for all $a$ in $\M$. We show this by induction. For each element processed by the algorithm, if it is independent from previously arrived elements, then $f(a) = Ag(a)$ by construction. Otherwise, there are previously arrived elements such that $a,u_1,\hdots, u_m$ form a circuit. Hence their image under $Ag$ must be a linearly dependent set of $\F_2$-vectors, which means that:
$$Ag(a) = Ag(u_1) + \hdots + Ag(u_m) = f(u_1) + \hdots + f(u_m) = f(a)$$
where the second equality holds by induction. Since $f$ coincides with $Ag$, $f$ is a matroid morphism.
\end{proof}

Furthermore, there is a randomized OME that is order-independent. We will show it as a consequence of the following lemmas:

\begin{lemma}\label{lem:single_orbit}
Given a binary matroid $\M$ and the complete binary matroid $\F_2^n$, if there are two matroid morphisms $f,g: \M \rightarrow \F_2^n$, then there exist an automorphism $A \in \Aut(\F_2^n)$ such that $f = A \circ g$.
\end{lemma}

\begin{proof}
 Let $r$ be the rank of $\M$ and let $\{b_1, \hdots, b_r\}$ be a basis of $\M$. Then $\{f(b_1), f(b_2), \hdots, f(b_r)\}$ and $\{g(b_1), g(b_2), \hdots, g(b_r)\}$ are both sets of independent vectors in $\F_2^n$. Then there exists an automorphism $A$ of vector spaces (and hence a matroid automorphism)  that sends $g(b_i)$ to $f(b_i)$. For any other element in $v \in \M$ consider any circuit formed with a subset of the basis. If $\{v\} \cup \{b_i; i \in S\}$ is a circuit then it must be the case that: $f(v) = \sum_{i \in S} f(b_i)$ and $g(v) = \sum_{i \in S} g(b_i)$. Given that $A$ is a automorphism of vector spaces, then: $Ag(v) = \sum_{i \in S} Ag(b_i) = \sum_{i \in S} f(b_i) = f(v)$.
\end{proof}

\begin{theorem}\label{thm:order_independent_omm}
There is a order-independent randomized OMM from the class of binary matroids $\C$ into the complete binary matroid.
\end{theorem}

\begin{proof}
Let $f_{\M,\pi}$ be the online matroid morphism constructed in Theorem \ref{thm:omm_binary} and consider $A \circ  f_{\M, \pi}$ when $A$ is drawn uniformly at random from $\Aut(\F_2^n)$. It is clear that for every fixed $A$ the morphisms $A \circ f_{\M, \pi}$ still form an OMM. We only need to check that they are order-independent. To see that, observe that if $\pi$ and $\pi'$ are two different orderings of the ground set of $\M$ then $f_{\M, \pi}$ and $f_{\M, \pi'}$ are two morphisms $\M \rightarrow \F_2^n$. By the previous lemma, there is $A_0 \in \Aut(\F_2^n)$ such that $f_{\M, \pi} = A_0 \circ f_{\M, \pi'}$. Now, the distribution of $A \circ f_{\M, \pi}$ for a random $A \sim \Aut(\F_2^n)$ is the same distribution as $A \circ A_0 \circ f_{\M, \pi'}$ which is the same distribution of $A \circ f_{\M, \pi'}$, since $A_0 \circ A$ is also uniformly distributed over $\Aut(\F_2^n)$.
\end{proof}

\paragraph{Extending to Copies} In the following section we will be needing an online matroid embedding. For that reason, we need to extend the last two theorems to deal with copies. The extension is rather simple: we only need to observe that a matroid autormorphism of the commplete binary matroid with $n$ copies of each element $(\F_2^n)_{[n]}$ can be decomposed into an automorphism $A \in \Aut(\F_2^n)$ and indexings of the identities of the copies. 

\begin{lemma} If $f \in \Aut((\F_2^n)_{[n]})$ then there exists $A \in \Aut(\F_2^n)$ and indexings $\sigma_u:[n] \rightarrow [n]$ for each $u \in \F_2^n$ such that $f((u,j)) = (Au, \sigma_u(j))$
\end{lemma}

\begin{proof} Let $id$ be the identity map and $\phi: (\F_2^n)_{[n]} \rightarrow \F_2^n$ the natural projection. Now, $\phi \circ f$ and $\phi$ are two matroid morphisms from $(\F_2^n)_{[n]} \rightarrow \F_2^n$ so by Theorem \ref{lem:single_orbit} there is $A \in \Aut(\F_2^n)$ such that $\phi \circ f  = A \circ \phi$ (see the commutative diagram below). This means in particular that 
$f((u,j)) = (Au, \sigma_u(j))$ for some indexings $\sigma_u$.

\begin{center}
\begin{tikzcd}[row sep=tiny]
& (\F_2^n)_{[n]} \arrow[r, "\phi"] & \F_2^n\\
(\F_2^n)_{[n]} \arrow[rd, "id"'] \arrow[ru, "f"] & & \\
& (\F_2^n)_{[n]} \arrow[r, "\phi"]  & \F_2^n \arrow[uu, "A"']
\end{tikzcd}
\end{center}

\end{proof}

With that,  Theorem \ref{thm:order_independent_omm} automatically extends to the matroid with copies $(\F_2^n)_{[n]}$ by taking a random automorphism from group $\Aut((\F_2^n)_{[n]})$.

\paragraph{Remark A (Single Orbit Morphisms)} Theorem \ref{thm:order_independent_omm} is not particular to binary matroids. The only fact it uses is that all the morphisms of the OMM belong to the same orbit under the action of the automorphism group $\Aut(\BigM)$. We can state it more generally as follows. The proof is identical to Theorem \ref{thm:order_independent_omm}, so we omit it here.

\begin{theorem}[Generalization of Theorem \ref{thm:order_independent_omm}]\label{thm:general_order_independence}
    Let $f_{\M, \pi}$ be an OMM of class $\C$ into $\BigM$ such that given two orderings $\pi$ and $\pi'$ of the ground set of $\M$, there is an automorphism $A \in \Aut(\BigM)$ such that $f_{\M, \pi'} = A \circ f_{\M, \pi}$. Then there is an order independent randomized OMM of class $\C$ into $\BigM$.
\end{theorem}

\paragraph{Remark B (Other Fields)} In this section we repeatedly use the fact that $\F_2$ has only one non-zero constant, so whenever we identify a circuit in the matroid, we know exactly what is the linear dependency between the elements in the corresponding vector field. This is no longer true even in slightly larger fields like $\F_3$. If vectors $u, v, w \in \F_3^n$ form a circuit, it could be that: $w = \pm u \pm v$ in the representation. As a consequence, given two matroid morphisms $f,g: \M \rightarrow \F_3^n$ there may not exist a vector-space automorphism $A$ of $\F_3^n$ such that $f= Ag$. In the previous example, if $f(u) = g(u)$, $f(v) = g(v)$ but $f(w) = f(u) + f(v)$ but $g(w) = g(u) - g(v)$ no such automorphism can exist.

\section{OMEs and the Matroid Secretary Problem}\label{sec:reduction}

In this section, we use OMEs to explore the complexity of the matroid secretary problem (MSP) on binary matroids. We consider three 
versions of the problem, each making different assumptions on the data generation process and what's known to the algorithm. We will use OMEs to establish equivalences between these problems, showing that seemingly simpler problems are actually equivalent to harder ones (up to constants).

\paragraph{Three Versions of the MSP.} We consider the following three versions of the MSP,
and aim to establish the relations in Figure~\ref{fig:reductions}.
In all three variants, the goal is an algorithm for selecting elements that form an independent set, and whose combined weight is in expectation an $\alpha$-approximation to the weight of the optimal basis.

\begin{itemize}
\item {\bf Online-revealed-matroid MSP:} 
In this version of the problem,
there is an 
underlying matroid $\M$, which 
is a priori unknown to the algorithm. The algorithm has only access to the number of elements $n = \abs{\M}$ and to a promise that $\M \in \C$ for a class of matroids $\C$. For each element $u\in \M$, an adversary determines a weight $w_u \in \R_+$.
The algorithm 
processes pairs $(u, w_u)$ in random order at each time, but it only knows the rank function restricted to the subset of elements that have already arrived. Upon seeing the element, the algorithm must irrevocably decide whether to accept that element or not, subject to the constraint that the set of accepted elements must be an independent set of $\M$.
\item {\bf Known-matroid MSP:} 
In this version, 
the matroid $\M$ is known to the algorithm ahead of time. The only information missing is the weight of each element, which is again chosen adversarially.
As before elements arrive in random order, and the algorithm must make immediate accept/reject decisions, with the restriction that the chosen set of elements must be an independent set of $\M$. 
\item {\bf Prophet MSP:} In this version, the matroid $\M$ is again known to the algorithm ahead of time. 
However, this time
the weight of each element is drawn from a known distribution $\D$, which can potentially sample weights in a correlated manner. 
As in the other versions elements are then presented to the algorithm in random order, and the algorithm aims to select an independent set of high weight in an online manner. 
\end{itemize}

Note that the first version is clearly harder than the second and the second version is clearly harder than the third, in the sense that an $\alpha$-approximation to the harder problem 
immediately implies an $\alpha$-approximation to the simpler one.
We derive approximate ``inverses'' of these comparisons from the existence of OMEs, even if we restrict prophet MSP to pairwise-independent distributions.


\paragraph{Our Reductions.}

We first use OMEs to show an (essentially exact) ``inverse'' 
of the comparison between the online-revealed-matroid MSP and the known-matroid MSP, implying that for binary matroids the latter is as hard as the former.


\begin{theorem}\label{thm:reduction}
If a class $\C$ of matroids admits a randomized order-independent online matroid embedding into matroid $\BigM$, then an $\alpha$-approximation to the known-matroid MSP for $\BigM_{[n]}$ implies that for every $\epsilon > 0$ there is a $(\alpha-\epsilon)$-approximation to the online-revealed-matroid MSP for $\C$.
\end{theorem}

For the case of binary matroids that we previously discussed, the matroids $\BigM$ and $\BigM_{[n]}$ themselves are binary in which case we can obtain the following corollary:

\begin{corollary}\label{cor:reduction}
For binary matroids, there is no gap in approximability between 
the known-matroid MSP and the online-revealed-matroid MSP.
\end{corollary}

As our second result, we use the existence of OMEs into a 2-transitive host matroid as a tool 
to establish the approximate equivalence of known-matroid MSP and prophet MSP with pairwise-independent distributions.


\begin{theorem}\label{thm:gen_pw_ind_PS_to_MSP}
Suppose a class $\C$ of matroids admits a randomized order-independent online matroid embedding into matroid $\BigM$ which is $2$-transitive. Then an $\alpha$-approximation to the prophet MSP with pairwise-independent weight
distributions for $\BigM_{[k]}$, implies that for some constant $C>0$ there is a $C\cdot (\alpha - o(1))$-approximation to the known-matroid MSP for $\C$.
\end{theorem}

Noting that the full binary matroid is $2$-transitive, and there exists an order-independent OME from the class of binary matroids into a full binary matroid, we obtain the following corollary. 


\begin{corollary}\label{cor:gen_pw_ind_PS_to_MSP}
For binary matroids, there is a constant-factor gap in approximability between the known-matroid MSP and the prophet MSP with pairwise-independent weight distributions.
\end{corollary}

We note that we can also chain the two reductions, and this way relate the online-revealed-matroid MSP to the Prophet MSP with pairwise-independent distributions.
The rest of this section is devoted to the proofs of the reductions. 



\subsection{Proof of Theorem \ref{thm:reduction}}\label{sec:proof_of_MSP_known_to_unknow}

Let $\M$ be the unknown matroid that is revealed to the algorithm in an online fashion and let $f_{\M,\pi}$ be an order-independent randomized OME into $\BigM$. Let $n = \abs{\M}$, $N = \abs{\BigM}$ and $d,k$ be two larger integers (to be specified later) where $k$ is a multiple of $d$.

An instance of the MSP consists of a sequence of weighted elements from $\M$ that are presented to the algorithm in random order. Our goal is to map it on the fly to a random instance of the MSP on  $\BigM_{[k]}$. The main difficulty is, as usual, doing it online and preserving the random order. Our strategy will be to first provide an \emph{offline reduction} which preserves random order and obtains the desired approximation, but can't be implemented online. After that we will provide a \emph{mostly-online implementation} of this reduction, i.e.~a procedure that samples from the same distribution generated by the offline reduction and that with $1-\epsilon$ probability can be implemented online. With the remaining $\epsilon$ probability, the process \emph{raises a flag}. Raising a flag will indicate that from that point on, the reduction can no longer be implemented online. Algorithmically, we will stop the algorithm whenever we raise a flag and obtain zero reward.  Finally, we will show that the probability of raising a flag is very small for large values of $k$ and $d$.

\paragraph{Offline Reduction} 
We will view a weighted element of $\M$ as a pair $(u,w_u)$ with $u \in \M$ and $w_u \in \R_+$.  In the offline reduction, we assume we have access to the entire matroid $\M$ and the entire sequence of weights. Now, we will produce a distribution of instances of $\BigM_{[k]}$ as follows.

For each element $v$ in $\BigM$ sample $k$ different i.i.d.~timestamps $t_{vj}$ for $j \in [k]$ from the $\text{Uniform}([0,1])$ distribution. Those timestamps specify the arrival time of each of the $k$ copies of the elements in $\BigM$ and induce a random ordering over the ground set of $\BigM_{[k]}$. For the matroid $\M$, sample a random embedding $f:\M \rightarrow \BigM$ from the OME. (Since the embedding is order independent, we don't need to know the arrival order of elements in $\M$ to sample such embedding). For each $u \in \M$, pick a random copy of $f(u)$ and set its weight to $w_u$. For the remaining elements, set the weight equal to zero.

This random input is clearly in random order as it is equivalent to starting with $k$ copies of the elements of $\BigM$ where all but one copy has weight zero if that corresponds to an element of $\M$ and randomly permuting those elements. Now, feed this instance to the $\alpha$-competitive algorithm for the MSP on $\BigM_{[k]}$. From the set selected by the algorithm, discard any element with zero weight chosen by the algorithm. The elements with non-zero weight chosen in $\BigM$ correspond to an independent set of $\M$ with the same weight. Hence, in expectation, we select an $\alpha$-approximation to the optimal basis of $\M$.

\paragraph{Mostly-online Implementation} The drawback of the previous reduction is that it can't be implemented online as we are assuming we know everything in advance. We will describe the same sampling procedure in a way that with high probability we can generate the instance as we go. In the sampling procedure, we will also define an event \emph{raise a flag} which will mean that we can't generate that instance online as we learn the structure of the matroid $\M$.

The process will again start by sampling i.i.d.~timestamps $t_{vj}$ for $v \in \BigM$ and $j \in [k]$ from $\text{Uniform}([0,1])$. In addition, we will also sample $n$ additional timestamps from $\text{Uniform}([0,1])$ sort them in increasing order and denote the sorted list by 
$T_1, \hdots, T_n$.

We will now divide the interval $[0,1]$ into intervals $I_i = [\frac{i-1}{d}, \frac{i}{d})$ for $i \in [d]$. If more than one timestamp $T_s$ falls in the same interval $I_i$, we will raise a flag. We will count how many of the timestamps $t_{v,j}$ fall in each interval:
$$X_{vi} = \abs{\{ j \in [k]; t_{vj} \in I_i\}}$$
With that, also define:
$$A_{vi} = \min\left(\frac{k}{d}, X_{vi}\right) \qquad B_{vi} = \max\left(0, X_{vi}- \frac{k}{d}\right)$$
Now, process the elements of the matroid $\M$ according to order $\pi$ (which will be sampled at random). As we process the $s$-th  element $u = \pi(s) \in \M$, we will map it to an element in $v = f_{\M, \pi}(u) \in \BigM$ using the randomized OME. Now, we will apply the following procedure to choose a copy of $v$ in $\BigM_{[k]}$ to assign weight $w_u$:

\begin{itemize}
\item find the interval $I_i$ containing $T_s$.
\item with probability $A_{vi} d/k$, choose one of the $X_{vi}$ timestamps $t_{vj}$ in interval  $I_i$
\item with remaining probability (if any), raise a flag and choose a different interval $I_{i'}$ with probability proportional to $B_{vi'}$ and choose a timestamp $t_{vj}$ in that interval.
\end{itemize}

We assign weight $w_u$ to the element with the chosen timestamp and zero weight to others. 
 Now, we will show the following facts.

\begin{lemma}
The \emph{mostly online implementation} samples sequences with the same probability as the \emph{offline reduction}.
\end{lemma}

\begin{proof}
Observe that if we ignore the weights, the order of the elements of $\BigM$ is the same in both processes since they are determined by the timestamps $t_{vj}$. What we are left to argue is that for each $v$ we select uniformly random timestamp $t_{vj}$
to assign the non-negative weight. For that, observe that since $f_{\M, \pi}$ is order independent, it has the same distribution as if we first sampled a monomorphism $f:\M\rightarrow \BigM$, and then we sampled an independent uniform indexing $\pi$ for the arrival order of the elements in $f(\M)$. This implies that when we assign the timestamps $T_1,\dots, T_n$ according to $\pi$, the resulting distribution is the same as if we assigned i.i.d.~Uniform$[0,1]$ timestamps $T_v$ to each element $v\in f(\M)$, and therefore, the interval $T_v$ lands in is uniformly chosen and independent across elements $v\in f(\M)$.
Now fix a certain timestamp $t_{vj}$ and let $I_i$ be the interval containing it. We will show that the probability that this timestamp is selected is exactly $1/k$.

Consider two cases: either $X_{vi} \leq k/d$ in which case the probability of sampling $t_{vj}$ is the probability that the timestamp $T_v$ is in $I_i$ (which is $1/d$), times the probability we decide to sample a timestamp inside $I_i$ (which is $X_{vi} d/k$), times the probability that out of those, we choose $t_{vj}$ (which is $1/X_{vi}$). The total probability is:
$$\frac{1}{d} \cdot \frac{X_{vi}d}{k} \cdot \frac{1}{X_{vi}} = \frac{1}{k} $$
In the case where $X_{vi} > k/d$, then it is possible that we sample $t_{vj}$ also when $T_v$ is outside $I_i$. The probability that we sample $t_{vj}$ and $T_v$ is in $I_i$ is:
$$\frac{1}{d} \cdot 1 \cdot \frac{1}{X_{vi}} = \frac{1}{d X_{iv}} $$
The probability that we sample when it is outside is the probability that we choose a different interval $I_{i'}$, raise a flag and then move to interval $I_i$, which is:
$$\sum_{i'} \frac{1}{d} \left( 1- \frac{A_{vi'}d}{k}\right) \cdot \frac{B_{vi}}{\sum_{i'} B_{vi'}} \cdot \frac{1}{X_{vi}} = \frac{1}{k}\frac{ \left( k - \sum_{i'}A_{vi'} \right)}{\sum_{i'} B_{vi'}} \cdot \frac{B_{vi}}{X_{vi}}= \frac{1}{k} \frac{B_{vi}}{X_{vi}}$$
because $\sum_{i'} A_{vi'} + \sum_{i'} B_{vi'} = \sum_{i'} X_{vi'} = k$. Taking those two probabilities together, we have:
\[\frac{1}{d X_{iv}} + \frac{1}{k} \frac{B_{vi}}{X_{vi}} = \frac{1}{k}. \qedhere\]
\end{proof}

\begin{lemma}
If no flags were raised, we can produce the instance on the fly as we process $\M$.
\end{lemma}

\begin{proof}
Let $i_1 < \hdots < i_n$ be the indices of the intervals such that $T_s \in I_{i_s}$. Since no flag was raised, then each $T_s$ landed in a different interval and the $s$-th element that arrives from matroid $\M$ is mapped to a copy inside $I_{i_s}$. This enables the following online reduction: once the $s$-th element arrives we can decide the weights of all the elements in intervals $I_{i_{s-1}+1}$ to $I_{i_s}$ and feed to the MSP algorithm for $\BigM_{[k]}$. In this sub-sequence there will be at most one element on non-zero weight which corresponds to the arriving element of $\M$. We can observe if that element was selected in $\BigM_{[k]}$ and if so, we can select it in $\M$.
\end{proof}

\begin{lemma}
For any $n$ and $\epsilon$, there are large enough $k$ and $d$, such that the probability that we raise a flag is at most $\epsilon$.
\end{lemma}

\begin{proof}
The first event in which we raise a flag is when two timestamps $T_s$ land in the same interval. The probability that this happens is at most $n^2/d$. Now, note that for each interval $i$ and each of the $n$ elements $v$ in $\BigM$ that have non-zero weights, we have by the Chernoff bound that:
$$\P\left( X_{vi} \leq (1-\delta) \frac{k}{d} \right) \leq \exp\left( - \frac{\delta^2 k}{ 2d} \right)$$
Hence with probability at most $nd \exp\left( - \frac{\delta^2 k}{ 2d} \right)$, the timestamps $t_{vj}$ are such that the probability we raise a flag when we try to choose a timestamp in the same interval as $T_s$ is more than $\delta n$. Taking the union bound of those events, we get:
$$\frac{n^2}{d} + nd \exp\left( - \frac{\delta^2 k}{ 2d} \right) + \delta n$$
Taking $\delta= \epsilon/(3n)$, $d = 3n^2/\epsilon$ and $k$ large enough, we get that the total probability of raising a flag is at most $\epsilon$.
\end{proof}

Taking those lemmas together, we can conclude the proof of Theorem \ref{thm:reduction}. For that, let $\Alg$ be an $\alpha$-competitive algorithm for $\BigM_{[k]}$ and let $Y$ represent the sequence of the MSP sampled by the offline reduction. Let's represent by $\Alg(Y)$ the weight of the elements selected by $\Alg$ and $\Opt$ the weight of the optimal basis. By the fact that the offline reduction produces an instance in random order, we know that $\E[\Alg(?'Y)] \geq \alpha \Opt$.

Our online reduction, will attempt to construct $Y$ on the fly. If we raise the flag, we will stop the algorithm and pretend we had zero reward. If not, we will continue the reduction and collect $\Alg(Y)$ reward. We will denote by $\Flag$ the event that the flag was raised and by $\NFlag$ its complement. Our total reward will be:

$$\E[\Alg(Y) \cdot \one \{\NFlag\}] = \E[\Alg(Y)] - \E[\Alg(Y) \cdot \one \{\Flag\}] \geq \E[\Alg(Y)] - \Opt \cdot \P[\Flag] \geq (\alpha - \epsilon) \Opt.$$

\subsection{Proof of Theorem~\ref{thm:gen_pw_ind_PS_to_MSP}}
Let $\M$ be the unknown matroid with $|\M| = n$ that is revealed to the algorithm in an online fashion, which admits an order-independent randomized OME into $\BigM$ with $|\BigM| = M$. First, we obtain the following simple reduction that allows us to focus on a special class of the prophet MSP in which each element takes a weight from the set of weights $W$ with $|W|= m = O(n^2)$. In addition, we can restrict the weight distribution such that each element has a distinct weight from the weight class. 

\begin{lemma}\label{lem:reduction_simple_to_unique_weight}
    If there exists an $\alpha$-approximation to the prophet MSP on $\M$ with weight distribution $\D$ supported over the set of weights $W$ with $\rank(\M) = d$, $|\M| = n$, $|W| = O(n^2)$ with $\max_{w\in W} w \leq 1$ and $\E_{\D}[\opt(\M)] \in \left[\frac{1}{16} , 1 \right]$ such that for all $w\in W$, there exists at most one element assigned weight of $w$ with probability one, then there exists an $\left( \frac{\alpha}{256} -  \frac{1}{32d} \right) $-approximation to prophet MSP on matroid $\M$ with any arbitrary weight distribution.
\end{lemma}
The proof of the above lemma simply follows from Sublemma-4.2 from \cite{dughmi2021matroid} that reduces any arbitrary prophet MSP with $O(\log (|\rank(\M)|))$ many weights and $\E_{\D}[\opt(\M)] \in \left[\frac{1}{16} , 1 \right]$. We then add distinct noise of the order of $O\left( \frac{1}{n^2}\right)$ to ensure that the weight of each element is distinct. The full proof of the reduction is delegated to Section~\ref{sec:proof_simple_reduction}.



For simplicity, we let $W = \{w_1 , \dots, w_m\}$ and consider the prophet MSP on matroid $\M$ and weight distribution $\D$ supported over the set of weights $W$ satisfying the conditions from Lemma~\ref{lem:reduction_simple_to_unique_weight}. 

\paragraph{Extending $\BigM$ with Copies}We let  $\BigM_{[m\cdot N]}$ be a matroid with $m\cdot N$ parallel copies of each element of $\BigM$ with $|\BigM| = M$ and integer $N = \Omega\left( 2^{M^2}\right)$.  We divide the set of $N\cdot m$ copies into $m$ sets of size $N$, each part corresponding to weight class $w_i$. We use $N_i$ to denote the set of labels corresponding to weight $w_i$ for all $i\in [m]$ with $|N_i| = N$. We sometimes denote $N_i$ by $[N] = \{1,2,\dots, N\}$ whenever it is clear from the context which $w_i$ we are referring to. In addition, the $\ell$-th copy of the weight class corresponding to weight $w_i$ of element $\vec v \in \BigM$ is denoted as $\vec v^{i,\ell }$.  

\paragraph{Reduction to ``Almost" Pairwise Independent Prophet MSP} We first define the weight distribution $\D^*$ over $\BigM_{[m\cdot N]}$ in Definition~\ref{defn:almost_pw_ind_dist}, which is ``almost" pairwise independent. Then in Theorem~\ref{thm:tv_distance}, we show an existence of exact pairwise independent weight distribution $\tilde \D$ over $\BigM_{[m\cdot N]}$ which is ``close" to the distribution defined in Definition~\ref{defn:almost_pw_ind_dist} in total-variation distance. This allows us to utilize the fact that any algorithm $\A$ can not distinguish between the almost pairwise independent weight distribution $\D^*$ and $\tilde \D$ with high probability. Finally, we complete the proof of Theorem~\ref{thm:gen_pw_ind_PS_to_MSP}. 

We begin by defining the almost pairwise independent weight distribution over $\BigM_{[m\cdot N]}$.  
\begin{definition}[Almost P.W. Independent Distribution]\label{defn:almost_pw_ind_dist}
Consider the weight distribution $\D^*$ over the elements of $\BigM_{[m\cdot N]}$ defined as follows:
    \begin{enumerate}
    \item Given an order independent OMM $f': \M \rightarrow \BigM$, we sample a random automorphism $f'' \in \operatorname{Aut}(\BigM)$ and obtain an order independent matroid morphism $f = f'' \circ f'$. 
    \item For any $\vec v\in \M$ with $w(\vec v) = w_i$, let $\vec u = f(\vec v)$. We sample $\ell \sim \operatorname{Unif}(N_i)$ and assign the weight of $w(\vec u^{i,\ell}) = w_i$. 
    \item We assign the weight of the rest of the elements of $\BigM_{[m\cdot N]}$ to be zero.
\end{enumerate}
\end{definition}
We first observe that when $N$ is much larger than $M$, the distribution in the above definition is almost pairwise independent. To see this, we first observe that for any $i\in [m], \ell \in N_i$ and $\vec u\in \BigM$, $\vec u^{i,\ell}\in \BigM_{[m\cdot N]}$ can potentially either take a weight of $w_i$ or zero. For simplicity, now consider any two distinct elements $\vec u_{i,\ell}, \vec u'_{j,\ell'} \in \BigM_{[m\cdot N]}$ and weight distribution $\D$ such that there always exists a pair of elements $\vec v, \vec v'' \in \M$ that are assigned weights of $w_i, w_j$, respectively. Since $f'$ is a random automorphism $\operatorname{Aut}(\BigM)$, we have $\Pr[w(\vec u_{i,\ell}) = w_i] = \Pr[f(\vec v) = \vec u]\cdot \frac{1}{N} = \frac{1}{M\cdot N}$. On the other hand, we have $$\Pr[w(\vec u_{i,\ell}) = w_i \land w(\vec u_{j,\ell}) = w_j] = \Pr[f'(\vec v) = \vec v \land f'(\vec v') =\vec u']\cdot \frac{1}{N^2} = \frac{1}{M(M-1)}\cdot \frac{1}{N^2},$$
which is close to the product $\Pr[w(\vec u_{i,\ell}) = w_i] \cdot \Pr[w(\vec u_{j,\ell}) = w_j] = \frac{1}{M^2\cdot N^2}$. However, in general, we can not guarantee that $\D$ will always assign weights $w_i, w_j$ to some pair of elements of $\vec v, \vec v'$ of $\M$. In addition, the above argument also fails if we have $i=j$ or $\vec u = \vec u'$ as in both of these cases, $\Pr[w(\vec u_{i,\ell}) = w_i \land w(\vec u_{j,\ell}) = w_j] = 0$. Intuitively, we can circumvent these pairwise correlation issues by taking $N$ large enough as it makes pairwise correlations small enough. More precisely, $$|\Pr[w(\vec u_{i,\ell}) = w_i \land w(\vec u_{j,\ell}) = w_j] -\Pr[w(\vec u_{i,\ell}) = w_i]\cdot  \Pr[w(\vec u_{j,\ell}) = w_j] | = O\left( \frac{1}{M^2\cdot N^2}\right).$$
Using this observation, we prove the following technical theorem. 
\begin{theorem}\label{thm:tv_distance}
For $N = \Omega\left(2^{M^2} \right)$, let weight distribution $\D^*$ over $\BigM_{[m\cdot N]}$ be defined as in Definition~\ref{defn:almost_pw_ind_dist}, then there exist a pairwise-independent weight distribution $\tilde \D$ over $\BigM_{[m\cdot N]}$ such that $\operatorname{TV}_{\D^*, \tilde \D} \leq O \left ( \frac{m^3}{M}\right)$.
\end{theorem}

The proof of the above theorem is highly technical and constructs an explicit $\tilde \D$ by a sequence of small perturbations to $\D^*$. For the sake of the uninterrupted flow of the presentation, we delegate it to Appendix~\ref{sec:techncial_proof}. We emphasize that the choice of $N = \Omega(2^{M^2})$ is required due to the limitations of the techniques developed to prove Theorem~\ref{thm:tv_distance}. We conjecture that one can prove the similar theorem for $N = \Omega(\operatorname{Poly}(M))$, which we leave as an intriguing technical open problem. 

\paragraph{Proof of Theorem~\ref{thm:gen_pw_ind_PS_to_MSP}}We now complete the proof of the main theorem. 

\begin{proof}[Proof of Theorem~\ref{thm:gen_pw_ind_PS_to_MSP}]
To prove the main theorem, we first prove the following: if there exists an $\alpha$-approximate algorithm to the prophet MSP instance with a pairwise independent weight distribution for the matroid $\BigM_{[m\cdot N]}$, then there exists an $(\alpha - o(1))$-approximate algorithm for the prophet MSP instance for the matroid $\M$ with weight distribution $\D$ supported over the set of weights $W$ with $\rank(\M) = d$, $|\M| = n$ and $|W| = O(n^2)$ with $\max_{w\in W} w \leq 1$ and $\E_{\D}[\opt(\M)] \in \left[\frac{1}{16}, 1 \right]$ such that for all $w\in W$, there exists at most one element assigned weight of $w$ with probability one. Combining this with Lemma~\ref{lem:reduction_simple_to_unique_weight}, we will conclude the proof of the Theorem~\ref{thm:gen_pw_ind_PS_to_MSP}

Given an order-independent matroid morphism $f': \M \rightarrow \BigM$, we sample a random automorphism $f'' \in \operatorname{Aut}(\BigM)$ and obtain an order-independent matroid morphism $\tilde f = f'' \circ f'$. Given $\tilde f:\M \rightarrow \BigM$, we obtain $f: \M \rightarrow \BigM_{[m]}$ that maps each $\vec v\in \M$ to $\vec u^i \in \BigM_{[m]}$ iff $\tilde f(\vec v) = \vec u$ and $w(\vec v) = w_i$.
We then consider $\BigM_{[m\cdot N]}$, i.e. matroid $\BigM_{[m]}$ with $N$ many copies of each element. 

We first consider an offline reduction as follows: for all $\vec v^i \in \BigM_{[m]}$, sample $N$ many independent arrival times from $\operatorname{Unif}[0,1]$ denoting the uniformly random arrival times of elements of $\BigM_{[m\cdot N]}$. Given any $\vec v \in \M$ and pair $(\vec v, w(\vec v) = w_i)$, with $\vec u^i = f(\vec v)$, let $w(\vec u^{i,\ell}) = w_i $ uniformly random from $\ell \in [N]$. 

Since $f$ does not require any information about the arrival order, the above-described offline reduction is a valid instance of matroid prophet secretary over $\BigM_{[m\cdot N]}$. In addition, the induced weight assignment over $\BigM$ due to offline reduction is identical to the distribution $\D^*$ defined in Definition~\ref{defn:almost_pw_ind_dist}. 

Given the uniformly random arrival of elements of $\M$, we construct an ``almost online implementation" of the above offline reduction similar to the proof of 
Theorem~\ref{thm:reduction}. We let $N$ be large enough ($\Omega(n^2/\eps)$) such that the probability of ``almost online implementation" raising a flag is at most $\varepsilon$. In fact, for the proof of Theorem~\ref{thm:tv_distance}, we let $N = \Omega(2^{M^2})$ which satisfies the required condition.

We let $\A$ be an $\alpha$-approximate algorithm for the pairwise-independent prophet MSP on matroid $\BigM_{[m\cdot N]}$. Since there exists a pairwise independent weight distribution $\tilde \D$ over $\BigM$ within the total variation distance of $O \left( \frac{m^3}{M}\right)$, the algorithm $\A$ can not distinguish the weight distribution $\D^*$ from $\tilde \D$ with probability at least $1 - O \left( \frac{m^2}{M}\right)$. 

Let $\mathcal E$ be the event when the algorithm $\A$ can not distinguish between $\D^*$ and $\tilde D$. We note that when event $\mathcal E$ does not hold, the offline optimal can be bounded by $$\E [\opt(\BigM_{m\cdot N})\mid \mathcal E^c]\leq \rank(\BigM_{[m\cdot N]})\cdot \max_{w_i \in W} {w_i} \leq n\cdot 1 = n,$$
where the second inequality follows because $\rank(\BigM_{[m\cdot N]})= n$ and $w_i \leq 1$ for all $i\in [m\cdot N]$.


Now, let $S$ be the selected set of elements of $\BigM$ by $\A$ w.r.t. weight distribution $\D^*$. We can bound,
$$\begin{aligned}
    \E [w(S)] &= \E [w(S) \mid \mathcal E ]\cdot \Pr[\mathcal E] + \E [w(S) \mid \mathcal E^c]  \cdot \Pr[\mathcal E^c]\\
    &\leq \E [w(S) \mid \mathcal E ] + \E[\opt(\BigM_{[M\cdot m]})\mid \mathcal E^c] \cdot \frac{m^3}{M}\\
    &\leq  \E [w(S) \mid \mathcal E ] +  \frac{m^3\cdot n}{M}.
\end{aligned}$$
Above, the first inequality holds because $\Pr[\mathcal E^c] \leq \frac{m^2}{M}$ and the second inequality holds because $\E [\opt(\BigM_{[m\cdot N]})\mid \mathcal E^c]\leq n$. Due to our reduction, the performance of the algorithm on the original Prophet MSP instance $\F'$ is lower bounded by $\E[w(S)\mid \mathcal E]$, next we lower bound the expectation $\E[w(S)\mid \mathcal E]$,
$$\begin{aligned}
    \E [w(S) \mid \mathcal E ] &\geq \E[w(S)] - \frac{m^3\cdot n}{M}\\
    &\geq \alpha \cdot \E_{\tilde \D}[\opt(\BigM_{[m\cdot N]})]- o(1) \cdot \opt(\M) \\
    &\geq \opt(\M) \cdot \left( \alpha - \frac{\alpha\cdot m^3}{M} \right) - o(1)\cdot \opt(\M) \geq  \left( \alpha - o(1) \right)\cdot \opt(\M) .
\end{aligned}$$
Above, the second inequality holds because $\A$ is an $\alpha$ approximate algorithm for Prophet MSP with pairwise independent prior and $\frac{m^3\cdot n}{M} = O(m^4 /2^{m}) =  o(1)$. The third inequality holds because $\operatorname{TV}_{\D^*, \tilde \D} \leq O(m^3/M)$.   
Finally, since the performance of the reduction on the original matroid secretary is $$\E[w(S) \mid \mathcal E] \cdot \Pr[\text{Reduction does not Flag}] \geq (1-\eps) \cdot \left( \alpha - o(1) \right) \cdot \opt(\M).$$ This concludes the proof. 
\end{proof}
\color{black}

\section{OMEs Beyond Binary Matroids}

We next explore OMEs beyond binary matroids. In Section~\ref{sec:graphic_regular}, we show that there cannot be an OME that embeds graphical matroids into graphical matroids or, more generally, regular matroids. Then, in Section~\ref{sec:laminar}, we give an OME that embeds laminar matroids into laminar matroids. Finally, in Section~\ref{sec:bounded_rank}, we show that there is no universal host matroid, that allows embedding of all matroids on $n$ elements of a given rank.

\subsection{Graphic and Regular Matroids}
\label{sec:graphic_regular}

For graphic matroids, there exists an elegant $2e$-approximation algorithm for the MSP in the known-matroid case by Korula-P\'{a}l \cite{korula2009algorithms}. Their algorithm assumes that when an edge arrives, the algorithm learns the pair of vertices it connects as well as the weight. In other words, the algorithm processes in each step $(u,v), w_{uv}$ where $u$ and $v$ are vertices. Even though the full graph is not known in advance, there is enough information about the graph structure to randomly decompose the problem into instances of the single-item secretary problem.

In the online-revealed matroid case, the algorithm has only access to an oracle that tells which subsets of previously arrived edges are independent (contain no cycles). For example, if $3$ edges arrive and the algorithm knows that they are all independent, it is impossible for the algorithm to know if they form a path, a star or if they share not endpoints (recall Figure \ref{fig:graphic_matroid}). Korula-P\'{a}l heavily relies on having vertex information and doesn't easily extend to this model. In fact, we are not aware of any $O(1)$-approximation algorithm for the MSP for graphic matroids in the online-revealed-matroid setting.

In the remainder of this section we investigate whether we can obtain such an algorithm using a OME. A natural idea is to try to construct an embedding where the host matroid $\BigM$ is itself a graphic matroid. If one could do that, it would be possible to combine the reduction in Section \ref{sec:reduction} with the Korula-P\'{a}l algorithm to obtain an algorihtm for the graphic MSP in the online-revealed-matroid setting. Unfortunately, we show below that no such embedding exists:

\begin{theorem}\label{thm:no_big_graphic}
If $\C$ is the class of graphic matroids, there is no online matroid embedding into a host matroid $\BigM$ where $\BigM$ is also graphic.
\end{theorem}

\begin{proof}
Assume that such embedding exists and let $G = (V,E)$ be the graph representing $\BigM$. Now, consider two graphic matroids $\M_1$ and $\M_2$ represented respectively by the graphs in the left and right of Figure \ref{fig:graphic_matroid_2}. Observe that when restricted to $\{a,b,c\}$, the matroids are identical since every non-empty subset of elements is independent.
When an OMM observes the restriction to those three elements, it has no way to know in which matroid we are in, so it needs to map those three edges to the same edges of the graph $G$ representing $\BigM$. Let's denote those edges by $a',b',c'$. Let also $d',e',f'$ be the edges in $G$ that the edges of the left matroid are mapped to and let $g'$ be the edge that the $g$ edge in the right matroid is mapped to. Note that $\{a',b',d'\}$, $\{b',c',e'\}$ and $\{a', c', f'\}$ must form cycles in $G$. The only way that this is possible while keeping $\{a',b',c'\}$ independent is if edges $a',b',c'$ all share an endpoint. However, edges $a', b', c', g'$ must form a cycle in $G$ as well, which is not possible if the first three edges share an endpoint.
\end{proof}

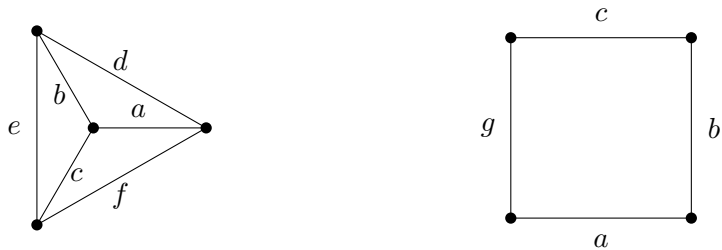
\begin{figure}[h]
\centering
\begin{tikzpicture}[scale=1.5, inner sep=1.5pt] 
 \node[circle,fill] at (0,0) {};
 \node[circle,fill] at (1,0) {};
 \node[circle,fill] at (-.5,.86) {};
 \node[circle,fill] at (-.5,-.86) {};

 \draw (0,0)--(1,0);
 \draw (0,0)--(-.5,.86) ;
 \draw (0,0)--(-.5,-.86);
 
 \draw (1,0)--(-.5,.86)--(-.5,-.86)--(1,0);
 \node at (.4,.15) {$a$};
 \node at (-.3,.3) {$b$};
 \node at (-.15,-.42) {$c$};
 \node at (.24,.6) {$d$};
 \node at (.24,-.6) {$f$};
 \node at (-.7,0) {$e$};

 \begin{scope}[xshift=4cm]

 \node[circle,fill] at (-.3,.8) {};
 \node[circle,fill] at (-.3,-.8) {};
 \node[circle,fill] at (1.3,.8) {};
 \node[circle,fill] at (1.3,-.8) {};
 \draw (-.3,.8)--(1.3,.8)--(1.3,-.8)--(-.3,-.8)--(-.3,.8);
 \node at (.5,-1) {$a$};
 \node at (1.5,0) {$b$};
 \node at (.5,1) {$c$};
 \node at (-.5,0) {$g$};
 \end{scope}
 
\end{tikzpicture}
\caption{Two graphic matroids whose restriction to $\{a,b,c\}$ coincide.}
\label{fig:graphic_matroid_2}
\end{figure}

Graphic matroids are a special case of regular matroids, for which there exists an $O(1)$-competitive algorihtm by Dinitz-Kortsarz \cite{dinitz2014matroid}. It is then tempting to construct an embedding from graphic into a host matroid $\BigM$ that is regular. However, that is again not possible:

\begin{lemma}
If a host matroid $\BigM$ is regular and admits an online matroid embedding of all rank $n$ graphic matroids, then $\BigM$ must also admit an online matroid embedding of all rank $n$ binary matroids.
\end{lemma}

\begin{proof}
We will show that $\BigM$ contains an isomorphic copy of $\F_2^n$. We first define for each $S \subseteq [n]$ a graphic $\M_S$ with ground set $[n+1]$ represented by the graph where the only cycle is formed by the edges with labels in $S \cup \{n+1\}$. Equivalently, the rank function is given by 
$$\rank_{\M_S}(T) = \left\{
\begin{aligned}
  & \abs{T}-1 & & \text{for } T \supseteq S \cup \{n+1\} \\
  & \abs{T} & & \text{otherwise }
\end{aligned} \right.$$
Let also $\M'$ be the matroid on $[n]$ such that $\rank_{\M'}(T) = \abs{T}$ for all subsets $T$. Assuming all ground serts are ordered according to the labels, note that  $\M'$ is a prefix of all matroids $\M_S$. Now, let $b_1, \hdots, b_n \in \BigM$ be the image of the ground set of $\M'$ by the online embedding. Since $\BigM$ is regular, it admits an $\F_2$-representation, so we can think of $b_1, \hdots, b_n$ as linearly independent vectors in $\F_2^n$.

Since this is an online embedding, it can be extended to a morphism $\M_S \rightarrow \BigM$ for each $S \subseteq [n]$. So, for each $S$, let $b_S$ be the element in $\BigM$ $[n+1]$ maps to. If we view $b_S$ as an $\F_2$-vector, we must have $b_S = \sum_{i \in S} b_i$ in $\F_2^n$. As a consequence, the vectors $\{b_S; S \subseteq [n]\}$ form an isomorphic copy of $\F_2^n$.
\end{proof}

We can now derive the following theorem as as corollary:

\begin{theorem}\label{thm:no_big_regular}
If $\C$ is the class of graphic matroids, there is no online matroid embedding into a host matroid $\BigM$ where $\BigM$ is regular.
\end{theorem}

\begin{proof}
By the previous lemma, if such embedding exists then $\BigM$ must contain an isomorphic copy of $\F_2^n$. However, $\F_2^n$ contains an isomorphic copy of the Fano plane which is not representable over $\F_3$ \cite{tutte1958homotopy}. Hence $\BigM$ can't be regular.
\end{proof}

The reader will notice that Theorem \ref{thm:no_big_graphic} can be derived as a trivial corollary of Theorem \ref{thm:no_big_regular} since graphic matroids are regular. Nevertheless, we find the more direct proof of Theorem \ref{thm:no_big_graphic} enlightening and opted to keep it.

\subsection{Laminar Matroids}
\label{sec:laminar}

Another important class of matroids that admits an online embedding is the class of laminar matroids.
We begin this section by recalling the definition of a laminar matroid and proving a useful structural
lemma about them. We then present an OMM for the class of laminar matroids with at most $n$ elements, 
using a host matroid $\BigM$ which is a complete linear matroid of rank $n$ over any field with 
sufficiently many elements. 


\begin{lemma} \label{lem:laminar-monomorphism}
    If $f : \M \to \N$ is a matroid monomorphism 
    and $\N$ is laminar, then $\M$ is laminar as well.
\end{lemma}
\begin{proof}
    Denote the ground sets of $\M, \N$ by $M,N$, respectively.
    If $\mathcal{A}$ is a laminar family of subsets of $N$ 
    such that $\N$ is a laminar matroid with respect to the
    function $c : \mathcal{A} \to \mathbb{Z}_+$, then the family
    of sets $f^{-1}(\mathcal{A})$ consisting of the sets
    $f^{-1}(A)$ for each $A \in \mathcal{A}$ is a laminar
    family, and $\M$ is a laminar matroid with respect to
    the function $\tilde{c} : f^{-1}(\mathcal{A}) \to \mathbb{Z}_+$
    defined by 
    \[
        \tilde{c}(B) = \min \{ c(A); A \in \mathcal{A} \mbox{ and } B = f^{-1}(A)\}. \qedhere
    \]
\end{proof}

Let us also recall the definitions of the span and flat in a matroid.
Given a subset $S$ of the ground set of a matroid $\M$, we define the $\span_\M(S) = \{ x \in \M; \rank_\M(S \cup \{x\}) = \rank_\M(S)\}$. From the properties of the rank function, this is the equivalent to the maximal set containing $S$ that has the same rank as $S$. We say that a subset $S$ is a flat if $S = \span(S)$. For the complete $K$-representable matroid $K^n$, the $\span$ coincides with the notion of the linear span of vector fields and the flats are linear subspaces.


The span of laminar matroids has the following useful property:

\begin{theorem}[Fife and Oxley~\cite{FO17}]\label{thm:span_laminar}
A matroid is laminar if and only if, for all circuits $C_1$ and $C_2$ with $C_1\cap C_2 \neq \emptyset$, either $\span(C_1)\subseteq \span(C_2)$ or $\span(C_2)\subseteq \span(C_1)$.
\end{theorem}

\begin{theorem}
    \label{thm:laminar}
    Let $\C$ be the class of laminar matroids of at most $n$ elements and $\F$ be a field with at least $2^n$ elements. Then there is an OMM from $\C$ into $\F^n$.
\end{theorem}

\begin{proof}
The OMM is defined inductively. If $\M$ is an empty matroid
then $f_{\M,\pi}$ is the trivial morphism from the empty matroid
to $\BigM = \F^n$. Otherwise let $n = |\M|$, let $(\M',\pi')$ denote the
restriction of $(\M,\pi)$ to $[n-1]$, let $g$ denote the
morphism $f_{\M',\pi'}$, and let $u = \pi(n)$.
The morphism $f = f_{\M,\pi}$ is defined as follows. For
$u' \neq u$ we set $f(u') = g(u')$. 
If $u$ doesn't belong to any circuit with previously arrived elements, define $f(u)$ to be any element outside the linear span of $f(\M')$. If $u$ is in some circuit, define $A \subseteq \M$ to be the intersection of $\span(C)$ over all circuits $C$ in $\M$ containing $u$. By Theorem \ref{thm:span_laminar}, $A$ is the span of some circuit in $\M$ and hence a flat. Now, choose $f(u)$ to be an element in the linear subspace $V = \span_{\F^n}(f(A \setminus \{u\}))$ that is not contained in $V_S = \span_{\F^n}(f(S))$ for every subset $S$ of $\M$ such that $S \cup \{u\}$ is independent on $\M$.

To justify that $f(u)$ is well-defined, we must
argue that the linear subspace $V$
contains at least one element that is not in
$V_S$ for every subset $S$ such that $S \cup \{u\}$ is independent.
We will use a counting argument that consists
of showing that $V_S \cap V$ is a proper linear
subspace of $V$ for every such $S$, and then
observing that a vector space over a field with
at least $2^n$ elements cannot be expressed as the
union of $2^n$ or fewer proper linear subspaces.
To show that $V_S \cap V$ is a proper linear
subspace of $V$ we argue by contradiction.
Let $C$ be a circuit containing $u$ such that $A = \span_\M(C)$.
As $S \cup \{u\}$ is independent, there must be
some $u' \in C \setminus \{u\}$ such that
$S \cup \{u'\}$ is also independent. Then,
since $g : \M' \to \BigM$ is a matroid
morphism, $g(u')$ is linearly independent
of $g(S)$. In particular, $g(u') \in V \setminus V_S$ and
hence $V_S \cap V$ is a proper linear subspace of $V$ as claimed.

Having justified that $f$ is well-defined,
we must show that it is a matroid morphism.
Consider any set $S$ in $\M$. To show
$\rank_{\BigM}(f(S)) = \rank_{\M}(S)$ we
will proceed by case analysis.
\begin{enumerate}
\item If $u \not\in S$ then
  $f(S) = g(S)$ and we use the
  induction hypothesis that
  $g = f_{\M',\pi'}$ is a matroid morphism.
\item If $u \in S$ and $\rank_{\M}(S) =
  \rank_{\M}(S \setminus \{u\}) + 1$
  then let $B$ be a maximal independent
  subset of $S \setminus \{u\}$. Since
  $g$ is a matroid morphism, every
  $u' \in S \setminus \{u\}$ lies in
  the linear span of $g(B)$. Since
  $B \cup \{u\}$ is independent in $\M$,
  by construction $f(u)$ lies outside
  the linear span of $g(B)$. Hence,
  $\rank_{\BigM}(S) = \rank_{\BigM}(B \cup \{u\}) =
  |B| + 1 = \rank_{\M}(S)$.
\item If $u \in S$ and $\rank_{\M}(S) =
  \rank_{\M}(S \setminus \{u\})$ then
  there is a circuit $C \subseteq S$
  containing $u$. The  $\span(C \setminus \{u\})$
  is a flat $A'$ that contains $u$,
  so $A'$ must be a superset of $A$ by Theorem \ref{thm:span_laminar}.
  Then we have the following chain
  of containments.
  \[
  \span_{\F^n}(g((S \setminus \{u\}))
  \supseteq
  \span_{\F^n}(g((C \setminus \{u\}))
  = \span_{\F^n}(g((A')) \supseteq \span_{\F^n}(g(A \setminus \{u\})) .
  \]
  Since $f(u)$, by construction, belongs to $\span_{\F^n}(g(A \setminus \{u\}))$,
  it belongs to $\span_{\F^n}(g(S \setminus \{u\}))$ and therefore
  \[
  \rank_{\F^n}(f(S)) =
  \rank_{\F^n}(g(S \setminus \{u\})) =
  \rank_{\M'}(S \setminus \{u\}) =
  \rank_{\M}(S).
  \qedhere\] 
\end{enumerate}
\end{proof}


Similarly to the situation of graphic matroids, we can embed laminar matroids into a large linear matroid, but there is no embedding for which $\BigM$ is also laminar:

\begin{theorem}
    If $\C$ is the class of laminar matroids of at most $n$ elements, there is no online matroid embedding into a host matroid $\BigM$ where $\BigM$ itself is laminar.
\end{theorem}

\begin{proof}
Given a set $S$ let $\U_{n,r}(S)$ denote the uniform matroid with $n$ elements and rank $r$ defined on ground set $S$. Using this notation, consider the following three matroids on ground set $\{a,b,c,d\}$:
$$\U_{3,2}(\{a,b,d\}) \oplus \U_{1,1}(\{c\}) \qquad 
\U_{3,2}(\{a,c,d\}) \oplus \U_{1,1}(\{b\})  \qquad 
\U_{3,2}(\{b,c,d\}) \oplus \U_{1,1}(\{a\}) $$
It is easy to check that all three of those matroids are laminar. Moreover, the restriction to $\{a,b,c\}$ is the matroid $\U_{3,3}(\{a,b,c\})$. Assume now that there is an OMM into a laminar matroid $\BigM$ and let $a', b', c'$ be the elements it maps to. There must be some set in the laminar family of $\BigM$ that separates those elements, otherwise we can't extend it to the three matroids above. For example, in the first matroid we can't swap the roles of $a$ and $c$ preserving rank. Finally, there is at most one two-element subset of $\{a',b',c'\}$ that can be formed by intersecting one of the sets in the laminar family with $\{a',b',c'\}$. 
If the two-element subset that spans $d$ is not that one, then you have no way of extending the embedding to include $d$.
\end{proof}

\subsection{Matroids of Bounded Rank}\label{sec:bounded_rank}

Finally, we show that there is no universal host matroid $\BigM$ such that all matroids of at most $n$ elements admit an OMM into $\BigM$. In fact, we show that even if we restrict to matroids of rank $3$ this is not possible. The geometric intuition is that once we reach rank $3$, we start being able to represent finite projective planes, where elements that haven't arrived yet impose non-trivial constraints on the already arrived elements. Before we get there, it is useful to analyze ranks~$1$ and~$2$.

Let $\CM_{n,r}$ be the matroids of rank at most $r$ defined on at most $n$ elements. We will show that for $r=1,2$ it is trivial to construct an online matroid embedding. As usual, we will construct an OMM and convert it to an OME using Lemma \ref{lemma:omm_to_ome}.

\begin{lemma} There is an online matroid morphism from $\CM_{n,1}$ into $\U_{1,1} \oplus T$.
\end{lemma}

\begin{proof}
The matroid $\U_{1,1} \oplus T$ has ground set $\{0,1\}$ and the only independent set is $\{1\}$. Given a matroid $\M \in \CM_{n,1}$ and an element $e$ of the ground set of $\M$ we map it to $1$ if $\rank_\M(\{e\}) = 1 $ and to $0$ otherwise. The mapping is constructed online since it only depends on the rank of a single-set element consisting of the element we are processing. It is also simple to check that it is a morphism since $\M$ has rank $1$.
\end{proof}

\begin{lemma} There is an online matroid morphism from $\CM_{n,2}$ into $\U_{n,2} \oplus T$.
\end{lemma}

\begin{proof}
The matroid $\U_{n,2} \oplus T$ has ground set $\{0,1,2, \hdots, n\}$ and the independent sets are subsets of at most $2$ elements that don't contain zero. To construct an embedding $f: \M \rightarrow \U_{n,2} \oplus T$, when we process an element $e$, we first check if $\rank_\M(\{e\}) = 0$. If so, we map it to $0$. Otherwise, for every element $a$ processed before $e$, we check if $\rank_\M(\{a,e\}) = 1$. If so, we map $f(e) = f(a)$. Otherwise, we map $e$ to the first unused index in $[n]$.
\end{proof}

Once we reach rank $3$, the situation becomes a lot more interesting, as there exist many non-trivial matroids like finite projective planes (e.g. Fano plane). The richness of the space of rank $3$ matroids will also imply that an online matroid embedding no longer exists.

Before we prove it formally, let's give some geometric intuition. Given points on the plane $\R^2$, we can define a matroid of rank $3$ as follows: (i) every set of one point is independent (ii) every pair of different points is independent; (iii) every triple of points is independent iff it they are not collinear; (iv) no other set is independent.

Now, consider the following online problem: we are presented with labels $a,b,c,d,e,f,g$ in this order. Each label represents a point, but we are not told which point it is. Instead, we are told the dependency relation between them. As they arrive, we are asked to map each label to a $\R^2$ such that the dependency relations are satisfied.

Take the following instance of this problem: the first $6$ points $a,b,c,d,e,f$ arrive, we are told that all triples are independent, i.e., neither of them is collinear. Figure \ref{fig:projective1} shows two possible ways to place those points in the plane satisfying those dependencies. In the first arrangement, the segments $[a,b], [c,d], [e,f]$ all intersect in a single point. In the second arrangement, they don't.

Now, the last label $g$ arrives. If the algorithm chose the arrangement on the left, we can ask to place $g$ such that $\{a,g,b\}, \{c,g,d\}$ are dependent, but $\{e,f,g\}$ are independent. There is no way to satisfy those dependencies in the figure on the left, since $g$ must be in the intersection of the $[a,b]$ and $[c,d]$ but the only way to do so is by also being in the line $[e,f]$. On the other hand, if the algorithm chose the arrangement on the right in Figure \ref{fig:projective1}, we can ask to place $g$ such that $\{a,g,b\}, \{c,g,d\}$ and $\{e,f,g\}$ are dependent. There is no way to satisfy those dependencies in the figure on the right since the lines $[a,b]$, $[c,d]$ and $[e,f]$ must meet on the same point.

Those two dependencies are satisfiable offline, as we can see in Figure \ref{fig:projective2}. The difficulty is that the dependencies on points that have not arrived pose constraints on the relative position of points that have already arrived.

\begin{figure}[h]
\centering
\begin{tikzpicture}[scale=1.2, inner sep=1.5pt] 
 \node[circle,fill] at (0,0) {};
 \node[circle,fill] at (2,0) {};
 \node[circle,fill] at (1,1.4) {};

 \node[circle,fill] at (1,-.3) {};
 \node[circle,fill] at (0.2,0.8) {};
 \node[circle,fill] at (1.8,0.8) {};

 \node at (-.2,-.1) {$a$};
 \node at (2.2,-.1) {$e$};
 \node at (1,1.6) {$c$};

 \node at (2,0.9) {$b$};
 \node at (0,0.9) {$f$};
 \node at (1,-.5) {$d$};

 \draw (0,0)--(2,0);
 \draw (2,0)--(1,1.4);
 \draw (0,0)--(1,1.4);
 \draw (1,-.3)--(1,1.4);
 \draw (0.2,0.8)--(2,0);
 \draw (1.8,0.8)--(0,0);

 \begin{scope}[xshift=5cm]

 \node[circle,fill] at (0,0) {};
 \node[circle,fill] at (2,0) {};
 \node[circle,fill] at (1,1.4) {};

 \node[circle,fill] at (.5,-.3) {};
 \node[circle,fill] at (0.2,0.8) {};
 \node[circle,fill] at (1.8,0.8) {};

 \node at (-.2,-.1) {$a$};
 \node at (2.2,-.1) {$e$};
 \node at (1,1.6) {$c$};

 \node at (2,0.9) {$b$};
 \node at (0,0.9) {$f$};
 \node at (.5,-.5) {$d$};

 \draw (0,0)--(2,0);
 \draw (2,0)--(1,1.4);
 \draw (0,0)--(1,1.4);
 \draw (.5,-.3)--(1,1.4);
 \draw (0.2,0.8)--(2,0);
 \draw (1.8,0.8)--(0,0);
 \end{scope}
 
\end{tikzpicture}
\caption{Two sets of points that induce the same matroid on $\{a,b,c,d,e,f\}$}
\label{fig:projective1}
\end{figure}
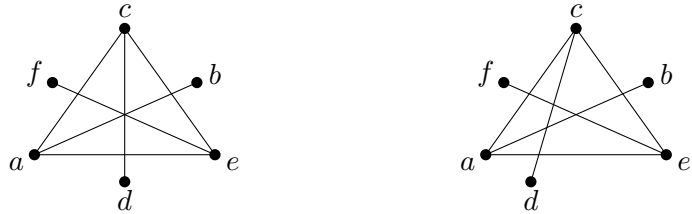

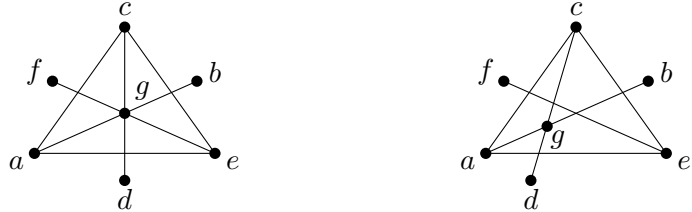
\begin{figure}[h]
\centering
\begin{tikzpicture}[scale=1.2, inner sep=1.5pt] 
 \node[circle,fill] at (0,0) {};
 \node[circle,fill] at (2,0) {};
 \node[circle,fill] at (1,1.4) {};

 \node[circle,fill] at (1,-.3) {};
 \node[circle,fill] at (0.2,0.8) {};
 \node[circle,fill] at (1.8,0.8) {};

 \node[circle,fill] at (1,0.444) {};
\node at (1.2,0.7) {$g$};

 \node at (-.2,-.1) {$a$};
 \node at (2.2,-.1) {$e$};
 \node at (1,1.6) {$c$};

 \node at (2,0.9) {$b$};
 \node at (0,0.9) {$f$};
 \node at (1,-.5) {$d$};

 \draw (0,0)--(2,0);
 \draw (2,0)--(1,1.4);
 \draw (0,0)--(1,1.4);
 \draw (1,-.3)--(1,1.4);
 \draw (0.2,0.8)--(2,0);
 \draw (1.8,0.8)--(0,0);

 \begin{scope}[xshift=5cm]

 \node[circle,fill] at (0,0) {};
 \node[circle,fill] at (2,0) {};
 \node[circle,fill] at (1,1.4) {};

 \node[circle,fill] at (.5,-.3) {};
 \node[circle,fill] at (0.2,0.8) {};
 \node[circle,fill] at (1.8,0.8) {};

 \node[circle,fill] at (.68,.3) {};
 \node at (.8,0.15) {$g$};

 \node at (-.2,-.1) {$a$};
 \node at (2.2,-.1) {$e$};
 \node at (1,1.6) {$c$};

 \node at (2,0.9) {$b$};
 \node at (0,0.9) {$f$};
 \node at (.5,-.5) {$d$};

 \draw (0,0)--(2,0);
 \draw (2,0)--(1,1.4);
 \draw (0,0)--(1,1.4);
 \draw (.5,-.3)--(1,1.4);
 \draw (0.2,0.8)--(2,0);
 \draw (1.8,0.8)--(0,0);
 \end{scope}
 
\end{tikzpicture}
\caption{Once the dependency of point $g$ with respect the remaining points is specified, this constrains the possible geometric arrangements of points $\{a,b,c,d,e,f\}$.}
\label{fig:projective2}
\end{figure}

Our goal with the previous discussion is to provide an intuition for the following proof. Note that while embedding a matroid of rank $3$, we are not restricted to $\R^2$. We could in principle embed it in a matroid that is not representable over any field. 
However, geometric intuition can now be turned into a combinatorial proof of the following statement:

\begin{theorem}\label{thm:no_big_rank_3}
There is no host matroid $\BigM$ for which there is an online matroid morphism from $\CM_{n,3}$ into $\BigM$.
\end{theorem}

\begin{proof} Define matroids $\M_1$ and $\M_2$ on elements $\{a,b,c,d,e,f,g\}$ represented by the points in Figure \ref{fig:projective2}. In those pictures, a triple of elements is independent iff the corresponding points are not collinear.

Those points induce the same matroid on the first $6$ points, which is simply the uniform matroid $\U_{6,3}$, but the matroids differ when we one considers element $g$.

\begin{itemize}
\item in $\M_1$, the sets $\{a,b,g\}$, $\{c,d,g\}$ and $\{e,f,g\}$ are dependent
\item in $\M_2$, the sets $\{a,b,g\}$ and $\{c,d,g\}$ are dependent but $\{e,f,g\}$ is independent
\end{itemize}

Now suppose $\BigM$ is a matroid and we’re trying to construct an online embedding of a matroid $\M$ into $\BigM$, where $\M$ is a rank-3 matroid that could either be $\M_1$ or $\M_2$. When the first six elements of $\M$ arrive, we have no way of distinguishing whether the input sequence is going to be $\M_1$ or $\M_2$. The online embedding algorithm chooses some function $h$ mapping $\{a,b,c,d,e,f\}$ to a six-element subset of $\BigM$. Denote the images of $a,b,c,d,e,f$ in $\BigM$ by capital letters, for example $h(a)=A$. Now suppose there are two different extensions of $h$ to the domain ${a,b,c,d,e,f,g}$, denoted by $h_1$ and $h_2$, such that $h_i$ is an  embedding of $\M_i$ into $\BigM$ for each $i$. Let $G_1 = h_1(g)$ and $G_2 = h_2(G)$. Then the following must hold:

\begin{itemize}
\item the sets $\{A,B,G_1\}$, $\{C,D,G_1\}$ and $\{E,F,G_1\}$ are dependent in $\BigM$
\item the sets $\{A,B,G_2\}$ and $\{C,D,G_2\}$ are dependent but $\{E,F,G_2\}$ is independent in $\BigM$
\end{itemize}

The sets $\{A,B,G1,G2\}$ and $\{C,D,G1,G2\}$ both have rank $2$ in $\BigM$, whereas their union has rank $3$ since it contains the rank-3 set $\{A,B,C,D\}$. By submodularity, the set $\{G1,G2\}$ must have rank 1. Again by submodularity, 
$$\rank_{\BigM}(\{E,F,G1,G2\}) \leq \rank_{\BigM}(\{E,F,G1\}) + \rank_{\BigM}(\{G1,G2\}) - \rank_{\BigM}(\{G1\})$$
which evalutes to $2 + 1 - 1 = 2$. This contradicts the fact that $\{E,F,G2\}$ has rank $3$.
\end{proof}

\begin{corollary}\label{cor:representable_counterexample}
For every field $K$ of characteristic $p \geq 7$, there is no host matroid $\BigM$ for which there is an online matroid embedding from all $K$-representable matroids into $\BigM$.    
\end{corollary}

\begin{proof}
Observe that the matroids $\M_1$ and $\M_2$ in the proof of Theorem \ref{thm:no_big_rank_3} are representable in every field of characteristic at least $7$. For example, they can be represented respectively by the columns of the following matrices:
$$\begin{bmatrix}
0 & 2 & 1 & 1 & 2 & 0 & 1\\
1 & 1 & 2 & 0 & 0 & 2 & 1\\
2 & 0 & 0 & 2 & 1 & 1 & 1
\end{bmatrix}
\quad \text{ and } \quad
\begin{bmatrix}
0 & 2 & 1 & 1 & 2 & 0 & 1\\
1 & 1 & 2 & 0 & 0 & 1 & 1\\
2 & 0 & 0 & 2 & 1 & 1 & 1
\end{bmatrix}
$$
\end{proof}

We know that because of the results in Section \ref{sec:binary_matroid}, Corollary \ref{cor:representable_counterexample} is false for $\F_2$. We leave it as an open question whether there is an online matroid embedding for class of $\F_3$ and $\F_5$-representable matroids.

\section{Approximate Embeddings}
\label{sec:conclusion}

Like metric embeddings, it is useful to extend the notion of matroid embedding to allow \emph{distortion}. We define an (offline) $\beta$-approximate embedding 
$f:\M \rightarrow \N$ as a map between ground sets that approximately preserves rank:
$$\frac{1}{\beta} \cdot \rank_\M(S) \leq \rank_\N(f(S)) \leq \rank_\M(S), \forall S \subseteq \M$$
We also define a randomized $\beta$-approximate matroid embedding as a family of functions $f_r:\M \rightarrow \N_r$ indexed by a random variable $r \sim R$ such:
\begin{equation}\label{eq:approx_1}
 \rank_{\N_r}(f_r(S)) \leq \rank_\M(S) \text{ a.s. } \forall S \subseteq \M
 \end{equation}
\begin{equation}\label{eq:approx_2}
 \E[\rank_{\N_r}(f_r(S))] \geq \frac{1}{\beta} \rank_\M(S), \forall S \subseteq \M
 \end{equation}
Their online counterparts can be defined in the natural way: given a class of matroids $\C$ and host matroid $\BigM$ then an online $\beta$-approximate randomized embedding is a family of functions: $f_{\M, \pi, r}:\M \rightarrow \BigM$ defined for each $\M \in \C$, $\pi:\M \rightarrow [n]$ and $r \in R$ such that: (i) it is a randomized $\beta$-approximate matroid embedding for each fixed $\M, \pi$; (ii) satisfied the prefix-restriction property defined in Section \ref{sec:OMEs}.

For example, if we can design a randomized order-independent $\beta$-approximate embedding for class $\C$ into a matroid $\BigM$ and there is a known $\alpha$-competitive algorithm for the known-matroid MSP on $\BigM$, then we can use the reduction in Section \ref{sec:reduction} to convert it into an $(\alpha \beta)$-competitive algorithm for the online-revealed-matroid MSP on class $\C$.\\

It is worth noting that every loop-free matroid $\M$  of rank $n$ admits a trivial $n$-approximate embedding into the free matroid $\Free_1 := U_{1,1}$. In the next paragraph, we observe below that we can't obtain better than $n$ for any graphic host matroid. This is related to the notion called \emph{partition property}.

\paragraph{Offline Embedding and the Partition Property} The notion of the $\alpha$-partition property was defined by Abdolazimi et al \cite{AbdolazimiKKG23} to generalize a property exploited by Korula-P\'{a}l \cite{korula2009algorithms} in their algorithm for the graphic matroid secretary problem. Translating it to our notation, we say that a matroid $\M$ satisfies the $\alpha$-partition property if there is an (offline) $\alpha$-approximate randomized embedding into a free matroid (i.e., a matroid where every non-empty subset is independent). The name partition property comes from the fact that partition matroids are the class of matroids that admit an (exact) morphism to a free matroid. Their respective lower and upper bounds translate to the following results: 

\begin{theorem}[Korula-Pal \cite{korula2009algorithms}] If $\M$ is a graphic matroid, then it admits a $2$-approximate randomized embedding into a free matroid.
\end{theorem}

\begin{theorem}[Abdolazimi et al.~\cite{AbdolazimiKKG23}, Dughmi et al.~\cite{DughmiKP24}] If $\M$ is the complete binary matroid of rank $n$  and if it admits an $\beta$-approximate randomized embedding into a free matroid, then $\beta \geq \Omega(n / \log n)$. Moreover, there exist a linear matroid $\M$ of rank $n$ such that $\beta$-approximate randomized embedding into a free matroid only exist for $\beta \geq \Omega(n)$.
\end{theorem}

Since composing an $\beta$-approximate embedding with a $\alpha$-approximate embedding we obtain an $(\alpha \beta)$-approximate embedding, we can strenghten the lower bound to also allow for embedding into graphic matroids.

\begin{corollary}\label{cor:graphic_distortion}
If $\M$ is the complete binary matroid of rank $n$  and if it admits an $\beta$-approximate randomized embedding into a graphic matroid, then $\beta \geq \Omega(n / \log n)$. Moreover, if there is a linear matroid $\M$ such that it admits a $\beta$-approximate randomized embedding into a graphic matroid, then $\beta = \Omega(n)$.
\end{corollary}

As we discussed, constructing an $n$-approximate embedding is trivial. For binary matroids, we complement the lower bound of Dughmi et al.~\cite{DughmiKP24} with an algorithm to construct a $O(n / \log n)$ embedding into the free matroid. Moreover, this embedding can be computed online.

\paragraph{Matching Upper Bound that is Also Online Computable}

We now construct a randomized approximate embedding from a binary matroid to $\Free_n$. Using Theorem \ref{thm:omm_binary} we can identify every received element $x \in \M$ with a vector in $\F_2^n \setminus \{0\}$. To map $x$ to a partition matroid, we sample a random basis $b_1, \hdots, b_n$ of $\F_2^n$ and define the function $f_b:\F_2^n \setminus \{0\} \rightarrow [n]$ that maps each $x \in \F_2^n$ to the smallest index $i$ such that $b_i$ is in the unique circuit in $\{x, b_1,\hdots, b_n\}$.

We first show that this embedding doesn't increase the rank (first property of an approximate embedding in equation \eqref{eq:approx_1}):

\begin{lemma}\label{lemma:approx-embedding-property-1} Given any loop-free matroid $\M$ of rank $n$ and a basis $b_1,\hdots, b_n$ of $\M$, let $f(x)$ be the smallest index $i$ such that $b_i$ is in the unique circuit formed by $\{x, b_1, \hdots, b_n\}$.  The function $f_b : \M \rightarrow \Free_n$ is such that $\rank_{\Free}(f_b(S)) \leq \rank_{\M}(S), \forall S \subseteq \M$.
\end{lemma}

\begin{proof} Given $S$, choose $S' \subseteq S$ be such that $\rank_{\Free}(f_b(S)) = \rank_{\Free}(f_b(S')) = \abs{S'}$. We will show that $S'$ is independent in $\M$ and hence $\abs{S'} = \rank_\M(S') \leq \rank_\M(S)$.

To show that, observe that for any element $x \in \M$, $f_b(x) = j$ iff $x \in \span_\M(\{b_j, b_{j+1}, \hdots, b_n\}) \setminus \span_\M(\{b_{j+1}, \hdots, b_n\})$. If that is the case, then: $\span_\M(\{x,b_{j+1}, \hdots, b_n\}) = \span_\M(\{b_j, b_{j+1}, \hdots, b_n\})$.

Now, construct $x_1, \hdots, x_n$ such that $x_j$ is the element in $S'$ that maps to $j$ if there is such element and $x_j = b_j$ otherwise. Then $f_b(x_j) = j$. We can show by induction that $\span_\M(\{x_j, \hdots, x_n\}) = \span_\M(\{b_j, \hdots, b_n\})$, by replacing elements of $b_1,\hdots, b_n$ one by one by their corresponding $x_j$ element and observing that the spans are preserved. In particular, $\{x_1, \hdots, x_n\}$ must be independent since they span the entire matroid. As a consequence the original set $S'$ is independent.
\end{proof}

We now show a lower bound on the expected rank of $f_b(S)$ which will lead to the second property of an approximate embedding:

 
\begin{lemma}\label{lemma:approx-embedding-property-2} For the complete binary matroid $\F_2^n$, the embedding $f_b$ for a random basis $b_1,\hdots, b_n$ of $\F_2^n$ satisfies: 
 $$\E[\rank_{\Free}(f_b(S))] \geq (1-1/e) \log_2(\rank_\M(S)) - O(1), \forall S \subseteq [n]$$
\end{lemma}

\begin{proof} 

\emph{Step 1: Sampling the basis}
We consider the following procedure for sampling an uniformly random basis $b_1, \hdots, b_n$ of $\F_2^n$. For each $i$ we sample an uniform independent vector from $\F_2^n$ and call it $b'_i$. Now, if $b'_i$ is not in $\span(b_1,\hdots, b_{i-1})$ then we set $b'_i = b_i$. Otherwise we resample until we get an element not in the span of the previous elements and set it as $b_i$. Nevertheless, we still record $b'_i$ as the first element sampled such that $b'_1,\hdots, b'_n$ are i.i.d. vectors. Now, let's define $\mathcal{E}$ as the event that $b_i = b'_i$ for $i \leq n/2$. We observe that:
$$\P(\mathcal{E})=(1-2^{-n})(1-2^{-n+1}) \hdots (1-2^{-n/2-1}) \geq 1-\textstyle \sum_{k=1}^{n/2} 2^{-n+k-1}\geq 1-2^{-n/2}$$
Finally, define $B$ to be the $n \times n$ matrix over $\F_2$ that has $b_i$ as the $i$-th column. Since $b_1,\hdots, b_n$ is a basis, $B$ is invertible. \\

\emph{Step 2: Change of basis } Fix a set $S \subseteq \F_2^n$ such that $r = \min(\rank_{\F_2^n}(S),n/2)$ and let $S'=\{s_1, \hdots, s_r\}$ be an independent subset of $S$. Our goal is to bound $\E[\rank_{\Free}(f_b(S'))].$ 

Let also $X$ be a $n\times n$ invertible matrix over $\F_2$ where the first $r$ columns correspond to $s_1, \hdots, s_r$. Define now $\tilde{b}_i = X B^{-1} e_i$ (where $e_i$ is the standard basis). Since the columns of $B$ are uniformly random and $X$ is invertible, $\tilde{b}_1, \hdots, \tilde{b}_n$ is also a uniformly random basis. With that we observe that:
$$\E[\rank(f_{\tilde{b}_1..\tilde{b}_n}(S))] \geq \E[\rank(f_{\tilde{b}_1..\tilde{b}_n}(S'))] = \E[\rank(f_{e_1..e_n}(\{b_1, \hdots, b_r\}))] $$
since $b_i = BX^{-1}s_i$, $e_i = BX^{-1} \tilde{b}_i$ and $z \mapsto B X^{-1} z $ is a matroid isomorphism of $\F_2^n$.\\

\emph{Step 3: Conditioning on $\mathcal{E}$} The vectors $b_1, \hdots, b_r$ are not sampled i.i.d. but the vectors $b'_1, \hdots, b'_r$ are and conditioned on $\mathcal{E}$ they are the same. So we write:
$$\begin{aligned}
\E[\rank(f_{e_1..e_n}(\{b_1, \hdots, b_r\}))] & \geq
\E[\rank(f_{e_1..e_n}(\{b_1, \hdots, b_r\})) \mid \mathcal{E}] \cdot \P[\mathcal{E}] \\ & = \E[\rank(f_{e_1..e_n}(\{b'_1, \hdots, b'_r\})) \mid \mathcal{E}] \cdot \P[\mathcal{E}]\\ & \geq   
\E[\rank(f_{e_1..e_n}(\{b'_1, \hdots, b'_r\})) ] - n/2^{n/2}
\end{aligned}\\$$

\emph{Step 4: Bounding the rank in the free matroid} The rank in the free matroid is simply the number of elements in the image, so we can re-write:
$$\begin{aligned}\E[\rank(f_{e_1..e_n}(\{b'_1, \hdots, b'_r\})) ] & = \sum_{j=1}^n \P[j \in f_{e_1..e_n}(\{b'_1, \hdots, b'_r\}))] \\ & = \sum_{j=1}^n \left( 1-\prod_{i=1}^r \P[j \neq f_{e_1..e_n}(b'_i)] \right) 
\\ & = \sum_{j=1}^n \left( 1-( 1-2^{-j})^r\right)
\end{aligned}$$
since $\P[j \neq f_{e_1..e_n}(b'_i)]$ is the probability that the first non-zero entry of the random vector $b'_i$ is not $j$, which happens with probability $(1-2^{-j})$. Finally observe that for $j \leq \log_2(r)$ we have $( 1-2^{-j})^r \leq (1-1/r)^r \leq 1/e$. Putting this together we obtain that:
$$
\E[\rank(f_{e_1..e_n}(\{b'_1, \hdots, b'_r\})) ] \geq (1-1/e) \log_2(r) \geq (1-1/e) \log_2(\rank_{\F_2^n}(S)/2)
$$
This completes the proof.
\end{proof}

\begin{theorem}\label{thm:upper_bound_binary_partition}
For any loop-free binary matroid $\M$ there is a $O(n / \log n)$-online approximate embedding to the partition matroid.
\end{theorem}

\begin{proof} For a loop-free binary matroid $\M$ use the composition of the embedding $f_b$ for a random basis $b$ of $\F_2^n$ from the previous lemma with the online morphism $\M \rightarrow \F_2^n$ in Theorem \ref{thm:omm_binary}. Since both can be computed online, their composition can also be computed online.

By Lemma \ref{lemma:approx-embedding-property-1} the embedding satisfies the first property of an approximate embedding. By Lemma \ref{lemma:approx-embedding-property-2} a set of $S$ of rank $r$ in $\M$ is mapped to a set with expected rank $\max(1, c \log(r)-c')$ for constants $c,c'$. Finally, observe that $r/\max(1, c \log(r)-c') \leq O(n / \log(n))$.
\end{proof}

\paragraph{No Constant-Approximate Embedding from Graphical Matroid to Free Matroid}

Next, we show that there is no constant approximate online embedding of a graphical matroid into a free matroid when the elements are revealed online with access to an independence oracle. This is in contrast to the result of \cite{korula2009algorithms} that allows us to construct an approximate offline embedding. 
\begin{theorem}\label{thm:no_embedding_graphical}
    There is no constant approximate online embedding from the class of graphical matroids into a free matroid.
\end{theorem}
\begin{proof}
First, we construct the underlying graph that is revealed online to an algorithm that embeds the corresponding graphical matroid into a free matroid. Suppose the underlying graph has two disjoint paths, $e_1, e_2, \dots,  e_d$ and $f_1,\dots, f_d$. Let the endpoints of the edge $e_i = (v_i , v_{i+1})$ and $f_i = (u_i, u_{i+1})$. We consider two special edges $e_* = (v_1, u_1)$ and $f_* = (v_{d+1}, u_{d+1})$. In addition, we have a set of edges $g_i = (v_i , u_i)$ for all $i \in \{2, \dots , d\}$. We make $d$ many identical copies of the above graph and $k$-th copy of vertices and corresponding edges are denoted as $u_i^k, v_i^k, f_i^k, g_i^k$  and $f_*^k, e_*^k$ for all $i\in [d]$ and $k\in [d]$. 

To prove the lemma, we need to define an arrival order $\pi$  over the edges such that any algorithm that maps these edges into a free matroid, and equivalently, a simple partition matroid, can not obtain $O(1)$-approximate embedding. 

Suppose the arrival order of edges is as follows: first the set of edges $e^{k}_1 , \dots, e^{k}_d; f^{k}_1,\dots, f^{k}_d$ and $e^k_*$ arrives in uniformly random order for $k\in [d]$. Since these edges form an independent set, the algorithm can not distinguish the identity of these edges. Therefore, any constant approximate algorithm needs to map the arrived edges into $O(d^2)$ many disjoint parts of the underlying simple partition matroid. Let the algorithm map the arrived edges into $C\cdot d^2$ many parts. We first observe that at most $ C^2\cdot d^2$ many parts with more than $\frac 1 {C^2}$ edges are assigned to them by the algorithm. We let $P$ be the set of parts that contains at most $\frac{1}{C^2}$ many edges.  In addition, since the algorithm can not distinguish between the arrived edges, we can assume that the algorithm adds these edges into $C\cdot d^2$ many parts uniformly at random.

Next, for any fix $k$, we lower bound the probability that the algorithm embeds all $e^{k}_1 , \dots, e^{k}_d$; $f^{k}_1,\dots, f^{k}_d$ and $e^k_*$ into separate parts --- denoted as the event $\mathcal E$. 
We let $i_1, \dots, i_d$,  $j_1,\dots, j_d$ and $\ell^*$ be the parts in which algorithm embeds the edges $e^k_1, \dots, e^k_d$, $f_1^k, \dots, f_d^k$ and $e^k_*$ respectively. We can lower the probability of the event $\mathcal E$ by $i_q, j_q \in P$ for all $q\in [d]$ intersecting with the event $\mathcal E$. Since each part in $P$ contains at most $\frac{1}{C^2}$ many edges, we can lower bound the probability of the event $\mathcal E$ by,
$$ \prod_{i=1}^{d+1} \left(1 - \frac{1}{(C-C^2)d^2 - \frac{i}{C^2}} \right) \geq \left(1 - \frac{d}{(C-C^2)d^2 - \frac{d}{C^2}} \right)\geq \left(1 - O\left ( \frac{1}{d} \right) \right).$$

For the rest of the proof, we condition on the event that the algorithm embeds all $e^{k}_1 , \dots, e^{k}_d$; $f^{k}_1,\dots, f^{k}_d$ and $e^k_*$ into separate parts. Therefore, from now on, we also drop the superscript $k$ as it is fixed in the rest of the proof. In addition, we let $P_i$, $Q_i$ and $P_*$ be the parts that contains the edges $e_i$, $f_i$ and $e_*$ for $i\in [d]$. 

Next, let the edge $f_*$ arrive followed by the edge $e_*$. Since, the edges $e_1 , \dots, e_d; f_1,\dots, f_d$ and $e_*, f_*$ form a cycle, the algorithm is forced to add the edge $e_*$ into one of the existing parts. Since all the edges are symmetric to the algorithm, at this stage, we can assume that the algorithm adds the edge $f_*$ into one of the parts uniformly at random that contains one of the edges $e_1 , \dots, e_d$; $f_1,\dots, f_d$ and $e_*$.  We define an event $\mathcal F$ as the event where the edge $f_*$ does not belong to the parts $P_{d}, P_{d-1}, \dots, P_{d - \sqrt d}$  and $Q_{d}, Q_{d-1}, \dots, Q_{d - \sqrt d}$. We note that $\Pr[\mathcal F] \geq 1 - \frac{2}{\sqrt d}$. For the rest of the sketch, we condition on the event $\mathcal F$. 

Next, we define the arrival order of the edges as: $g_d, g_{d-1}, g_{d-2}, \dots, g_{d-\sqrt d}$. Upon arrival of the edge $g_i$ for any $i\in \{d,\dots , d-\sqrt d\}$, we consider the following cycles 
$$ \begin{aligned} \mathcal C &= (e_*, e_1,\dots, e_{i-1}, g_i, u_{i-1}, u_{i-2},\dots, u_1)\\
\mathcal C' &= ( g_i, v_i, v_{i+1}, \dots, v_d, f_*, u_d, u_{d-1},\dots, u_{i}).
\end{aligned}$$ Under the event $\mathcal F$, we claim that the edge $g_i$ has to be added to the part $P^*$ that contains the edge $f_*$. 

Suppose the algorithm does not add the edge $g_i$ to the part $P^*$. In this case, we have either $g_i \in \{P_1,\dots, P_{i-1}, Q_1,\dots, Q_{i-1}, \bar P\} \setminus P^*$  or $g_i \in \{P_{i+1},\dots, P_d, Q_{i+1}, \dots, Q_d \}$. In the first case, we can observe that all the edges in $C'$ belong to different parts which breaks the assumption that the algorithm comes up with a valid embedding. In the second case, all the edges in $C$ are in different parts which again breaks the same assumption. Therefore, the edge $g_i$ has to be in the part $P^*$. We can further observe that if $g_i$ is mapped to the part $P^*$ then it leads to a valid (approximate) embedding. 

The above argument shows that all the edges $g_i: i\in \{d-\sqrt d, \dots, d\}$ have to be in the same part while conditioned on $\mathcal E$ and $\mathcal F$ while the rank of $g_i: i\in \{d-\sqrt d, \dots, d\}$ is $\sqrt d$. Since $\Pr[\mathcal E \cap \mathcal F] \geq 1 - O(1/\sqrt d)$,  the algorithm is at most $\sqrt d$-approximate. This rules out any constant approximate algorithm. 
\end{proof}
\color{black}

\bibliographystyle{alpha}
\bibliography{msp}

\newcommand{\etalchar}[1]{$^{#1}$}
\begin{thebibliography}{CGLW22}

\bibitem[AAK{\etalchar{+}}07]{alon2007testing}
Noga Alon, Alexandr Andoni, Tali Kaufman, Kevin Matulef, Ronitt Rubinfeld, and
  Ning Xie.
\newblock Testing k-wise and almost k-wise independence.
\newblock In {\em STOC 2007}, pages 496--505, 2007.

\bibitem[AGM03]{alon2003almost}
Noga Alon, Oded Goldreich, and Yishay Mansour.
\newblock Almost k-wise independence versus k-wise independence.
\newblock {\em Inf. Process. Lett.}, 88(3):107--110, 2003.

\bibitem[AKKG23]{AbdolazimiKKG23}
Dorna Abdolazimi, Anna~R. Karlin, Nathan Klein, and Shayan~Oveis Gharan.
\newblock Matroid partition property and the secretary problem.
\newblock In {\em ITCS 2023}, pages 2:1--2:9, 2023.

\bibitem[AL12]{alon2012almost}
Noga Alon and Shachar Lovett.
\newblock Almost k-wise vs. k-wise independent permutations, and uniformity for
  general group actions.
\newblock In {\em APPROX'12}, pages 350--361, 2012.

\bibitem[Ala14]{Alaei14}
Saeed Alaei.
\newblock Bayesian combinatorial auctions: Expanding single buyer mechanisms to
  many buyers.
\newblock {\em {SIAM} J. Comput.}, 43(2):930--972, 2014.

\bibitem[Bar98]{bartal1998approximating}
Yair Bartal.
\newblock On approximating arbitrary metrices by tree metrics.
\newblock In {\em STOC 1998}, pages 161--168, 1998.

\bibitem[BFU20]{barta2020online}
Yair Bartal, Nova Fandina, and Seeun~William Umboh.
\newblock Online probabilistic metric embedding: {A} general framework for
  bypassing inherent bounds.
\newblock In {\em SODA 2020}, pages 1538--1557, 2020.

\bibitem[BIK07]{BabaioffIK07}
Moshe Babaioff, Nicole Immorlica, and Robert Kleinberg.
\newblock Matroids, secretary problems, and online mechanisms.
\newblock In {\em SODA 2007}, pages 434--443, 2007.

\bibitem[BIKK07]{BabaioffIKK07}
Moshe Babaioff, Nicole Immorlica, David Kempe, and Robert Kleinberg.
\newblock A knapsack secretary problem with applications.
\newblock In {\em APPROX-RANDOM 2007}, pages 16--28, 2007.

\bibitem[BIKK18]{BabaioffIKK18}
Moshe Babaioff, Nicole Immorlica, David Kempe, and Robert Kleinberg.
\newblock Matroid secretary problems.
\newblock {\em J. {ACM}}, 65(6):35:1--35:26, 2018.

\bibitem[Bou85]{Bourgain85}
Jean Bourgain.
\newblock On lipschitz embedding of finite metric spaces in hilbert space.
\newblock {\em Isr. J. Math.}, 52:46–52, 1985.

\bibitem[CGLW22]{pi-uniform-prophet}
Ioannis Caragiannis, Nick Gravin, Pinyan Lu, and Zihe Wang.
\newblock Relaxing the independence assumption in sequential posted pricing,
  prophet inequality, and random bipartite matching.
\newblock In {\em WINE 2022}, pages 131--148, 2022.

\bibitem[CHMS10]{ChawlaHMS10}
Shuchi Chawla, Jason~D. Hartline, David~L. Malec, and Balasubramanian Sivan.
\newblock Multi-parameter mechanism design and sequential posted pricing.
\newblock In Leonard~J. Schulman, editor, {\em STOC 2010}, pages 311--320,
  2010.

\bibitem[CL12]{ChakrabortyL12}
S.~Chakraborty and O.~Lachish.
\newblock Improved competitive ratio for the matroid secretary problem.
\newblock In {\em SODA 2012}, page 1702–1712, 2012.

\bibitem[DFKL20]{DuttingFKL20}
Paul D{\"{u}}tting, Michal Feldman, Thomas Kesselheim, and Brendan Lucier.
\newblock Prophet inequalities made easy: Stochastic optimization by pricing
  nonstochastic inputs.
\newblock {\em {SIAM} J. Comput.}, 49(3):540--582, 2020.

\bibitem[DK14]{dinitz2014matroid}
Michael Dinitz and Guy Kortsarz.
\newblock Matroid secretary for regular and decomposable matroids.
\newblock {\em SIAM J. Comput.}, 43(5):1807--1830, 2014.

\bibitem[DKP24]{DughmiKP24}
Shaddin Dughmi, Yusuf~Hakan Kalayci, and Neel Patel.
\newblock Limitations of stochastic selection problems with pairwise
  independent priors.
\newblock In {\em STOC 2024}, pages 479--490, 2024.

\bibitem[Dug20]{dughmi-outer-limits}
Shaddin Dughmi.
\newblock The outer limits of contention resolution on matroids and connections
  to the secretary problem.
\newblock In {\em ICALP 2020}, volume 168 of {\em LIPIcs}, pages 42:1--42:18,
  2020.

\bibitem[Dug21]{dughmi2021matroid}
Shaddin Dughmi.
\newblock Matroid secretary is equivalent to contention resolution.
\newblock In {\em ITCS 2022}, volume 215 of {\em LIPIcs}, pages 58:1--58:23,
  2021.

\bibitem[EHKS18]{EhsaniHKS18}
Soheil Ehsani, MohammadTaghi Hajiaghayi, Thomas Kesselheim, and Sahil Singla.
\newblock Prophet secretary for combinatorial auctions and matroids.
\newblock In {\em SODA 2018}, pages 700--714, 2018.

\bibitem[FO17]{FO17}
Tara Fife and James Oxley.
\newblock Laminar matroids.
\newblock {\em Eur. J. Comb.}, 62:206--216, 2017.

\bibitem[FSZ18]{FeldmanSZ18}
Moran Feldman, Ola Svensson, and Rico Zenklusen.
\newblock A simple \emph{O}(log log(rank))-competitive algorithm for the
  matroid secretary problem.
\newblock {\em Math. Oper. Res.}, 43(2):638--650, 2018.

\bibitem[FSZ21]{FeldmanSZ21}
Moran Feldman, Ola Svensson, and Rico Zenklusen.
\newblock Online contention resolution schemes with applications to bayesian
  selection problems.
\newblock {\em {SIAM} J. Comput.}, 50(2):255--300, 2021.

\bibitem[HKS07]{HajiaghayiKS07}
Mohammad~Taghi Hajiaghayi, Robert~D. Kleinberg, and Tuomas Sandholm.
\newblock Automated online mechanism design and prophet inequalities.
\newblock In {\em AAAI 2007}, pages 58--65, 2007.

\bibitem[IMSZ10]{indyk2010online}
Piotr Indyk, Avner Magen, Anastasios Sidiropoulos, and Anastasios Zouzias.
\newblock Online embeddings.
\newblock In {\em APPROX-RANDOM 2010}, pages 246--259, 2010.

\bibitem[Ind01]{Indyk01}
Piotr Indyk.
\newblock Algorithmic applications of low-distortion geometric embeddings.
\newblock In {\em FOCS 2001}, pages 10--33, 2001.

\bibitem[IW11]{im2011secretary}
Sungjin Im and Yajun Wang.
\newblock Secretary problems: {L}aminar matroid and interval scheduling.
\newblock In {\em SODA 2011}, pages 1265--1274, 2011.

\bibitem[JSZ13]{jaillet2013advances}
Patrick Jaillet, Jos{\'e}~A Soto, and Rico Zenklusen.
\newblock Advances on matroid secretary problems: {F}ree order model and
  laminar case.
\newblock In {\em IPCO 2013}, pages 254--265, 2013.

\bibitem[Kan85]{kantor1985homogeneous}
William~M Kantor.
\newblock Homogeneous designs and geometric lattices.
\newblock {\em J. Comb. Theory Ser. A}, 38(1):66--74, 1985.

\bibitem[KP09]{korula2009algorithms}
Nitish Korula and Martin P{\'a}l.
\newblock Algorithms for secretary problems on graphs and hypergraphs.
\newblock In {\em ICALP 2009}, pages 508--520, 2009.

\bibitem[KW12]{KleinbergW12}
Robert Kleinberg and S.~Matthew Weinberg.
\newblock Matroid prophet inequalities.
\newblock In {\em STOC 2012}, pages 123--136, 2012.

\bibitem[Lac14]{Lachish14}
Oded Lachish.
\newblock {O}(log log rank) competitive ratio for the matroid secretary
  problem.
\newblock In {\em FOCS 2014}, pages 326--335, 2014.

\bibitem[LLR95]{LinialEtAl95}
Nathan Linial, Eran London, and Yuri Rabinovich.
\newblock The geometry of graphs and some of its algorithmic applications.
\newblock {\em Combinatorica}, 15:215–245, 1995.

\bibitem[LMP22]{LeichterMP22}
Marilena Leichter, Benjamin Moseley, and Kirk Pruhs.
\newblock On the impossibility of decomposing binary matroids.
\newblock {\em Oper. Res. Lett.}, 50(5):623--625, 2022.

\bibitem[LW{\etalchar{+}}06]{luby2006pairwise}
Michael Luby, Avi Wigderson, et~al.
\newblock Pairwise independence and derandomization.
\newblock {\em Foundations and Trends{\textregistered} in Theoretical Computer
  Science}, 1(4):237--301, 2006.

\bibitem[Mat02]{Matousek02}
Jiri Matousek.
\newblock {\em Lectures on Discrete Geometry}.
\newblock Springer, New York, NY, 2002.

\bibitem[NR25]{newman2023online}
Ilan Newman and Yuri Rabinovich.
\newblock Online embedding of metrics.
\newblock {\em Isr. J. Math.}, 2025.
\newblock Forthcoming.

\bibitem[OGV13]{oveis2013variants}
Shayan Oveis~Gharan and Jan Vondr{\'a}k.
\newblock On variants of the matroid secretary problem.
\newblock {\em Algorithmica}, 67:472--497, 2013.

\bibitem[SSZ23]{SantiagoSZ23}
Richard Santiago, Ivan Sergeev, and Rico Zenklusen.
\newblock Simple random order contention resolution for graphic matroids with
  almost no prior information.
\newblock In {\em SOSA 2023}, pages 84--95, 2023.

\bibitem[SSZ25]{SantiagoSZ25}
Richard Santiago, Ivan Sergeev, and Rico Zenklusen.
\newblock Constant-competitiveness for random assignment matroid secretary
  without knowing the matroid.
\newblock {\em Math. Program.}, 210(1):815--846, 2025.

\bibitem[Tut58]{tutte1958homotopy}
William~Thomas Tutte.
\newblock A homotopy theorem for matroids. {I}, {II}.
\newblock {\em Trans. Am. Math. Soc.}, 88(1):144--174, 1958.

\bibitem[Vad12]{salil2012pseudorandomness}
Salil~P. Vadhan.
\newblock Pseudorandomness.
\newblock {\em Foundations and Trends{\textregistered} in Theoretical Computer
  Science}, 7(1--3):1--336, 2012.

\end{thebibliography}
\newpage
\appendix
\section{Reduction to a Special Case of Prophet MSP}\label{sec:proof_simple_reduction}

In this appendix, we provide a proof of  Lemma~\ref{lem:reduction_simple_to_unique_weight}. First, we present a simple reduction from \cite{dughmi2021matroid}. 

\begin{lemma}[Sublemma-4.2 from \cite{dughmi2021matroid}]\label{lem:red_dughhmi}
    If there exists an $\alpha$-approximate Prophet MSP on matroid $\M$ with $\textbf{Rank}(\M) = d$ and weight distribution $\D$ supported over $\left \{ \frac{1}{256 \cdot d}, \frac{2}{256\cdot d}, \frac{2^2}{256 \cdot d}, \dots ,1  \right \}$ with the offline optimum $\E_{\D}[\opt(\M)] \in \left[\frac{1}{16} , 1 \right]$ then there exists $\frac{\alpha}{256}$-approximate prophet MSP on matroid $\M$ with any arbitrary weight distribution.
\end{lemma}
Given the simple reduction in Lemma~\ref{lem:red_dughhmi}, we prove the following simple reduction that allow us to focus on the weight distribution that assigns distinct weight to each element of the matroid. We consider weight class, $$W = \left \{ \frac{1}{256 \cdot d} + \frac{i-1}{256 d\cdot n^2}: i \in \{0,1,\dots , 256 d\cdot n^2+1 \}\right\} \cup \left \{ \frac{1}{256 \cdot d}, \frac{2}{256\cdot d}, \frac{2^2}{256 \cdot d}, \dots ,1  \right \} .$$
Above, $n$ denotes the set of elements of the matroid $\M$.

Given any prophet MSP instance consists of matroid $\M$ and arbiritary weight distribution, we reduce it to a prophet MSP instance $\F$ consists of matroid $\M$ and weight distribution $\D$ supported on $\left \{ \frac{1}{256 \cdot d}, \frac{2}{256\cdot d}, \frac{2^2}{256 \cdot d}, \dots ,1  \right \}$ with the offline optimum $\E_{\D}[\opt(\M)] \in \left[\frac{1}{16} , 1 \right]$. Then we construct prophet MSP instnace $\F'$ over the same matroid but slightly perturbed weight distribution as follows: we find a random permutation $\pi$ over the set of elements of the matroid. We sample the weights of the elements using the distribution $\D$. Then for all $i\in [n]$, we subtract $\frac{i-2}{256 \cdot d\cdot n^2}$ from the weight of element $e$, iff the element $e$ appears at the $i$-th position on the permutation $\pi$.

We observe that in the perturbed matroid secretary instance, the weight distribution is supported over the set of weights $W$. In addition, for any $w \in W$, there is at most one element whose weight is assigned to be $w$. This simply follows from the fact that for any two elements $e, e'$ with distinct weight in the draw from distribution $\D$ will be assigned a different weight as the prtrubution to both elements is smaller than $\frac{1}{256d \cdot n}$. In addition, the perturbation is distinct for two distinct elements $e,e' \in V$. Therefore, the pair of elements $e,e'$ with the identical weight in the draw from $\D$ are assigned different weights in the perturbed instance.

\begin{proof}[Proof of Lemma~\ref{lem:reduction_simple_to_unique_weight}]
    To prove the lemma, we prove the following statement: if there is an $\alpha$ approximate algorithm for perturb prophet matroid secretary instance $\F'$ then we can construct $\alpha - O\left( \frac{1}{d} \right)$-approximate algorithm for prophet matroid secretary instance $\F$.

    Suppose, we are given an algForithm $\A$ that is $\alpha$-approximate for the instance $\F'$. Our goal is to design an algorithm for the instance $\F$ using $\A$ that is $\alpha - \frac{1}{d}$-approximate. We give the following online reduction from $\F$ to $\F'$: in the instance $\F$, upon arrival of the element $e$ at the position $i\in [n]$, we perturb the weight of the element $e$ by subtracting $\frac{i-1}{256\cdot d\cdot n^2}$. We note that since the arrival order of the elements is uniformly at random, the modified weight distribution is identical to that of weight distribution in the perturbed instance. Hence, we feed the perturbed weight of the arrived element to the algorithm $\A$. We then select an arrived element $e$ iff the algorithm $\A$ selects the element $e$ in the perturbed instance. 
    
    Let $S$,  $\opt$ and $\opt'$ be the selected elements by $\A$, the optimal value of the instance $\F$, and the optimal value of the perturbed instance $\F'$. We denote the modified weight as $w'(\cdot)$. By construction, we have,
    $$\begin{aligned}
      w(S) \geq w'(S) \geq \alpha\cdot \opt' \geq \alpha \cdot  \left(1- \frac{1}{32\cdot d} \right) \cdot \opt.
    \end{aligned}
    $$
    Above, the first inequality holds because $w'(e)\leq w(e)$. The second inequality holds because $\A$ is $\alpha$-approximate for thee perturbed distribution. The final inequality holds because the total decrease in the weight $\sum_{e\in V} (w(e) - w'(e)) \leq \frac{1}{512\cdot d}$ and $\opt \geq \frac{1}{16}$ by assumption on $\D$. This concludes the proof. 
\end{proof}

\section{Exact from Approximate Pairwise Independence} 
\label{sec:techncial_proof}
In this section, we present the proof of Theorem~\ref{thm:tv_distance}, and construct an exact pairwise-independent weight distribution on $\BigM_{[m\cdot N]}$ that is close to the weight distribution $\D^*$ defined in Definition~\ref{defn:almost_pw_ind_dist} on $\BigM_{[m\cdot N]}$. First, in Section~\ref{sec:proc_to_pw_ind}, we present our construction that iteratively applies small perturbations as defined in Procedure~\ref{proc:one} and Procedure~\ref{proc:two} to turn the weight distribution $\D^*$ into an exact pairwise-independent weight distribution $\tilde \D$. Then in Section~\ref{sec:proc_parmeters}, we obtain tight upper and lower bounds on the parameters of Procedure~\ref{proc:one} and Procedure~\ref{proc:two}. Afterwards, in Sections~\ref{sec:proc_1_analysis} and \ref{sec:proc_2_analysis}, we bound the total perturbations due to Procedures~\ref{proc:one} and~\ref{proc:two}, respectively.  Finally, Section~\ref{sec:wrap-up} shows how this implies the theorem.

\subsection{Constructing an Exact-Pairwise Independent Distribution}\label{sec:proc_to_pw_ind}
    
    
    Before we present the procedure to obtain an exact pairwise independent weight distribution, we recall the notations from Theorem~\ref{thm:tv_distance}. We have a prophet MSP instance on matroid $\M$ and weight distribution $\D$ supported over the set of weights $W=\{ w_1 , w_2 , \dots, w_m \}$ satisfying the conditions from Lemma~\ref{lem:reduction_simple_to_unique_weight}. Given the prophet MSP instance we let $\D^*$ be the weight distribution over $\BigM$ defined in Definition~\ref{defn:almost_pw_ind_dist}. For any $i, j \in [m]$, we let,
    $$\begin{aligned}
    p_i& = \Pr_{w\sim \D}[\exists \vec v\in \M: w(\vec v) = w_i]\\
    p_{ij} &= \Pr_{w\sim \D}[\exists \vec v\in \M: w(\vec v) = w_i \land \exists \vec v'\in \M: w(\vec v') = w_j ].\end{aligned}$$    
    Since $\D$ is arbitrarily correlated over the support $W$, we can potentially have $p_{ij}\neq p_i \cdot p_j$. 

    We first observe that for any $i\in [m], \ell \in N_i$ and $\vec u\in \BigM$, $\vec u^{i,\ell}\in \BigM_{[m\cdot N]}$ can potentially take either weight of $w_i$ or zero. Therefore, we define Bernoulli random variables $X^{\vec u}_{i\ell}$ for all $\vec u^{i,\ell}\in \BigM_{[m\cdot N]}$ such that $X^{\vec u}_{i\ell}= 1$ iff $w(\vec u^{i,\ell}) = w_i$. We note that the Bernoulli random variables $X^{\vec u}_{i\ell}$ for all $\vec u^{i,\ell}\in \BigM_{[m\cdot N]}$ contain the full information of the weight distribution over $\BigM$.

We next observe that,
$$\begin{aligned}
    \E[X^{\vec u}_{i,\ell} = 1] &= \Pr \left[ \exists \vec v \in \M: w(\vec v) = w_i \land f(\vec v) = \vec u \right]\cdot \Pr_{\ell' \sim \operatorname{Unif}(N_i)}[\ell' = \ell]\\
    &=\Pr \left[ \exists \vec v \in \M: w(\vec v) = w_i \right] \cdot \Pr\left[f(\vec v) = \vec u  \right]\cdot \frac{1}{N}\\
    &= \frac{p_i}{M \cdot N}.
\end{aligned}$$
Above, the second equality holds because the weight assignment to $\M$ is independent of the matroid morphism $f$. The second equality holds because $f = f''\circ f'$ and $f''$ is a uniformly random automorphism from $\operatorname{Aut}(\BigM)$, which implies $\Pr[f(\vec v) = \vec u] = \frac{1}{|\BigM|} = \frac{1}{M}$. 
Similarly, we compute the pairwise joint distribution of the Bernoulli random variables. First, for any elements $\vec u, \vec u' \in \BigM$, $i, j\in [m]$ and $\ell\in [N_i], \ell' \in [N_j]$, we have,
\begin{equation*}
    \E[X^{\vec u}_{i,\ell} = 1 \land X^{\vec u'}_{j,\ell'} =1 ]=0 
\end{equation*}
if $i=j$ or $\vec u = \vec u'$ because there can be at most one element $\vec v\in \M$ with $w(\vec v) =w_i$ and $f$ is a matroid morphism and two distinct elements from $\M$ with weight $w_i$ and $w_j$ can not be mapped to $\vec u \in \BigM$ via $f$. 

Finally, for any pair of elements ${\vec u}_{i,\ell}, {\vec u'}_{j,\ell}$, for $i\neq j$ and $\vec u \neq \vec u' $, we have,
$$\begin{aligned}
    \E[X^{\vec u}_{i,\ell} = 1\land X^{\vec u'}_{j,\ell'} =1] &= \Pr \left[ \exists\vec  v \in \M: w(\vec v) = w_i \land  \exists \vec v' \in \M: w(\vec v') = w_j \land  f(\vec v) = \vec u \land f(\vec v') = u' \right]\cdot \frac{1}{N^2}\\
    &= p_{ij} \cdot \Pr[f(\vec v) = \vec u \land f(\vec v') = u'] \cdot \frac{1}{N^2} \\
    &= \frac{p_{ij}}{M\cdot (M-1)\cdot N^2}.
\end{aligned}$$
Above, the second equality holds because the weight assignment to $\M$ is independent of the matroid morphism $f$. The second equality holds because $f = f''\circ f'$ and $f''$ is a uniformly random automorphism from $\operatorname{Aut}(\BigM)$, which implies $\Pr[f(\vec v) = \vec u \land f(\vec v') = u'] =  \frac{1}{M \cdot (M-1)}$. Given the definition of $X^{\vec u}_{i,\ell}$, since it captures the complete information about the weight distribution over $\BigM_{[m\cdot M]}$, we prove the existence of Bernoulli random variables $\tilde X_{i,\ell}^{\vec u}:\forall \vec u^{i,\ell}$ with pairwise-independent joint distribution $\tilde {\vec X}$ s.t. $\operatorname{TV}_{\vec X, \tilde {\vec X}} \leq O \left ( \frac{m^3}{M}\right)$. This immediately implies the proof of Theorem~\ref{thm:tv_distance}.


It is difficult to give a closed form construction of $\tilde {\vec X}$ in a ``simple" procedure, therefore,  we give a sequential process that iteratively generates Bernoulli random variables $ X_{i,\ell}^{\vec u}(k):\forall \vec u^{i,\ell} \in \BigM_{[m\cdot N]}$ for $k=0,\dots, k^*$ starting from 
$X_{i,\ell}^{\vec u}(0)= X_{i,\ell}^{\vec u}:\forall \vec u^{i,\ell}\in \BigM_{[m\cdot N]}$ such that at the end of the process, $\tilde X_{i,\ell}^{\vec u}(\bar k):\forall \vec u^{i,\ell}$ are pairwise independent. For any $i,j \in [m]$, and distinct $\vec u, \vec u' \in \BigM$ $i\neq j$, we define bias
    $$\begin{aligned}
        \varepsilon_{ij} &= | \E[X^{\vec u}_{i,\ell} = 1\land X^{\vec u'}_{j,\ell'} =1] -  \E[X^{\vec u}_{i,\ell} = 1] \cdot \E[X^{\vec u'}_{j,\ell'}=1 ] | = \frac{|\bar{p}_{ij} - p_i\cdot p_j|}{M^2\cdot N^2},
    \end{aligned}$$
    where $\bar {p}_{ij} = \frac{M}{M-1}\cdot p_{ij}$. The above expression captures the closeness of the random variables $X_{i,\ell}^{\vec u}:\forall \vec u^{i,\ell} \in \BigM_{[n\cdot M]}$ from being pairwise independent. The above bias $\eps_{ij}$ does not depend on the choice of vectors $\vec u, \vec u'$ and their corresponding labels $\ell, \ell'$. On the other hand, for $i\in [m]$ and distinct $\vec u, \vec u' \in \BigM$, we let,
    \begin{equation*}
         \varepsilon_{i} = | \E[X^{\vec u}_{i,\ell} = 1\land X^{\vec u'}_{i,\ell'} =1]  -  \E[X^{\vec u}_{i,\ell} = 1] \cdot \E[X^{\vec u'}_{i,\ell'} =1] | = \frac{p_i^2}{M^2\cdot N^2}.
    \end{equation*}
    
 
    Our high-level idea is to sequentially transform the random variables starting from $\{X_{i,\ell}^{\vec u}(0):\forall \vec u^{i,\ell} \in \BigM_{[m\cdot N]} \}$ distributed as $\D(0) = \D^*$ to $\{X_{i,\ell}^{\vec u}(k):\forall \vec u^{i,\ell} \}$ for bounded $k$ such that at each step, the resultant distribution becomes ``closer" to being pairwise independent. More formally, we prove the following lemma. 
\begin{lemma}\label{lem:pairwise_ind}
    Let $\{\tilde X_{i,\ell}^{\vec u}: \vec u_{i,\ell} \in \BigM_{{[m\cdot N]}}\}$ be the random variables obtained after applying Procedure~\ref{proc:one} and Procedure~\ref{proc:two} to $\{ X_{i,\ell}^{\vec u}: \vec u_{i,\ell} \in \BigM_{[m\cdot N]}\}$, then $\{\tilde X_{i,\ell}^{\vec u}: \vec u_{i,\ell} \in \BigM_{[m\cdot N]}\}$ are pairwise-independent. 
\end{lemma}

\begin{proof}[Proof of Lemma~\ref{lem:pairwise_ind}]
    We define initial bias $\eps_{ij}(0) = \eps_{ij}$ and $\eps_i(0) = \eps_i$.  We let the evolution of the weight distribution starting from $\D^* = \D(0) \rightarrow \D(1) \rightarrow \dots \rightarrow \D(\bar k)$, where $\bar k = O(m^2)$ in Procedure~\ref{proc:one}. For any $k>0$, we let,
    $$\begin{aligned}
         \varepsilon_{i}(k) &= | \E[X^{\vec u}_{i,\ell}(k) = 1\land X^{\vec u'}_{i,\ell'}(k) =1]  -  \E[X^{\vec u}_{i,\ell}(k) = 1] \cdot \E[X^{\vec u'}_{i,\ell'}(k) =1] |\\
         \varepsilon_{ij}(k) &= | \E[X^{\vec u}_{i,\ell}(k) = 1\land X^{\vec u'}_{j,\ell'}(k) =1] -  \E[X^{\vec u}_{i,\ell}(k) = 1] \cdot \E[X^{\vec u'}_{j,\ell'}(k) = 1 ] |. 
    \end{aligned}$$
    Through the process of transforming Bernoulli random variables in Procedure~\ref{proc:one}, we make sure that $\eps_{ij}(k)$ and $\eps_i(k)$ do not depend on $\vec u, \vec u' \in \BigM$ and their labels $\ell\in [N]$ and $\ell' \in [N] $ and rather only depends on the weight class $i, j \in [m]$.

    Suppose we have completed $k$ many iterations of Procedure~\ref{proc:one} and have obtained distribution $\D(k)$ over $\{X_{i,\ell}^{\vec u}(k): \forall \vec u^{i,\ell} \in \BigM_{[m\cdot N]}\}$. In addition, we have for any $i\in [m], \vec u \in \BigM$ and $\ell \in [N]$, we have $\E[X^{\vec u}_{i,\ell}] = p_i(k)$ and any distinct $i, j \in [m], \vec u, \vec u' \in \BigM$ and any labels $\ell , \ell' \in [N]$, we have $\E[X^{\vec u}_{i,\ell} \cdot X^{\vec u'}_{j,\ell'}] = p_{ij}(k)$.
    Given $X_{i,\ell}^{\vec u}(k):\forall \vec u^{i,\ell}\in \BigM_{[m\cdot N]}$, and $i\neq j$ with $\eps_{ij}(k) \neq 0$, we consider the two cases: $\eps_{ij}(k) > 0$ and $\eps_{ij}(k) < 0$ separately and resolve their pairwise correlations (See Procedure~\ref{proc:one}). 

    \begin{tcolorbox}
    \paragraph{Procedure 1:} \customlabel{proc:one}{1}
    \paragraph{Initialize $k = 0$ and $\{ X^{\vec u}_{i,\ell}(0) = X^{\vec u}_{i,\ell}: \vec u^{i,\ell} \in \BigM_{[m\cdot N]}\}$}
     
    \paragraph{Run until for all distinct $i,j \in [m]$, $p_{ij}(k)= p_i(k) \cdot p_j(k)$:}
     
    \paragraph{Case-1 ($p_{ij}(k)> p_i(k) \cdot p_j(k)$):} We define $X_{i,\ell}^{\vec u}(k+1):\forall \vec u^{i,\ell}\in \BigM_{[m\cdot N]}$ as follows: 
     
    \begin{enumerate}    
        \item Assign $X_{i,\ell}^{\vec u}(k+1) = X_{i,\ell}^{\vec u}(k):\forall \vec u^{i,\ell}\in \BigM_{[m\cdot N]}$.
        \item With probability $q_{ij}$, sample $Z\sim \operatorname{Ber}(1/2)$ and do the following:
        \begin{enumerate}
            \item If $Z=1$ then assign $X_{i,\ell}^{\vec u}(k+1) = 1 $ with probability $\frac{1}{N}$ independently $\forall \vec u \in  \BigM, \forall \ell \in [N]$  and $X_{j,\ell}^{\vec u}(k+1) = 0:\forall \vec u \in  \BigM, \forall \ell \in [N]$.
            \item Otherwise assign $X_{i,\ell}^{\vec u}(k+1) = 0:\forall \vec u \in  \BigM_{[m\cdot N]}, \ell \in [N]$ and $X_{j,\ell}^{\vec u}(k+1) = 1$ with probability $\frac{1}{N}$ independently $\forall \vec u \in  \FF_2^d, \ell \in [N]$.
        \end{enumerate}
    \end{enumerate}
    
    \paragraph{Case-2 ($p_{ij}(k)< p_i(k) \cdot p_j(k)$):} We define $X_{i,\ell}^{\vec u}(k+1):\forall \vec u^{i,\ell}$ as follows: 
    \begin{enumerate}
        \item Assign $X_{i,\ell}^{\vec u}(k+1) = X_{i,\ell}^{\vec u}(k):\forall \vec u^{i,\ell}\in \BigM_{[m\cdot N]}$.
        \item With probability $ q_{ij}$, let $X_{i,\ell}^{\vec u}(k+1) = 1:\forall \vec u \in  \BigM, \ell \in [N]$ and $X_{j,\ell}^{\vec u}(k+1) = 1:\forall \vec u \in  \BigM, \ell \in [N]$.
    \end{enumerate}
    \end{tcolorbox}

First, we can straightforwardly observe that while obtaining $\{X_{i,\ell}^{\vec u}(k): \forall \vec u^{i,\ell} \in \BigM_{[m\cdot N]}\}$ we ensure that for any fix $i \in [N]$, $\vec u, \vec u' \in \BigM$ and $\ell, \ell' \in [N]$, we have $\E[X_{i,\ell}^{\vec u}(k+1)] = \E[X_{j,\ell'}^{\vec u'}(k+1)]$. In addition, given any fix distinct $i, j \in [m]$, for any distinct  $\vec u, \vec u' \in \BigM$ and labels $\ell, \ell' \in [N]$, $\E[X_{i,\ell}^{\vec u}(k+1) \cdot X_{j,\ell'}^{\vec u'}(k+1)]$ does not depend on $\vec u,\vec u'$ and their labels $\ell, \ell' \in [N]$.

In both cases, we let $q_{ij}$ such that we obtain $\eps_{ij}(k+1) = 0$. We will define the closed form and an upper bound on $q_{ij}$ by writing the conditions to obtain $\eps_{ij}(k+1) = 0$ in Section~\ref{sec:proc_parmeters}. However, we prove the following crucial property that shows that the above procedure stops after iterating over all pairs of indices $i\neq j$ and leads to the distribution where $X^{\vec u}_{i,\ell}(k)$ and $X^{\vec u'}_{j,\ell'}(k)$ are independent as long as $i \neq j$.
\begin{claim}\label{claim:monotonicity_correlation}
    For any pair $r,r'$ such that $r\neq r'$, we have $\eps_{rr'}(k+1)\leq \eps_{rr'}(k)$. In addition, for any $\vec u\in \BigM$, $\ell, \ell' \in [N]$ and distinct $r, r' \in [m]$, we have, $$ |\E[X_{r,\ell}^{\vec u}(k+1)\cdot X_{r',\ell'}^{\vec u}(k+1)] - \E[X_{r,\ell}^{\vec u}(k+1)]\cdot\E[ X_{r',\ell'}^{\vec u}(k+1)] |\leq\E[X_{r,\ell}^{\vec u}(k)\cdot X_{r',\ell'}^{\vec u}(k)]- \E[X_{r,\ell}^{\vec u}(k)]\cdot\E[ X_{r',\ell'}^{\vec u}(k)] $$
\end{claim}
\begin{proof}
    We can easily observe that when $r,r' \notin \{i,j\}$ then $\eps_{rr'}(k+1) = \eps_{rr'}(k)$. In addition, when $r = i, r' =j$, the choice of $q_{ij}$ ensures that $\eps_{ij}(k+1) = 0 < \eps_{ij}(k+1)$. Hence, we can focus on the case, $r' = i$. We note that for any pair of $\vec u, \vec u'\in \BigM$,
    $$\begin{aligned}
        &|\E[X^{\vec u}_{r,\ell}(k+1) \cdot X^{\vec u}_{r',\ell ' }(k+1)] - \E[X_{r,\ell}^{\vec u}(k+1)]\cdot\E[ X_{r',\ell'}^{\vec u}(k+1)] | \\
        &= | \E[X^{\vec u}_{i,\ell}(k+1) =  X^{\vec u'}_{r,\ell'}(k+1) =1] -    \E[X^{\vec u}_{i,\ell}(k+1) = 1] \cdot \E[X^{\vec u'}_{r,\ell'}(k+1) =1] |\\
        &=|(1-q_{ij})\cdot \E[X^{\vec u}_{i,\ell}(k) =  X^{\vec u'}_{r,\ell'}(k) =1] + \alpha \cdot q_{ij} \cdot  \E[X^{\vec u'}_{r,\ell'}(k) =1] \\
        &-    ((1-q_{ij})\E[X^{\vec u}_{i,\ell}(k) = 1] + q_{ij}\cdot \alpha) \cdot \E[X^{\vec u'}_{r,\ell'}(k) =1] | \qquad \qquad (\alpha =1 \text{ or 
 } \alpha = 1/2)\\
        &= (1- q_{ij})\cdot \eps_{ir}(k).
    \end{aligned} $$
    Next, we observe that for the case when $\vec u \neq \vec u'$, $|\E[X^{\vec u}_{i,\ell}(k') \cdot X^{\vec u}_{r',\ell ' }(k')] - \E[X_{i,\ell}^{\vec u}(k')]\cdot\E[ X_{r',\ell'}^{\vec u}(k')] | = \eps_{ir}(k') $ for $k'\in \{k, k+1\}$. This further implies that $\eps_{i,r}(k+1)\leq \eps_{i,r}(k)$. 
\end{proof}
Claim~\ref{claim:monotonicity_correlation} implies that that once Procedure~\ref{proc:one} iterates over all possible distinct pairs $i, j \in [m]$ over $\bar k$ many iterations, for the resultant , the resultant $X^{\vec u}_{i,\ell}$ and $X^{\vec u'}_{j,\ell'}$ are independent as long as $i \neq j$. 
Suppose the above procedure terminates after $\bar k$ many iterations. We note that $\bar k \leq \frac{m\cdot (m-1)}{2}$ due to Claim~\ref{claim:monotonicity_correlation}. Once the above procedure terminates for all $i\neq j$, the only correlation we have left is between $X^{\vec u}_{i,\ell}$ and $X^{\vec u'}_{i,\ell'}$ for $\vec u, \vec u' \in \BigM$, $i\in [m]$ and $\ell, \ell \in [N]$ and between $X^{\vec u}_{i,\ell}, X^{\vec u}_{j, \ell'}$ for $\vec u \in \M$, $i, j \in [m]$ and $\ell, \ell' \in [N]$.   



We now describe Procedure~\ref{proc:two} to resolve the rest of the correlations, which is similar to our Procedure~\ref{proc:one}. In Procedure~\ref{proc:two}, we initialize $\{ \tilde X^{\vec u}_{i,\ell}(0) = X^{\vec u}_{i,\ell}(\bar k): \vec u^{i,\ell} \in \BigM_{[m\cdot N]}\}$ and at each iteration, we resolve the correlation between a pair of Bernoulli random variables 
which exhibits pairwise correlation. This is in contrast to Procedure~\ref{proc:one} that resolved the pairwise correlation between the sets of random variables corresponding to pairs of weight classes $i,j \in [m]$.

    \begin{tcolorbox}
    \paragraph {Procedure 2: } 
    \customlabel{proc:two}{2}
    \paragraph {Initialize $k = 0$ and $\{ \tilde X^{\vec u}_{i,\ell}(0) =\tilde X^{\vec u}_{i,\ell}(\bar k): \vec u^{i,\ell} \in \BigM_{[m\cdot N]}\}$}
     
    \paragraph{Run until $\{ \tilde X^{\vec u}_{i,\ell}(k): \vec u^{i,\ell} \in \BigM_{[m\cdot N]}\}$ are Pairwise Independent:}
     
    \paragraph {$\exists \tilde X^{\vec u}_{i,\ell}(k )   \tilde X^{\vec u'}_{j,\ell'}(k)$ not independent for some $i, j\in [m]$,  $\vec u,\vec u'\in \BigM$ and $\ell, \ell' \in [N]$}
     
    \paragraph {Case-1: ($\E [\tilde X^{\vec u}_{i,\ell}(k )\cdot   \tilde X^{\vec u'}_{j,\ell}(k )] >\E [\tilde X^{\vec u}_{i,\ell}(k ) ]\cdot  \E[\tilde X^{\vec u'}_{j,\ell}(k )] $):} 
     
    \begin{enumerate}    
        \item Assign $X_{i,\ell}^{\vec u}(k+1) = X_{i,\ell}^{\vec u}(k):\forall \vec u^{i,\ell}\in \BigM_{[m\cdot N]}$.
        \item With probability $q(k)$, sample $Z\sim \operatorname{Ber}(1/2)$ and do the following:
        \begin{enumerate}
            \item If $Z=1$ then assign $X_{i,\ell}^{\vec u}(k+1) = 1$ and $X_{j,\ell'}^{\vec u'}(k+1) = 0$.
            \item Otherwise assign $X_{i,\ell}^{\vec u}(k+1) = 0$ and $X_{j,\ell'}^{\vec u'}(k+1) = 1$.
        \end{enumerate}
    \end{enumerate}
    
    \paragraph {Case-2: ($\E [\tilde X^{\vec u}_{i,\ell}(k )\cdot   \tilde X^{\vec u'}_{j,\ell}(k )] < \E [\tilde X^{\vec u}_{i,\ell}(k ) ]\cdot  \E[\tilde X^{\vec u'}_{j,\ell}(k )] $):} 
    \begin{enumerate}
        \item Assign $X_{i,\ell}^{\vec u}(k+1) = X_{i,\ell}^{\vec u}(k):\forall \vec u^{i,\ell}\in \BigM_{[m\cdot N]}$.
        \item With probability $ q(k)$, let $X_{i,\ell}^{\vec u}(k+1) = X_{j,\ell'}^{\vec u'} = 1$.
    \end{enumerate}
    \end{tcolorbox}

At any iteration $k$ of Procedure~\ref{proc:two}, let $\tilde X_{i,\ell}^{\vec u}(k)$ and $\tilde X_{j,\ell'}^{\vec u'}(k)$ are the selected random variables whose correlation is being resolved. We set $q(k)$ such that we end up with $\E [\tilde X^{\vec u}_{i,\ell}(k )\cdot   \tilde X^{\vec u'}_{j,\ell}(k )] = \E [\tilde X^{\vec u}_{i,\ell}(k ) ]\cdot  \E[\tilde X^{\vec u'}_{j,\ell}(k )]$. Similar to Claim~\ref{claim:monotonicity_correlation}, we show that the correlation between all pair of random variables can only go lower. 
\begin{claim}
    For any $k\geq 0$ and pair of random variables  $X(k), X'(k) \in \{\tilde X_{i,\ell}^{\vec u}(k): \vec u^{i,\ell} \in \BigM_{[m\cdot N]}\}$, we have, $$|\E [X(k +1)\cdot   X'(k +1)] - \E [X(k+1) ]\cdot  \E[X'(k+1)]| \leq  |\E [X(k)\cdot   X'(k )] - \E [X(k) ]\cdot  \E[X'(k)]|.$$
\end{claim}
\begin{proof}
 Let $\tilde X_{i,\ell}^{\vec u}(k)$ and $\tilde X_{j,\ell'}^{\vec u'}(k)$ are the selected random variables whose correlation is being resolved.   For any $X(k),X'(k) \notin \{\tilde X_{i,\ell}^{\vec u}(k),\tilde X_{j,\ell'}^{\vec u'}(k)\}$, the condition in the claim trivially satisfies. On the other hand, if $X(k) = \tilde X_{i,\ell}^{\vec u}(k)$ and $X'(k) = \tilde X_{j,\ell'}^{\vec u'}(k)$ then we have $\E[X(k+1)\cdot X'(k+1)] = \E[X(k+1)]\cdot \E[X'(k+1)]$ which implies the condition in the claim. 
 
 To complete the proof, we focus on the case when $X(k) = \tilde X_{i,\ell}^{\vec u}(k)$ and $X'(k+1) \notin \{\tilde X_{i,\ell}^{\vec u}(k),\tilde X_{j,\ell'}^{\vec u'}(k)\}$. The proof of this case is identical to the proof of Claim~\ref{claim:monotonicity_correlation} but we give a detailed proof for the sake of completeness. We observe that,
 $$\begin{aligned}
     &|\E[X(k+1)\cdot X'(k+1)] - \E[X(k+1)]\cdot \E[X'(k+1)]|\\
     =& |(1 - q(k))\cdot \E[X(k)\cdot X'(k)] + q(k)\cdot \alpha \cdot \E[X'(k)]- ((1-q(k)) \cdot \E[X(k)] + q(k) \cdot \alpha) \cdot \E[X'(k)]|\\
     =& (1 - q(k))\cdot|\E[X(k)\cdot X'(k)] - \E[X(k)]\cdot \E[X'(k)]|\\
     \leq & |\E[X(k)\cdot X'(k)] - \E[X(k)]\cdot \E[X'(k)]|.
 \end{aligned}$$
 In the above calculations, $\alpha = 1/2$ if $\E[\tilde X_{i,\ell}^{\vec u}(k) \cdot \tilde X_{j,\ell'}^{\vec u'}(k)] > \E[\tilde X_{i,\ell}^{\vec u}(k)] \cdot \E[\tilde X_{j,\ell'}^{\vec u'}(k)]$ and $\alpha = 1$ otherwise. 
\end{proof}
The above claim implies that at any iteration $k$, if $\tilde X_{i,\ell}^{\vec u}(k)$ and $\tilde X_{j,\ell'}^{\vec u'}(k)$ are the selected random variables whose correlation is being resolved, either $i = j$ or $\vec u = \vec u'$ since we have $\E [\tilde X^{\vec u}_{i,\ell}(0 )\cdot   \tilde X^{\vec u'}_{j,\ell}(0 )] =\E [\tilde X^{\vec u}_{i,\ell}(0 ) ]\cdot  \E[\tilde X^{\vec u'}_{j,\ell}(0 )] $ for distinct $\vec u, \vec u' \in \BigM$ and distinct $i, j \in [m]$ at the end of Procedure~\ref{proc:one}. This implies that Procedure~\ref{proc:two} stops after iterative over pair of random variables $X_{i,\ell}^{\vec u},X_{j,\ell'}^{\vec u} $ for $\vec u\in \BigM, i,j \in [m], \ell, \ell' \in [N]$ and $X_{i,\ell}^{\vec u},X_{i,\ell'}^{\vec u'} $ for $\vec u, \vec u' \in \BigM, i \in [m], \ell, \ell' \in [N]$. This completes the proof fo the lemma that the resultant random variables $\{\tilde X_{i,\ell}^{\vec u}: \vec u_{i,\ell} \in \BigM_{{[m\cdot N]}}\}$ at the end of Procedure~\ref{proc:one}, followed by Procedure~\ref{proc:two} are pairwise independent.
\end{proof}

\subsection{Parameters of Procedure~\ref{proc:one} and Procedure~\ref{proc:two}}\label{sec:proc_parmeters}
Next, in order to complete the proof of Theorem~\ref{thm:tv_distance}, we first bound $q_{ij}$ and $q(k)$, i.e, the probability by which Procedure~\ref{proc:one} and Procedure~\ref{proc:two} deviate from the original random variables $\{ X_{i,\ell}^{\vec u}: \vec u_{i,\ell} \in \BigM_{{[m\cdot N]}}\}$. We next obtain upper and lower bound on the parameters of the Procedures~\ref{proc:one} and \ref{proc:two}, respectively. 
\begin{lemma}\label{claim:bound_perturb_prob}
    Let $(i,j)$ be the indices selected at iteration $k$ of Procedure~\ref{proc:one}, then we have, 
    \begin{enumerate}
        \item When $p_{ij}(k)> p_i(k)\cdot p_j(k)$, we have $ \frac{N\cdot \eps_{ij}(k)}{2(p_i(k) + p_j(k))} \leq q_{ij} \leq  \frac{N\cdot \eps_{ij}(k)}{p_i(k) + p_j(k)}$, and
        \item When $p_{ij}(k) < p_{i}(k)\cdot p_{j}(k) $, we have $\varepsilon_{ij} \leq q_{ij} \leq  \frac{\eps_{ij}(k)}{(1 -p_i(k) - p_j(k) )}$
    \end{enumerate}
In addition, for $k\geq 0$, let $\tilde X_{i,\ell}^{\vec u}(k)$ and $\tilde X_{j,\ell'}^{\vec u'}(k)$ are the selected random variables whose correlation is being resolved at the $k$-th iteration of Procedure~\ref{proc:two}. Let $p_1(k) = \E[\tilde X_{i,\ell}^{\vec u}]$ and $p_2(k) = \E[\tilde X_{j,\ell'}^{\vec u'}]$, then we have, 
    \begin{enumerate}
        \item When $\E [\tilde X^{\vec u}_{i,\ell}(k )\cdot   \tilde X^{\vec u'}_{j,\ell}(k )] >\E [\tilde X^{\vec u}_{i,\ell}(k ) ]\cdot  \E[\tilde X^{\vec u'}_{j,\ell}(k )] $, we have $$ \frac{|\E [\tilde X^{\vec u}_{i,\ell}(k )\cdot   \tilde X^{\vec u'}_{j,\ell}(k )] -\E [\tilde X^{\vec u}_{i,\ell}(k ) ]\cdot  \E[\tilde X^{\vec u'}_{j,\ell}(k )] |}{2(p_1(k) + p_2(k))}\leq q(k) \leq  \frac{|\E [\tilde X^{\vec u}_{i,\ell}(k )\cdot   \tilde X^{\vec u'}_{j,\ell}(k )] -\E [\tilde X^{\vec u}_{i,\ell}(k ) ]\cdot  \E[\tilde X^{\vec u'}_{j,\ell}(k )] |}{p_1(k) + p_2(k)},$$
        \item When $\E [\tilde X^{\vec u}_{i,\ell}(k )\cdot   \tilde X^{\vec u'}_{j,\ell}(k )] <\E [\tilde X^{\vec u}_{i,\ell}(k ) ]\cdot  \E[\tilde X^{\vec u'}_{j,\ell}(k )] $, we have $$|\E [\tilde X^{\vec u}_{i,\ell}(k)\cdot   \tilde X^{\vec u'}_{j,\ell}(k )] -\E [\tilde X^{\vec u}_{i,\ell}(k )|\leq q(k) \leq  \frac{|\E [\tilde X^{\vec u}_{i,\ell}(k)\cdot   \tilde X^{\vec u'}_{j,\ell}(k )] -\E [\tilde X^{\vec u}_{i,\ell}(k )| ]\cdot  \E[\tilde X^{\vec u'}_{j,\ell}(k )] |}{1- p_1(k) - p_2(k)},$$
    \end{enumerate}
\end{lemma}
\begin{proof}
    In both cases, we want to find $q_{ij}$ such that $p_{ij}(k+1) =p_i(k+1)\cdot p_j(k+1) $. We analyze both cases of Procedure~\ref{proc:one} separately, 
    
     \paragraph{Case-1 ($p_{ij}(k) > p_i(k)\cdot p_j(k)$)} In this case, by construction, we have,
    \begin{equation*}
        p_i(k+1) = (1- q_{ij})\cdot p_{i}(k) + \frac{q_{ij}}{2N} \quad  \text{ and  } \quad p_{ij}(k+1) = (1- q_{ij})\cdot p_{ij}(k). 
     \end{equation*}
    For simplicity in notations, we let $a_{ij} = \frac 1 {4N^2} - \frac{1}{2N}\cdot (p_i(k) + p_j(k))$ and $b_{ij} = \frac{1}{2N}\cdot (p_i(k) + p_j(k)) + \eps_{ij}(k) + p_i(k)\cdot p_j(k)$. By constraint, $p_{ij}(k+1) =p_i(k+1)\cdot p_j(k+1) $, we get,
    $$\begin{aligned}
        &(1- q_{ij})\cdot p_{ij}(k) = \left((1- q_{ij})\cdot p_{i}(k) + \frac{q_{ij}}{2N} \right)\cdot \left((1- q_{ij})\cdot p_{j}(k) + \frac{q_{ij}}{2N} \right)\\
        \implies & \left(\frac 1 {4N^2} - \frac{1}{2N}\cdot (p_i(k) + p_j(k))\right)\cdot q_{ij}^2 + \left(\frac 1 {2N}\cdot (p_i(k) + p_j(k)) + \eps_{ij}(k) + p_i(k)\cdot p_j(k) \right)\cdot q_{ij} - \eps_{ij}(k) = 0\\
        \implies& q_{ij} = \frac{ - b_{ij} + \sqrt{b_{ij}^2 + 4\cdot \eps_{ij}(k)\cdot a_{ij}}}{2\cdot a_{ij}} \implies q_{ij} = \frac{\eps_{ij}(k)}{ b_{ij} + \sqrt{b_{ij}^2 + 4\cdot \eps_{ij}(k)\cdot a_{ij}}} \leq \frac{N\cdot \varepsilon_{ij}(k)}{p_i(k) + p_j(k)}.
    \end{aligned}$$
    Above the last inequality follows because $b_{ij}\geq \frac 1 2 \cdot ( p_{i}(k) + p_j(k))$. In addition, we get $q_{ij}\geq \frac{\varepsilon_{ij}}{2\cdot b_{ij} + 2\varepsilon_{ij}\cdot a_{ij}} \geq \frac{\varepsilon_{ij}}{2\cdot (p_i(k) + p_j(k))}$. We conclude the proof of this case.     
    
    \paragraph{Case-2 ($p_{ij}(k) < p_i(k)\cdot p_j(k)$)} In this case, by construction, we have,
    \begin{equation*}
        p_i(k+1) = (1- q_{ij})\cdot p_{i}(k) + q_{ij} \quad  \text{ and  } \quad p_{ij}(k+1) = (1- q_{ij})\cdot p_{ij}(k) +q_{ij}. 
    \end{equation*}
    For simplicity in notations, we let $a_{ij} = 1 + p_i(k)\cdot p_j(k) - (p_i(k) + p_j(k))$ and $b_{ij} = (1 - p_i(k) - p_j(k)$  $- p_i(k)\cdot p_j(k) + \eps_{ij}(k) )$. By constraint, $p_{ij}(k+1) =p_i(k+1)\cdot p_j(k+1) $, we get,
    $$\begin{aligned}
        &(1- q_{ij})\cdot p_{ij}(k) +q_{ij} = \left( (1- q_{ij})\cdot p_{i}(k) + q_{ij} \right)\cdot \left( (1- q_{ij})\cdot p_{j}(k) + q_{ij} \right)\\
        \implies & \left(1 + p_i(k)\cdot p_j(k) - (p_i(k) + p_j(k))\right)\cdot q_{ij}^2 - \left(1 - p_i(k) - p_j(k) - p_i(k)\cdot p_j(k) + \eps_{ij}(k) \right)\cdot q_{ij} + \eps_{ij}(k) = 0\\
        \implies& q_{ij} = \frac{  b_{ij} - \sqrt{b_{ij}^2 - 4\cdot \eps_{ij}(k)\cdot a_{ij}}}{2\cdot a_{ij}} \implies q_{ij} = \frac{\eps_{ij}(k)}{ b_{ij} + \sqrt{b_{ij}^2 - 4\cdot \eps_{ij}(k)\cdot a_{ij}}}.
    \end{aligned}$$
    Above the last inequality follows because $b_{ij}\geq (1 - p_{i}(k) - p_j(k))$.

The proof of the second part for Procedure~\ref{proc:two} is identical to the proof of claim for the first part by replacing $\varepsilon_{ij}(k) = |\E [\tilde X^{\vec u}_{i,\ell}(0 )\cdot   \tilde X^{\vec u'}_{j,\ell}(0 )] -\E [\tilde X^{\vec u}_{i,\ell}(0 ) ]\cdot  \E[\tilde X^{\vec u'}_{j,\ell}(0 )] |$, $p_i(k), p_j(k)$ by $p_1(k), p_2(k)$ (defined in the statement of the lemma.) and $q$ by $N\cdot q$. This concludes the proof.
\end{proof}

\subsection{Analysis of Procedure~\ref{proc:one}}\label{sec:proc_1_analysis}
Next, we bound the marginals and bias in the set of random variables  $\{ X_{i,\ell}^{\vec u}(\bar k): \vec u_{i,\ell} \in \BigM_{{[m\cdot N]}}\}$ at the end of Procedure~\ref{proc:one}. Here, $\bar k$ denotes the number of iterations performed of Procedure~\ref{proc:one}.  We prove the following claim that lower bounds the marginals $p_i(\bar k)$ for all $i\in [m]$.
\begin{claim}\label{claim:bound_procedure1}
    For any $k\geq 0$ and $i\in [m]$, we have that $p_{i}(k)\geq p_i(0)$ at the end of the $k$-th iteration of Procedure~\ref{proc:one}. In addition, we have that,
    \begin{enumerate}
        \item For any pair of distinct $i,j \in [m]$, $q_{ij} \leq \frac{1}{(M-1)}$ when $p_{ij}(0)> p_{i}(0) \cdot p_j(0)$ and $q_{ij} \leq \frac{1}{(M-1)^2 \cdot N^2}$ when $p_{ij}(0) < p_{i}(0) \cdot p_j(0)$. 
        \item For any $i\in [m]$, $p_i(\bar k) \leq p_i(0) + \frac{m-1}{N \cdot (M-1)}$ at the end of Procedure~\ref{proc:one}.
    \end{enumerate}   
\end{claim}
\begin{proof}
     We first observe that $p_i(k) \leq \frac 1 {2N}$ for large enough $N, M >0$ which follows inductively. This implies that for any $k>0$, we have $p_i(k+1) = p_i(k) $ if weight class $i$ is not processed at the iteration $k$ of Procedure~\ref{proc:one}. Otherwise if weight classes $i,j$ is processed for some $j\in [m]$, we have $ p_i(k+1 )\geq  (1-q_{ij})\cdot  p_i(k) + q_{ij}/2 \geq p_i(k) $. This concludes the proof of the first part. 
     
     Second, we observe that for any distinct $i, j \in [m]$, we can bound $q_{ij}$ as follows:
     $$\begin{aligned}
         q_{ij} &\leq \frac{N\cdot \eps_{ij}(k)}{(p_i(k)+ p_j(k))}\leq \frac{\eps_{ij}(0)}{(p_i(0)+ p_j(0))}\\
         &\leq \frac{N\cdot |\frac{M}{M-1}\cdot {p}_{ij} - p_i\cdot p_j|}{M^2\cdot N^2\cdot \min(p_i(0), p_j(0))} \leq \frac{N\cdot M}{M-1} \cdot \frac{\min(p_i(0), p_j(0))}{M\cdot N \cdot \min(p_i(0), p_j(0))}\\
         &= \frac{1}{(M-1)}.
     \end{aligned}$$
     Above, the first inequality holds due to Claim~\ref{claim:bound_perturb_prob}, the second inequality holds because  of Claim~\ref{claim:monotonicity_correlation} and the first part of the claim $p_{i}(k)\geq p_i(0)$.  In the other case, we have, 
     $$\begin{aligned}
         q_{ij} &\leq \frac{\eps_{ij}(k)}{(1 - p_i(k) -  p_j(k))}\leq \frac{\eps_{ij}(0)}{1-p_i(k) - p_j(k)} \leq \frac{N\cdot |\frac{M}{M-1}\cdot {p}_{ij} - p_i\cdot p_j|}{M^2\cdot N^2\cdot (1-1/M)} \leq \frac{1}{(M-1)^2 \cdot N^2}.
     \end{aligned} $$
          
     Finally, the probability $p_i(\bar k)$ at the end of Procedure~\ref{proc:one} can be bound by union bound on the event that at least one of the Bernoulli $Z\sim \operatorname{Ber}(q_{ij})$ for some $j\neq i \in [m]$ turns out to be $1$ throughout Procedure~\ref{proc:one}. This implies,
     \begin{equation*}
         p_i(\bar k) \leq p_i(0) + \frac{1}{N} \cdot \sum_{j\neq i} q_{ij} \leq p_i(0) + \frac{m-1}{(M-1)\cdot N}. 
     \end{equation*}
     This completes the proof. 
\end{proof}
Next, we analyze the type of remaining correlation after the end of Procedure~\ref{proc:one}. We recall that at the end of Procedure~\ref{proc:one}, the only correlation we have left is between $X^{\vec u}_{i,\ell}$ and $X^{\vec u'}_{i,\ell'}$ for $\vec u, \vec u' \in \BigM$, $i\in [m]$ and $\ell, \ell \in [N]$ and between $X^{\vec u}_{i,\ell}, X^{\vec u}_{j, \ell'}$ for $\vec u \in \M$, $i, j \in [m]$ and $\ell, \ell' \in [N]$. We make the following claim. 
\begin{claim}\label{obs:correlation_proc1}
    Let Procedure~\ref{proc:one} run for $\bar k$ many iterations. The following statements holds:
    \begin{enumerate}
        \item For any $i\in [m]$ and pair of $\vec u, \vec u' \in \BigM$ and $\ell, \ell'\in [N]$, the random variables $X_{i,\ell}^{\vec u}(\bar k)$ and $X_{i,\ell'}^{\vec u'}(\bar k)$ are positively correlated. In addition, $$\E [ X^{\vec u}_{i,\ell}(\bar k )\cdot    X^{\vec u'}_{i,\ell'}(\bar k )] -\E [ X^{\vec u}_{i,\ell}( \bar k ) ]\cdot  \E[ X^{\vec u'}_{i,\ell'}(\bar k )] \leq \frac{m}{4\cdot N^2 \cdot (M-1)}. $$
        \item For any distinct $i,j \in [m]$,  $\vec u \in \BigM$ and $\ell, \ell' \in [N]$, we have $X_{i,\ell}^{\vec u}(\bar k)$ and $X_{j,\ell'}^{\vec u'}(\bar k)$ are negatively correlated. In addition,  $$\E [ X^{\vec u}_{i,\ell}(\bar k ) ]\cdot  \E[ X^{\vec u'}_{j,\ell}(\bar k )] - \E [ X^{\vec u}_{i,\ell}(\bar k )\cdot    X^{\vec u'}_{j,\ell}(\bar k )]  \leq \frac{1}{N^2 \cdot M^2}. $$
    \end{enumerate}
\end{claim}
\begin{proof}
We first consider any $i\in [m]$ and pair of $\vec u, \vec u' \in \BigM$ and $\ell, \ell'\in [N]$, we observe that before we start Procedure~\ref{proc:one}, $\E [X^{\vec u}_{i,\ell}(0 )\cdot    X^{\vec u'}_{i,\ell'}(0 )] =0$ as there is at most one element $\vec v\in \M$ that takes weight $w_i$. Therefore, at iteration $\bar k$, $ X^{\vec u}_{i,\ell}(\bar k )\cdot   X^{\vec u'}_{j,\ell}(\bar k ) = 1$ holds iff  $ X^{\vec u}_{i,\ell}(\bar k )\cdot   X^{\vec u'}_{j,\ell}(\bar k )$ with probability $q_{ij}$ while resolving correlation between $i, j \in [m]$ where $p_{ij} < p_i \cdot p_j$ or with probability $\frac{q_{ij}}{4\cdot N^2}$ while resolving correlation between $i,j \in [m]$ where $p_{ij} > p_i \cdot p_j.$ This implies, 
$$\begin{aligned}
    \E[ X^{\vec u}_{i,\ell}(\bar k )\cdot    X^{\vec u'}_{j,\ell}(\bar k )] &= \sum_{j\neq i: (p_{ij} < p_i \cdot p_j)} q_{ij} + \sum_{j\neq i: (p_{ij} > p_i \cdot p_j)} \frac{q_{ij}}{4N^2}\\
    & \leq \sum_{j\neq i: (p_{ij} < p_i \cdot p_j)} \frac{1}{(M-1)^2 \cdot N^2} +  \sum_{j\neq i: (p_{ij} > p_i \cdot p_j)} \frac{1}{4N^2} \cdot \frac{1}{(M-1)}\\
    &\leq \frac{m}{4\cdot N^2 \cdot (M-1)}.
\end{aligned}$$
Above, the first inequality follows from Claim~\ref{claim:bound_procedure1} and the last inequality follows from the fact that $ \frac{1}{4N^2} \cdot \frac{1}{(M-1)} > \frac{1}{(M-1)^2 \cdot N^2} $ for $M>3$. The positive correlation between $ X^{\vec u}_{i,\ell}(\bar k ),    X^{\vec u'}_{j,\ell}(\bar k )$ follows from the fact that Procedure~\ref{proc:one} mixes the original distribution with positively correlated distribution among  $ X^{\vec u}_{i,\ell}(\bar k ),    X^{\vec u'}_{j,\ell}(\bar k )$ with marginal higer than the original marginals. This above bound on $\E[ X^{\vec u}_{i,\ell}(\bar k )\cdot    X^{\vec u'}_{j,\ell}(\bar k )]$ concludes the proof of the first part. 

Finally, for any distinct $i,j \in [m]$,  $\vec u \in \BigM$ and $\ell, \ell' \in [N]$, we have $X_{i,\ell}^{\vec u}(\bar k)\cdot X_{j,\ell'}^{\vec u'}(\bar k) =1$, only if while resolving correlation between $i,j $ during Procedure~\ref{proc:one}, it ends up assigning $X_{i,\ell}^{\vec u}(k+1) = 1:\forall \vec u \in  \BigM, \ell \in [N]$ and $X_{j,\ell}^{\vec u}(k+1) = 1:\forall \vec u \in  \BigM, \ell \in [N]$. In this case, $q_{ij}$ satisfies, $p_{ij}(k+1) = p_{i}(k) \cdot p_j(k+1)$. Since,  $\E [X_{i,\ell}^{\vec u}( k)\cdot X_{j,\ell'}^{\vec u'}( k)] < p_{ij}(k)$ at round $k$, $X_{i,\ell}^{\vec u}( k), X_{j,\ell'}^{\vec u'}( k)$ remains negatively correlated at iteration $k$ while the correlation between $i, j$ is being resolved. During all the other iterations of Procedure~\ref{proc:one}, Claim~\ref{claim:monotonicity_correlation} ensures that $X_{i,\ell}^{\vec u}(\bar k), X_{j,\ell'}^{\vec u'}(\bar k)$ remains negatively correlated. 
\end{proof}

\subsection{Analysis of Procedure~\ref{proc:two}}\label{sec:proc_2_analysis}
Given the structure and bound on the correlation between $X^{\vec u}_{i,\ell}$ and $X^{\vec u'}_{i,\ell'}$ for $\vec u, \vec u' \in \BigM$, $i\in [m]$ and  $\ell, \ell' \in [N]$ and between $X^{\vec u}_{i,\ell}, X^{\vec u}_{j, \ell'}$ for $\vec u \in \M$, $i, j \in [m]$ and $\ell, \ell' \in [N]$ in Proposition~\ref{obs:correlation_proc1}, we next show that the probability that Procedure~\ref{proc:two} alters the original distribution is small which is crucial to bound the total variation distance between the original and the final distribution. 

First, we observe that the order in which Procedure~\ref{proc:two} resolves correlation does not affect its termination. Therefore, we first analyze the total probability of deviation from the original distribution while resolving the positive correlations. For simplicity, we assume that for the first $M\cdot N-1$ many iterations of Procedure~\ref{proc:two}, we resolve correlation between $X_{i,1}^{\vec u}$ with $X_{i,\ell}^{\vec u'}$ for all ${\vec u'}^{i,\ell'} \neq {\vec u}^{i,1}\in \BigM_{[m\cdot N]}$. We first bound $\sum_{k\leq MN-1} q_k$. Before, we obtain the bound on $\sum_{k\leq MN-1} q_k$, we prove the following crucial claims:
\begin{claim}\label{claim:prob_evol}
    For $k\leq  MN - 1$, we have $$\E[\tilde X_{i,1 }^{\vec u}(k)]  = \prod_{p<k}(1 - q_{p})  \cdot \E[\tilde X_{i,1 }^{\vec u}(0)] +\frac{1}{2} \cdot \sum_{p\leq k} q_p\cdot  \left( \prod_{k \geq j>p}\left (1 - \frac{q_{j}}{2} \right) \right).$$
\end{claim}
\begin{proof}
    Since, fir $k<MN-1$, we resolves the correlation between $X_{i,1}^{\vec u}$ with $X_{i,\ell}^{\vec u'}$, at each iteration $p<k$, we assign  $X_{i,1}^{\vec u}(p) =1$ independently at each round with probability $q_{p}/2$. Finally, we have $\tilde X_{i,1 }^{\vec u}(k) = 1$ iff we have $\tilde X_{i,1}^{\vec u}(p) =1$ and Procedure~\ref{proc:two} does not assign $\tilde X_{i,1}^{\vec u}(j)$ for all $j>p$. Combining the above argument for the set of disjoint events $\tilde X_{i,1}^{\vec u}(p-1) =0$ and $\tilde X_{i,1}^{\vec u}(j) =1$ for all $j\geq p$, we obtain the proof of the claim.
\end{proof}
Next, we prove the monotonicity of the probability $q_k$ for $k = 1,\dots , MN-1$.
\begin{claim}\label{claim:monotonicity_proc2}
    For any $k< MN-1$, we have $q_{k} > q_{k+1}$. 
\end{claim}
\begin{proof}
    This claim simply follows from the fact that $\E[\tilde X_{i,1 }^{\vec u}(k+1)] \geq \E[\tilde X_{i,1 }^{\vec u}(k)]$,  $\E [\tilde  X^{\vec u}_{i,1}(k +1)\cdot    \tilde X^{\vec u'}_{i,\ell}(k+1)] -\E [ \tilde X^{\vec u}_{i,1}( k+1) ]\cdot  \E[ \tilde X^{\vec u'}_{i,\ell}( k +1 )] = (1 - q_k)\cdot( \E [\tilde  X^{\vec u}_{i,1}(k)\cdot    \tilde X^{\vec u'}_{i,\ell}(k)] -\E [ \tilde X^{\vec u}_{1,\ell}( k) ]\cdot  \E[ \tilde X^{\vec u'}_{i,\ell}( k )])$ and $q_{k+1} \leq \frac{\E [\tilde  X^{\vec u}_{i,1}(k +1)\cdot    \tilde X^{\vec u'}_{i,\ell}(k+1)] -\E [ \tilde X^{\vec u}_{i,1}( k+1) ]\cdot  \E[ \tilde X^{\vec u'}_{i,\ell}( k +1 )]}{\E[\tilde X_{i,1}^{\vec u}(k)]}$ due to Lemma~\ref{claim:bound_perturb_prob} and Claim~\ref{claim:monotonicity_correlation}.
\end{proof}
Finally, we are now ready to bound the the probability $q_k$. 
\begin{lemma}\label{lem:dev_bound_pos_proc2}
    For $k\leq MN-1$, we have $$q_k \leq \frac{\E [\tilde  X^{\vec u}_{i,1}(0)\cdot    \tilde X^{\vec u'}_{i,\ell}(0)] -\E [ \tilde X^{\vec u}_{i,1}( 0) ]\cdot  \E[ \tilde X^{\vec u'}_{i,\ell}( 0 )]}{ \E[\tilde X_{i,1}^{\vec u} (0)] + \frac{1}{4} \cdot \sum_{p\leq k}  \frac{\varepsilon}{ p_0 +\frac{p\cdot \varepsilon}{2\cdot \E[ \tilde X^{\vec u'}_{i,\ell}( 0 )]}}}. $$
\end{lemma}
\begin{proof}
    First, we observe that for the first $MN-1$ iterations of Procedure~\ref{proc:two}, we resolve correlation between $\tilde X_{i,1}^{\vec u}$ and $\tilde X_{i,\ell}^{\vec u'}$ therefore due to Claim~\ref{claim:monotonicity_correlation}, at iteration $k$ when we resolve correlation between $\tilde X_{i,1}^{\vec u}$ and $\tilde X_{i,\ell}^{\vec u'}$  for some $\ell \in [N]$ and $\vec u' \in \BigM$, we have,
    \begin{equation*}
        \E [\tilde  X^{\vec u}_{i,1}(k )\cdot    \tilde X^{\vec u'}_{i,\ell}(k)] -\E [ \tilde X^{\vec u}_{i,1}( k) ]\cdot  \E[ \tilde X^{\vec u'}_{i,\ell}( k )] = \left( \prod_{p < k} (1 - q_p)\right) \cdot \varepsilon,
    \end{equation*}
    for $\varepsilon = \E [\tilde  X^{\vec u}_{i,1}(0)\cdot    \tilde X^{\vec u'}_{i,\ell}(0)] -\E [ \tilde X^{\vec u}_{i,1}( 0) ]\cdot  \E[ \tilde X^{\vec u'}_{i,\ell}( 0 )]$. For the sake of simplicity, we define $p_0 := \E[\tilde X_{i,1 }^{\vec u}(0)]$. Therefore, we can bound, 
    $$\begin{aligned}
        q_k &\leq \frac{\left( \prod_{p < k} (1 - q_p)\right) \cdot \varepsilon}{\E[X_{i,1}^{\vec u} (k)]} = \frac{\left( \prod_{p < k} (1 - q_p)\right) \cdot \varepsilon}{\prod_{p<k}(1 - q_{p})  \cdot p_0 +\frac{1}{2} \cdot \sum_{p\leq k} q_p\cdot  \left( \prod_{k \geq j>p}\left (1 - \frac{q_{j}}{2} \right) \right)} \\
        &= \frac{\varepsilon}{  p_0 +\frac{1}{2} \cdot \frac{\sum_{p\leq k} q_p\cdot  \left( \prod_{k \geq j>p}\left (1 - \frac{q_{j}}{2} \right) \right)}{\prod_{p<k}(1 - q_{p})}}\\
        &\leq \frac{\varepsilon}{  p_0 +\frac{1}{2} \cdot \frac{\sum_{p\leq k} q_p\cdot  \left( \prod_{k \geq j>p}\left (1 - q_{j} \right) \right)}{\prod_{p<k}(1 - q_{p})}}
    \end{aligned}$$
    Above, the first inequality holds due to Claim~\ref{claim:bound_perturb_prob}, the first equality holds due to Claim~\ref{claim:prob_evol}, the second equality holds by simple re-arrangement and the last inequality follows because $1- \frac{q_j}{2} \geq 1- q_j$. We now focus on lower-bounding the term $\frac{\sum_{p\leq k} q_p\cdot  \left( \prod_{k \geq j>p}\left (1 - q_{j} \right) \right)}{\prod_{p<k}(1 - q_{p})}$. We now again expand the the expression,
    $$\begin{aligned}
        \frac{\sum_{p\leq k} q_p\cdot  \left( \prod_{k \geq j>p}\left (1 - q_{j} \right) \right)}{\prod_{p<k}(1 - q_{p})} &\geq \frac{1}{2}\cdot \frac{\sum_{p\leq k}  \frac{\varepsilon\cdot \left( \prod_{p < k}\left (1 - q_{j} \right) \right)}{\prod_{j<p}(1 - q_{p})  \cdot p_0 +\frac{1}{2} \cdot \sum_{j\leq p} q_j\cdot  \left( \prod_{p \geq j'>j}\left (1 - \frac{q_{j'}}{2} \right) \right)} }{\left( \prod_{p < k}\left (1 - q_{j} \right) \right)}\\
        &=\frac{1}{2}\cdot \sum_{p\leq k}  \frac{\varepsilon}{\prod_{j<p}(1 - q_{p})  \cdot p_0 +\frac{1}{2} \cdot \sum_{j\leq p} q_j\cdot  \left( \prod_{p \geq j'>j}\left (1 - \frac{q_{j'}}{2} \right) \right)}\\
        &\geq \frac{1}{2}\cdot \sum_{p\leq k}  \frac{\varepsilon}{ p_0 +\frac{1}{2} \cdot \sum_{j\leq p} q_j}\\
        &\geq \frac{1}{2}\cdot \sum_{p\leq k}  \frac{\varepsilon}{ p_0 +\frac{p\cdot \varepsilon}{2\cdot p_0}}.
    \end{aligned}$$
    Above, the first inequality holds because of Claim~\ref{claim:bound_perturb_prob}, the second equality holds because $\left (1 - q_{j} \right) <1 $. The last inequality holds due to Claim~\ref{claim:monotonicity_proc2}, i.e. $q_j \leq q_1$ and due to Claim~\ref{claim:bound_perturb_prob} $q_1 \leq \frac{\eps}{p_0}$. Combining this with the earlier inequality, we obtain,
    \begin{equation*}
        q_k \leq \frac{\varepsilon}{  p_0 +\frac{1}{4} \cdot \sum_{p\leq k}  \frac{\varepsilon}{ p_0 +\frac{p\cdot \varepsilon}{2\cdot p_0}}} 
    \end{equation*}
    This completes the proof. 
\end{proof}
Next, we analyze the bound on $q_k$ obained in the previous lemma. First, wre observe that,
\begin{claim}\label{obs:monotonicity_of_bound}
    For $\varepsilon\in (0,1)$,  $p_{0} \in (0,1)$ and $\frac{\eps}{p_0^2} >1$, function $h(\eps, p_0) = \frac{\varepsilon}{  p_0 +\frac{1}{4} \cdot \sum_{p\leq k}  \frac{\varepsilon}{ p_0 +\frac{p\cdot \varepsilon}{2\cdot p_0}}}$ is increasing in $\varepsilon$ and decreasing in $p_0$. 
\end{claim}
\begin{proof}
    First we prove the function's monotonicity in $\varepsilon$. Absolutely, let's analyze the given function and prove that it's increasing with respect to $\varepsilon$.
We let $D(\varepsilon) = p_0 +\frac{1}{4} \cdot \sum_{p\leq k} \frac{\varepsilon}{ p_0 +\frac{p\cdot \varepsilon}{2\cdot p_0}}.$ We need to show that, $ \frac{D(\varepsilon)  - \varepsilon \cdot D'(\varepsilon)}{[D(\varepsilon)]^2}\geq 0$ which is eqvivelenet to proving that $D(\varepsilon) - \varepsilon \cdot D'(\varepsilon) > 0$. We observe that,
$D'(\varepsilon) = \frac{1}{4} \cdot \sum_{p\leq k} \frac{p_0}{[p_0 + \frac{p\cdot \varepsilon}{2\cdot p_0}]^2}$. This implies that,
$$p_0 +\frac{1}{4} \cdot \sum_{p\leq k} \frac{\varepsilon}{ p_0 +\frac{p\cdot \varepsilon}{2\cdot p_0}} - \varepsilon \cdot \frac{1}{4} \cdot \sum_{p\leq k} \frac{p_0}{[p_0 + \frac{p\cdot \varepsilon}{2\cdot p_0}]^2} =p_0 + \frac{1}{4} \sum_{p\leq k} \left(\frac{p\varepsilon^2}{2p_0(p_0 + \frac{p\varepsilon}{2p_0})^2}\right) > 0. $$
Next, we analyze its monotonicity property in $p_0$. 

Let's analyze the function with respect to $p_0$. We let, $D(p_0) = p_0 + \frac{1}{4} \sum_{p \le k} \frac{\varepsilon}{p_0 + \frac{p\varepsilon}{2p_0}}$. Therefore, to complete the claim, we need to show that $D'(p_0) >0$. We observe that $D'(p_0) =1 - \frac{\varepsilon}{4} \sum_{p \le k} \frac{1 - \frac{p\varepsilon}{2p_0^2}}{\left( p_0 + \frac{p\varepsilon}{2p_0} \right)^2} $. Since, $\frac{\eps}{p_0^2}>1$ (due to positive correlation at the beginning of Procedure~\ref{proc:two}), we get $$D'(p_0)>1 - \frac{\varepsilon}{4} \sum_{p \le k} \frac{1 - \frac{p}{2}}{\left( p_0 + \frac{p\varepsilon}{2p_0} \right)^2}   > 1.$$ Above the last inequality follows because $k >2$.
\end{proof}

Finally, due to Claim~\ref{obs:correlation_proc1}, we get the following lemma.
\begin{lemma}\label{lem:proc2_pos}
    Let Procedure~\ref{proc:two} resolve positive correlations for the first $k^* \leq m\cdot M^2\cdot N$many iterations, then we have $$\sum_{k\leq k^*}q_k \leq O \left( \frac{m\cdot M^2}{\log (M\cdot N)} \right) \leq O\left(  \frac{m}{M}\right).$$ 
\end{lemma}
\begin{proof}
    First, we let the procedure run for each $\vec u\in \BigM$ and resolve positive correlations for all pairs  $X_{i,1}^{\vec u}$ with $X_{i,\ell}^{\vec u'}$ for all ${\vec u'}^{i,\ell'} \neq {\vec u}^{i,1}\in \BigM_{[m\cdot N]}$. Since the correlation between the pair of random variables goes down after each iteration, we can bound the total deviation  $\sum_{k\leq k^*}q_k \leq M \cdot \sum_{k=1}^{MN-1} q_k$. In other words, we need to bound the total deviation while resolving the correlation between for all ${\vec u'}^{i,\ell'} \neq {\vec u}^{i,1}\in \BigM_{[m\cdot N]}$ for the first considered element $\vec u\in \BigM$. We can bound, 
    $$\begin{aligned}
        \sum_{k=1}^{MN-1} q_k &\leq \sum_{k=1}^{MN-1} \frac{\varepsilon}{  p_0 +\frac{1}{4} \cdot \sum_{p\leq k}  \frac{\varepsilon}{ p_0 +\frac{p\cdot \varepsilon}{2\cdot p_0}}}\\
        & \leq \sum_{k=1}^{MN-1} \frac{\frac{m}{N^2 (M-1)}}{\frac{1}{MN} + \frac{1}{4} \cdot \sum_{p\leq k}\frac{\frac{m}{N^2 (M-1)}}{\frac 1 {MN} + p\cdot \frac{m}{2N}}}\\
        &\leq \sum_{k=1}^{MN-1} \frac{\frac{m}{N^2 (M-1)}}{\frac{1}{MN} + \frac{1}{2MN} \cdot \sum_{p\leq k}\frac{1}{p+1}}\\
        &\leq \frac{2m}{N} \cdot \sum_{k=1}^{MN-1} \frac{1}{1 + \log k} \leq O \left( \frac{m\cdot M}{\log(MN)} \right). 
    \end{aligned}$$
    Above, the first inequality holds due to Claim~\ref{lem:dev_bound_pos_proc2}. The second inequality holds due to Claim~\ref{obs:monotonicity_of_bound} and $p_0 \geq \frac{1}{M\cdot N}$ and $\varepsilon \leq \frac{m}{(M-1)\cdot N}$. The third inequality holds because $\sum_{p\leq k} \frac{1}{p+1}\geq \frac{1}{\log k}$ and the final inequality holds because $\sum_{k=1}^{MN} \frac{1}{\log k}$ is the order of $O(MN /\log (MN))$. Since $N\geq \Omega\left ( 2^{M^3}\right)$, we conclude the lemma. 
\end{proof}

Next, we analyze the total probability of deviation from the original distribution while resolving the negative correlations. This case is simpler because the correlation resolution requires to alter the distribution with a small probability.
\begin{lemma}\label{lem:proc2_neg}
     For any $k\geq 0$, $i\in [m]$, $\vec u\in \BigM$ and $\ell \in [N]$, we have $\E[\tilde X_{i,\ell}(k)]\geq \E[\tilde X_{i,\ell}(0)]$ at the end of the $k$-th iteration of Procedure~\ref{proc:two}. In addition, at iteration $k$, for distinct $i, j \in [m]$, $\vec u\in \BigM$ and $\ell, \ell' \in [N]$ if correlation between $\tilde X_{i,\ell}^{\vec u}(k)$ and $\tilde X_{j,\ell}^{\vec u}(k)$ being resolved then, $q_k \leq \frac{M\cdot m}{(M-1)^3 \cdot N^2}. $ 
\end{lemma}
\begin{proof}
    We first observe that $\E[\tilde X_{i,\ell}(k)] \leq \frac 1 {2}$ for large enough $N, M >0$ which follows inductively. This implies that for any $k>0$, we have $\E[\tilde X_{i,\ell}(k+1)]= \E[\tilde X_{i,\ell}(k)] $ if $\tilde X_{i,\ell}(k)$ is not processed at the iteration $k$ of Procedure~\ref{proc:one}. Otherwise, we have $ \E[\tilde X_{i,\ell}(k+1)]\geq  (1-q_{k})\cdot  \E[\tilde X_{i,\ell}(k)] + q_{ij}/2 \geq \E[\tilde X_{i,\ell}(k)] $. This concludes the proof of the first part. Next, Lemma~\ref{claim:bound_perturb_prob} implies that,
    $$\begin{aligned}
         q_{k} &\leq \frac{\E [ \tilde X^{\vec u}_{i,\ell}( k ) ]\cdot  \E[ \tilde X^{\vec u}_{j,\ell'}(k )] - \E [ \tilde X^{\vec u}_{i,\ell}(k )\cdot    X^{\vec u}_{j,\ell'}(k )]}{(1 - \E[ X^{\vec u}_{i,\ell}(k )] - \E[ X^{\vec u}_{i,\ell}( k )])}\\
         &\leq \frac{\E [ X^{\vec u}_{i,\ell}(0 ) ]\cdot  \E[ X^{\vec u}_{j,\ell'}(0)] - \E [ X^{\vec u}_{i,\ell}(0 )\cdot    X^{\vec u}_{j,\ell'}(0 )]}{(1 - \E[ X^{\vec u}_{i,\ell}(k)] - \E[ X^{\vec u}_{i,\ell'}(k )])}\\
         &\leq \frac{\frac{m}{(M-1)^2\cdot N^2}}{1-p_i(k) - p_j(k)} \leq \frac{M\cdot m}{(M-1)^3 \cdot N^2}.
    \end{aligned} $$

Above, the first inequality follows from Claim~\ref{claim:bound_perturb_prob}, the second inequality follows because of Claim~\ref{claim:monotonicity_correlation}.  
\end{proof}

\subsection{Wrapping Things Up}\label{sec:wrap-up}

Finally, we show that the resultant random variables at the end of the procedure are close to the original random variables in terms of total variation distance and complete the proof of Theorem~\ref{thm:tv_distance}.

\begin{proof}[Proof of Theorem~\ref{thm:tv_distance}]
  Let $\tilde{\vec X}:= \tilde X_{i,\ell}^{\vec u}: \forall \vec u^{i,\ell}\in \BigM_{[m\cdot N]}$ be the random variables after applying Procedures~\ref{proc:one} and~\ref{proc:two} on the set of random variables $\vec X:= X_{i,\ell}^{\vec u}: \forall \vec u^{i,\ell}\in \BigM_{[m\cdot N]}$. Due to Lemma~\ref{lem:pairwise_ind}, we have that $\tilde X_{i,\ell}^{\vec u}: \forall \vec u^{i,\ell}\in \BigM_{[m\cdot N]}$ are pairwise independent. Finally, we can bound the total variation distance between $\vec X$ and $\tilde {\vec X}$. 
  
  Let $k_1$ and $k_2$ be the number of iteration performed by Procedures~\ref{proc:one} and~\ref{proc:two}, respectively. We divide the iterations of Procedure~\ref{proc:two} into two parts, during the first $\bar k_{2}$ many iterations, Procedure~\ref{proc:two} resolves positive correlation between $X^{\vec u}_{i,\ell}$ and $X^{\vec u'}_{i,\ell'}$ for $\vec u, \vec u' \in \BigM$, $i\in [m]$ and during the last $k_2 - \bar k_2$ many iterations, it resolves negative correlation between $X^{\vec u}_{i,\ell}, X^{\vec u}_{j, \ell'}$ for $\vec u \in \M$, $i, j \in [m]$ and $\ell, \ell' \in [N]$. 
  $$\begin{aligned}
      \operatorname{TV}_{\vec X, \tilde {\vec X}} &\leq \sum_{i\neq j\in [m]} q_{ij} + \sum_{k=1}^{k_1} q_k + \sum_{k=k_1+1}^{\bar k_2} q_k + \sum_{k=\bar k_2+1}^{k_2} q_k\\
      &\leq \frac{m^2}{(M-1)\cdot N} + \frac{m}{M^2} \cdot M + |k_2 - \bar k_2|\cdot \frac{M\cdot m}{(M-1)^3\cdot N^2}\\
      &\leq \frac{2m^3}{M}.
  \end{aligned}$$
The first inequality follows because it bounds the total probability that either Procedure~\ref{proc:one} or Procedure~\ref{proc:two} alters the original distribution $\vec X$. The second inequality follows because of Lemmas~\ref{lem:proc2_neg} and \ref{lem:proc2_pos}. The last inequality follows because $(k_2 - k_1)\leq M\cdot N^2\cdot m^2$ combining with $N= \Omega(2^{M^2})$ and $M = \Omega(2^m)$. 
\end{proof}

\end{document}